\documentclass{LMCS}

\def\dOi{9(3:12)2013}
\lmcsheading%
{\dOi}
{1--47}
{}
{}
{Feb.~15, 2013}
{Sep.~\phantom09, 2013}
{}

\usepackage[noadjust]{cite}
\usepackage{amsmath}
\usepackage{amsbsy}
\usepackage{amssymb}
\usepackage{stmaryrd}
\usepackage{mathtools}
\usepackage{wasysym}
\usepackage{proof}
\usepackage{xspace}
\usepackage{latexsym}
\usepackage{color}
\usepackage{enumerate,hyperref}

\newcommand{\mylambda}{\boldsymbol{\lambda}}
\newcommand{\pref}{\mathpunct{\downharpoonleft}}
\newcommand{\fdiff}{\mathop{{\backslash}\!\!{\backslash}}}

\newcommand{\id}{{\sf id}}
\newcommand{\Obj}{{\sf Obj}}

\newcommand{\thread}{{\sf tid}}
\newcommand{\ThreadID}{{\sf ThreadID}}

\newcommand{\mypar}[1]{\vspace{4pt}\noindent\textbf{#1}}

\newcommand{\fillin}[2]{#2({#1})}

\usepackage{longtable}

\usepackage{listings}
\newcommand{\baselstset}{\lstset{
    columns=fullflexible,
    basicstyle = \ttfamily,
    mathescape=true,
    escapechar=\#,
    literate={@@}{\#\color{black}\#}0
             {@}{\#\color{green}\#}0
             {\\<nil>}{nil}3
             {/\\}{$\land$}1
             {\\/}{$\lor$}1
             {<=}{$\leq\ $}1
             {<<}{$\ll$}1
             {>>}{$\gg$}1
             {=>}{$\Rightarrow$}1
             {>=}{$\geq\ $}1
             {->^z}{$\stackrel{z}{\to}$}1
             {|->}{$\mapsto$}1
             {_R}{$_R$}1
             {_2}{$_2$}1
             {_1}{$_1$}1
             {\\iter\{0\}\{i\}\{n-1\}}{$\iter{0}{i}{n-1}$}2
             {\\iter\{1\}\{i\}\{n\}}{$\iter{1}{i}{n}$}2
             {\\iter\{0\}\{i\}\{n\}}{$\iter{0}{i}{n}$}2
             {\\iter\{0\}\{i\}\{29\}}{$\iter{0}{i}{29}$}2,
    commentstyle=\itshape,
}%
}
\newcommand{\AuxStyle}{\color{blue}}
\newcommand{\mylstset}{%
\baselstset%
\lstset{
    language=C,
    morekeywords={cons,null,allocate,dispose,then,do,else,await,atomic},
    moredelim=[is][\AuxStyle]{/a }{ a/}
}}

\lstnewenvironment{source}{\mylstset}{}
\lstnewenvironment{sourcelineno}{\mylstset\lstset{numbers=left,numberstyle=\tiny,numbersep=5pt}}{}

\newcommand{\cover}{{\sf cover}}

\newcommand{\fdef}{\mathpunct{\downarrow}}
\newcommand{\fundef}{\mathpunct{\uparrow}}

\newcommand{\interf}{{\sf interf}}

\newcommand{\dom}{{\sf dom}}

\newcommand{\emp}{{\sf emp}}

\newcommand{\history}{{\sf history}}

\newcommand{\erase}{{\sf ground}}

\newcommand{\cH}{\mathcal{H}}
\newcommand{\power}[1]{2^{#1}}
\newcommand{\client}{{\sf client}}
\newcommand{\lib}{{\sf lib}}

\newcommand{\db}[1]{\llbracket #1 \rrbracket}
\newcommand{\dba}[1]{\langle #1 \rangle}
\newcommand{\dbf}[1]{\llbracket #1 \rrbracket^\sharp}
\newcommand{\dbp}[1]{\llparenthesis\hspace{1pt} #1 \hspace{1pt}\rrparenthesis}

\newcommand{\Trace}{{\sf Trace}}

\newcommand{\WTrace}{{\sf Trace}}

\newcommand{\BHistory}{{\sf BHistory}}

\newcommand{\Cc}{\mathcal{C}}

\newcommand{\cL}{\mathcal{L}}
\newcommand{\cI}{I}

\newcommand{\cO}{\mathcal{P}}

\newcommand{\cS}{\mathcal{S}}

\newcommand{\diff}{\mathop{\backslash}}

\newcommand{\RetAct}{{\sf RetAct}}
\newcommand{\CallAct}{{\sf CallAct}}
\newcommand{\CallRetAct}{{\sf CallRetAct}}
\newcommand{\ECallRetAct}{{\sf CallRetAct}}
\newcommand{\ERetAct}{{\sf RetAct}}
\newcommand{\ECallAct}{{\sf CallAct}}

\newcommand{\Act}{{\sf Act}}

\newcommand{\State}{\Sigma}
\newcommand{\Perm}{{\sf Perm}}
\newcommand{\Foot}{\mathcal{F}}
\newcommand{\RAM}{{\sf RAM}}

\newcommand{\PComm}{{\sf PComm}}

\newcommand{\prefix}{{\sf prefix}}
\newcommand{\cmgc}{{\sf mgc}}

\newcommand{\assume}{{\sf assume}}

\newcommand{\call}{{\sf call}}
\newcommand{\ret}{{\sf ret}}
\newcommand{\myarg}{{\sf arg}}

\newcommand{\MethodName}{{\sf Method}}

\newcommand{\Val}{{\sf Val}}

\newcommand{\Loc}{{\sf Loc}}

\def\be{\begin{equation}}
\def\ee{\end{equation}}
\def\ba{\begin{equation*}\begin{array}{@{}c@{}}}
\def\ea{\end{array}\end{equation*}}

\newcommand{\commentout}[1]{}

\title[Linearizability with Ownership Transfer]{Linearizability with Ownership Transfer}
\author[A.~Gotsman]{Alexey Gotsman\rsuper a}
\address{{\lsuper a}IMDEA Software Institute, Madrid, Spain}
\email{Alexey.Gotsman@imdea.org}
\author[H.~Yang]{Hongseok Yang\rsuper b}
\address{{\lsuper b}University of Oxford, Oxford, UK}
\email{Hongseok.Yang@cs.ox.ac.uk}


\ACMCCS{[{\bf Theory of computation}]: Semantics and reasoning}

\begin{document}

\begin{abstract}
  Linearizability is a commonly accepted notion of correctness for libraries of
  concurrent algorithms. Unfortunately, it assumes a complete isolation between
  a library and its client, with interactions limited to passing values of a
  given data type. This is inappropriate for common programming languages, where
  libraries and their clients can communicate via the heap, transferring the
  ownership of data structures, and can even run in a shared address space
  without any memory protection.

  In this paper, we present the first definition of linearizability that lifts
  this limitation and establish an Abstraction Theorem: while proving a property
  of a client of a concurrent library, we can soundly replace the library by its
  abstract implementation related to the original one by our generalisation of
  linearizability. This allows abstracting from the details of the library
  implementation while reasoning about the client. We also prove that
  linearizability with ownership transfer can be derived from the classical one
  if the library does not access some of data structures transferred to it by
  the client.
\end{abstract}

\maketitle

\section{Introduction\label{sec:intro}}

The architecture of concurrent software usually exhibits some forms of
modularity.  For example, concurrent algorithms are encapsulated in libraries
and complex algorithms are often constructed using libraries of simpler
ones. This lets developers benefit from ready-made libraries of concurrency
patterns and high-performance concurrent data structures, such as
$\mathsf{java.util.concurrent}$ for Java and Threading Building Blocks for C++.
To simplify reasoning about concurrent software, we need to exploit the
available modularity. In particular, in reasoning about a client of a concurrent
library, we would like to abstract from the details of a particular library
implementation. This requires an appropriate notion of library correctness.

Correctness of concurrent libraries is commonly formalised by {\em
  linearizability}~\cite{linearizability}, which fixes a certain correspondence
between the library and its specification. The latter is usually just another
library, but implemented atomically using an abstract data type; the two
libraries are called {\em concrete} and {\em abstract}, respectively. A good
notion of linearizability should validate an {\em Abstraction
  Theorem}~\cite{tcs10,icalp}: the behaviours of any client using the concrete
library are contained in the behaviours of the client using the abstract
one. This makes it sound to replace a library by its specification in
reasoning about its client.

Classical linearizability assumes a complete isolation between a library and
its client, with interactions limited to passing values of a given data type as
parameters or return values of library methods. This notion is not appropriate
for low-level heap-manipulating languages, such as C/C++. There the library and
the client run in a shared address space; thus, to prove the whole program
correct, we need to verify that one of them does not corrupt the data structures
used by the other. Type systems~\cite{clarke01} and
program logics~\cite{seplogic-concurrent} usually establish this using the
concept of {\em ownership} of data structures by a program component:
the right to access a data structure is given only to 
a particular component or a set of them. 
When verifying realistic programs, this ownership of data structures cannot be assigned
statically; rather, it should be {\em transferred}
between the client and the library at calls to and returns from the latter.
The times when ownership is transferred are not determined
operationally, but set by the proof method: as O'Hearn famously put it,
``ownership is in the eye of the asserter''~\cite{seplogic-concurrent}. However,
ownership transfer reflects actual interactions between program components via
the heap, e.g., alternating accesses to a shared area of memory.  Such
interactions also exist in high-level languages providing basic memory
protection, such as Java. 
In this case, we need to ensure that a client does not
subvert a library by accessing a memory object after its ownership was
transferred to the latter.

For an example of ownership transfer between concurrent libraries and their
clients consider a memory allocator accessible concurrently to multiple
threads. We can think of the allocator as owning the blocks of memory on its
free-list; in particular, it can store free-list pointers in them. Having
allocated a block, a thread gets its exclusive ownership, which allows accessing
it without interference from the other threads.  When the thread frees the
block, its ownership is returned to the allocator.  Trying to write to a memory
cell after it was freed has dire consequences.

As another example, consider any container with concurrent access, such as a
concurrent set from $\mathsf{java.util.concurrent}$ or Threading Building
Blocks.  A typical use of such a container is to store pointers to a certain
type of data structures. However, when verifying a client of the container, we
usually think of the latter as holding the ownership of the data structures
whose addresses it stores~\cite{seplogic-concurrent}. Thus, when a thread
inserts a pointer to a data structure into a container, its ownership is
transferred from the thread to the container. When another thread removes a
pointer from the container, it acquires the ownership of the data structure the
pointer identifies. If the first thread tries to access a data structure after a
pointer to it has been inserted into the container, this may result in a race
condition. Unlike a memory allocator, the container code usually does not access
the contents of the data structures its elements identify, but merely ferries
their ownership between different threads. For this reason, correctness proofs
for such
containers~\cite{daphna-linearizability,turkish10,viktor-linearizability} have
so far established their classical linearizability, without taking ownership
transfer into account.

We would like to use the notion of linearizability and, in particular, an
Abstraction Theorem to reason about the above libraries and their clients in
isolation, taking into account only the memory that they own. To this end, we
would like the correctness of a library to constrain not only pointers that are
passed between it and the client, but also the contents of the data structures
whose ownership is transferred. So far, there has been no notion of
linearizability that would allow this. In the case of concurrent containers, we
have no way of using classical linearizability established for them to validate
an Abstraction Theorem that would be applicable to clients performing ownership
transfer. This paper fills in these gaps.

\mypar{Contributions.} In this paper, we generalise linearizability to a setting
where a library and its client execute in a shared address space, and boundaries
between their data structures can change via ownership
transfers. Linearizability is usually defined in terms of {\em histories}, which
are sequences of calls to and returns from a library in a given program
execution, recording parameters and return values passed. To handle ownership
transfer, histories also have to include descriptions of memory areas
transferred. However, in this case, some histories cannot be generated by any
pair of a client and a library: while generating histories of a library we
should only consider its executions in an environment that respects ownership.
For example, a client that transfers an area of memory upon a call to a library
not communicating with anyone else cannot then transfer the same area again
before getting it back from the library upon a method return. We propose a
notion of {\em balancedness} that characterises those histories that treat
ownership transfer correctly. We then define a {\em linearizability relation}
between balanced histories, matching histories of a concrete and an abstract
library (Section~\ref{sec:observ}).

This definition does not rely on a particular model of program states for
describing memory areas transferred in histories, but assumes an arbitrary model
defined by a {\em separation algebra}~\cite{asl}.  By picking a model with
so-called {\em permissions}~\cite{permissions,dg}, we can allow clients and
libraries to transfer non-exclusive rights to access certain memory areas in
particular ways, instead of transferring their full ownership. This makes our
results applicable even when libraries and their clients share access to some
areas of memory. Our definition of balancedness for arbitrary separation
algebras relies on a new formalisation of the notion of a {\em footprint} of a
state, describing the amount of permissions the state includes
(Section~\ref{sec:prelim}).

The rest of our technical development relies on a novel compositional semantics
for a language with libraries that defines the denotation of a library or a
client considered separately in an environment that communicates with the
component correctly via ownership transfers (Sections~\ref{sec:prog}
and~\ref{sec:semantics}). In particular, the semantics allows us to generate the
set of all histories that can be produced by a library solely from its code,
without considering all its possible clients. This, in its turn, allows us to
lift the linearizability on histories to libraries and establish the Abstraction
Theorem (Section~\ref{sec:refinement}). On the way, we also obtain an insight
into the original definition of linearizability without ownership transfer,
showing a surprising relationship between one of the ways of its formulation and
the plain subset inclusion on the sets of histories produced by concrete and
abstract libraries.

We note that the need to consider ownership transfer makes the proof of the
Abstraction Theorem highly non-trivial. This is because proving the theorem
requires us to convert a computation with a history produced by the concrete
library into a computation with a history produced by the abstract one, which
requires moving calls and returns to different points in the computation. In the
setting without ownership transfer, these actions are thread-local and can be
moved easily; however, once they involve ownership transfer, they become global
and the justification of their moves becomes subtle, in particular, relying on
the fact that the histories involved are balanced
(Section~\ref{sec:rearr-proof}).

To avoid having to prove the new notion of linearizability from scratch for
libraries that do not access some of the data structures transferred to them,
such as concurrent containers, we propose a {\em frame rule for linearizability}
(Section~\ref{sec:frame}). It ensures the linearizability of such libraries with
respect to a specification with ownership transfer given their linearizability
with respect to a specification without one.

%
%

We provide a glossary of notation at the end of the paper.

\section{Footprints of States\label{sec:prelim}}

\subsection{Separation Algebras\label{sec:sa}}

Our results hold for a class of models of program states called {\em separation
  algebras}~\cite{asl}, which allow expressing the dynamic memory partitioning
between libraries and clients.
\begin{defi}
  A \textbf{\em separation algebra} is a set $\Sigma$, together with a partial
  commutative, associative and cancellative operation $*$ on $\Sigma$
  and a unit element $e \in \Sigma$.
  Here commutativity and associativity hold for the equality that means both
  sides are defined and equal, or both are undefined.  The property of
  cancellativity says that for each $\sigma \in \Sigma$, the function $\sigma *
  \cdot : \Sigma \rightharpoonup \Sigma$ is injective.
\end{defi}
We think of elements of a separation algebra $\Sigma$ as {\em portions} of
program states and the $*$ operation as combining such portions. The partial
states allow us to describe parts of the program state belonging to a library or
the client. When the $*$-combination of two states is defined, we call them {\bf
  \em compatible}.  Incompatible states usually make contradictory claims about
the ownership of memory.  We sometimes use a pointwise lifting $*:
\power{\Sigma} \times \power{\Sigma} \rightarrow \power{\Sigma}$ of $*$ to sets
of states: for $p, q \in \power{\Sigma}$ we let
$
p * q = \{\sigma_1 * \sigma_2 \mid \sigma_1 \in p \wedge \sigma_2 \in q\}
$.

Elements of separation algebras are often defined using partial
functions. We use the following notation:
$g(x)\fdef$ means that the function $g$ is defined on $x$,
$g(x)\fundef$ means that it is undefined on $x$,
$\dom(g)$ denotes the set of arguments on which $g$ is defined,
$[\,]$ denotes a nowhere-defined function,
and $g[x:y]$ denotes the function that has the same value as $g$ everywhere,
except for $x$, where it has the value $y$.  We also write $\_$
for an expression whose value is irrelevant and implicitly existentially
quantified.

Below is an example separation algebra $\RAM$:
$$
\Loc  = \{1, 2, \ldots\};
\qquad \qquad
\Val = \mathbb{Z};
\qquad \qquad
\RAM 
= \Loc \rightharpoonup_{\it fin} \Val.
$$
A (partial) state in this model consists of a finite partial function from
allocated memory locations to the values they store. The $*$ operation on $\RAM$
is defined as the disjoint function union $\uplus$, with the
nowhere-defined function $[\,]$ as its unit. Thus, $*$ combines disjoint
pieces of memory.

More complicated separation algebras do not split memory completely, instead
allowing heap parts combined by $*$ to overlap. This is done by associating
so-called {\em permissions}~\cite{permissions} with memory cells in the model,
which do not give their exclusive ownership, but allow accessing them in a
certain way. Types of permissions range from read sharing~\cite{permissions} to
accessing memory in an arbitrary way consistent with a given
specification~\cite{dg}.  We now give an example of a separation algebra with
permissions of the former kind. 

We define the algebra $\RAM_\pi$ as follows:
$$
\Loc = \{1, 2, \ldots\};
\qquad
\Val = \mathbb{Z};
\qquad
\Perm = (0,1];
\qquad
\RAM_\pi
= \Loc \rightharpoonup_{\it fin} (\Val \times \Perm).
$$
A state in this model consists of a finite partial function from allocated
memory locations to values they store and so-called {\em permissions}---numbers
from $(0,1]$ that show ``how much'' of the memory cell belongs to the partial
state~\cite{permissions}. The latter allow a library and its client to share
access to some of memory cells.  Permissions in $\RAM_\pi$ allow only read
sharing: when defining the semantics of commands over states in $\RAM_\pi$, the
permissions strictly less than $1$ are interpreted as permissions to read; the
full permission $1$ additionally allows writing.
The $*$ operation on $\RAM_\pi$ adds up permissions for memory cells. Formally,
for $\sigma_1, \sigma_2 \in \RAM_\pi$, we write
$\sigma_1 \mathop{\sharp} \sigma_2$ if:
$$
 \forall
x\in\Loc.\, 
{\sigma_1(x)\fdef} \wedge {\sigma_2(x)\fdef} \implies 
(\exists u, \pi_1, \pi_2.\,\sigma_1(x) = (u,\pi_1) \wedge \sigma_2(x) =
(u,\pi_2) \wedge \pi_1 + \pi_2 \leq 1). 
$$
If $\sigma_1 \mathop{\sharp} \sigma_2$, then we define
\begin{multline*}
\sigma_1 * \sigma_2 = \{(x, (u, \pi)) \mid 
(\sigma_1(x) = (u, \pi) \wedge {\sigma_2(x)\fundef})
\vee {}
(\sigma_2(x) = (u, \pi) \wedge {\sigma_1(x)\fundef})
\vee {} \\
(\sigma_1(x) = (u, \pi_1) \wedge \sigma_2(x) = (u, \pi_2) \wedge \pi = \pi_1+\pi_2)\};
\end{multline*}
otherwise, $\sigma_1 * \sigma_2$ is undefined. 
The unit for $*$ is  the empty heap~$[\,]$. 
This definition of $*$ allows us, e.g., to split a memory area into two
disjoint parts. It also allows splitting a cell with a full permission $1$ into
two parts, carrying read-only permissions $1/2$ and agreeing on the value stored
in the cell. These permissions can later be recombined to obtain the full
permission, which allows both reading from and writing to the cell.

Since we develop all our results for an arbitrary separation algebra, by
instantiating it with algebras similar to $\RAM_\pi$, we can handle cases when a
library and its client share access to some memory areas.

Consider an arbitrary separation algebra $\Sigma$ with an operation $*$.
We define a partial operation $\diff : \Sigma \times \Sigma \rightharpoonup
\Sigma$, called {\bf \em state subtraction}, as follows: $\sigma_2 \diff
\sigma_1$ is a state in $\Sigma$ such that $\sigma_2 = (\sigma_2 \diff \sigma_1)
* \sigma_1$; if such a state does not exist, $\sigma_2 \diff \sigma_1$ is
undefined. The cancellativity of $*$ implies that $\sigma_2 \diff \sigma_1$ is
determined uniquely, and hence, the $\diff$ operation is well-defined.  When
reasoning about ownership transfer between a library and a client, we use the
$*$ operation to express a state change for the component that is receiving the
ownership of memory, and the $\diff$ operation for the one that is giving it
up.



\begin{prop}\label{prop-diff}
  For all $\sigma_1,\sigma_2,\sigma_3\in\Sigma$, if $\sigma_1 * \sigma_2$ and
  $\sigma_1\diff \sigma_3$ are defined, then
$$
(\sigma_1 * \sigma_2)\diff \sigma_3 = (\sigma_1 \diff \sigma_3) * \sigma_2.
$$
\end{prop}

\subsection{Footprints\label{sec:foot}}

Our definition of linearizability uses a novel formalisation of a
{\em footprint} of a state, which, informally, describes the amount of memory or
permissions the state includes.
\begin{defi}\label{defn-foot}
A {\bf \em footprint} of a state $\sigma$ in a separation algebra $\Sigma$ is
the set of states
$$\delta(\sigma)= \{\sigma' \mid \forall \sigma''.\, {(\sigma' * \sigma'')\fdef}
\iff {(\sigma * \sigma'')\fdef} \}.$$
\end{defi}
In the following, $l$ ranges over footprints. The function $\delta$ computes the
equivalence class of states with the same footprint as $\sigma$. In the case of
$\RAM$, we have $\delta(\sigma) = \{\sigma' \mid \dom(\sigma) = \dom(\sigma')\}$
for every $\sigma\in\RAM$. Thus, states with the same footprint contain the same
memory cells. Definitions of $\delta$ for separation algebras with permissions
are more complicated, taking into account not only memory cells present in the
state, but also permissions for them. In the case of the algebra $\RAM_\pi$, for
$\sigma\in\RAM_\pi$ we have
\begin{multline*}
\delta(\sigma) = 
\{\sigma' \mid \forall x.\, 
({\sigma(x)\fdef} \iff {\sigma'(x)\fdef}) \wedge {}
\\
\qquad
\forall u, \pi.\,
(\sigma(x)=(u,\pi) \wedge \pi < 1 {\implies} \sigma(x)=\sigma'(x)) \wedge
(\sigma(x)=(u,1) {\implies} \sigma'(x) = (\_, 1))
\}.
\end{multline*}
In other words, states with the same footprint contain the same memory cells
with the identical permissions; in the case of memory cells on read permissions,
the states also have to agree on their values.

Let $\Foot(\Sigma)=\{\delta(\sigma) \mid \sigma\in\Sigma\}$ be the set of
footprints in a separation algebra $\Sigma$. We now lift the $*$ and $\diff$
operations on $\Sigma$ to $\Foot(\Sigma)$. First, we define the operation $\circ
: \Foot(\Sigma) \times \Foot(\Sigma) \rightharpoonup \Foot(\Sigma)$ for adding
footprints.  Consider $l_1,l_2\in \Foot(\Sigma)$ and
$\sigma_1,\sigma_2\in\Sigma$ such that $l_1 = \delta(\sigma_1)$ and $l_2 =
\delta(\sigma_2)$. If $\sigma_1*\sigma_2$ is defined, we let $l_1 \circ l_2 =
\delta(\sigma_1*\sigma_2)$; otherwise $l_1\circ l_2$ is undefined. 
\begin{prop}
The $\circ$ operation is well-defined.
\end{prop}
\proof Consider $l_1,l_2\in \Foot(\Sigma)$ and
$\sigma_1,\sigma_2\in\State$ such that $l_1 = \delta(\sigma_1)$ and $l_2 =
\delta(\sigma_2)$. Take another pair of states $\sigma'_1,\sigma'_2\in\State$
such that $l_1 = \delta(\sigma'_1)$ and $l_2 = \delta(\sigma'_2)$.
Thus, $\sigma'_1 \in \delta(\sigma_1)$ and $\sigma'_2\in\delta(\sigma_2)$, which
implies:
$$
{{(\sigma'_1*\sigma'_2)}\fdef} \iff 
{{(\sigma_1*\sigma'_2)}\fdef} \iff
{{(\sigma_1*\sigma_2)}\fdef}.
$$
Furthermore, if $\sigma'_1*\sigma'_2$ and $\sigma_1*\sigma_2$ are defined, then
for all $\sigma'\in\State$, 
$$
{{(\sigma'_1*\sigma'_2*\sigma')}\fdef} \iff 
{{(\sigma_1*\sigma'_2*\sigma')}\fdef} \iff 
{{(\sigma_1*\sigma_2*\sigma')}\fdef}.
$$
Hence, $\delta(\sigma_1*\sigma_2) = \delta(\sigma'_1*\sigma'_2)$, so that 
$\circ$ is well-defined.\qed

\smallskip

For $\RAM$, $\circ$ is just the pointwise lifting of the disjoint function union $
\uplus$.

To define a subtraction operation on footprints, we use the following condition.
\begin{defi}\label{foot-canc}
The $*$ operation of a separation algebra $\Sigma$ is {\bf \em cancellative on
  footprints} when for all $\sigma_1,\sigma_2,\sigma'_1,\sigma_2' \in \Sigma$,
if $\sigma_1*\sigma_2$ and $\sigma'_1*\sigma'_2$ are defined, then
$$
{(\delta(\sigma_1 * \sigma_2) = \delta(\sigma'_1 * \sigma'_2) 
\wedge
\delta(\sigma_1) = \delta(\sigma'_1)) }
\implies
{\delta(\sigma_2) = \delta(\sigma'_2)}. 
$$
\end{defi}
For example, the $*$ operations on $\RAM$ and $\RAM_\pi$ satisfy this condition.

When the $*$ operation of an algebra $\Sigma$ is cancellative on footprints, we
can define an operation $\fdiff: \Foot(\Sigma) \times \Foot(\Sigma)
\rightharpoonup \Foot(\Sigma)$ of {\bf\em footprint subtraction} as follows.
Consider $l_1,l_2\in\Foot(\Sigma)$. If for some
$\sigma_1,\sigma_2,\sigma\in\Sigma$, we have $l_1 = \delta(\sigma_1)$, $l_2 =
\delta(\sigma_2)$ and $\sigma_2 = \sigma_1 * \sigma$, then we let $l_2\fdiff
l_1=\delta(\sigma)$.  When such $\sigma_1,\sigma_2,\sigma$ do not exist,
$l_2\fdiff l_1$ is undefined.
\begin{prop}
The $\fdiff$ operation is well-defined.
\end{prop}
\proof
Consider $l_1,l_2 \in \Foot(\Sigma)$ and
$\sigma_1,\sigma_2,\sigma'_1,\sigma'_2,\sigma,\sigma'\in\State$ such that  
$\sigma_1,\sigma'_1 \in l_1$, $\sigma_2,\sigma'_2 \in l_2$,
$\sigma_1 = \sigma_2 * \sigma$ and
$\sigma'_1 = \sigma'_2 * \sigma'$. 
We have:
$$
\delta(\sigma_2*\sigma)=\delta(\sigma_1)=l_1=
\delta(\sigma'_1)=\delta(\sigma'_2*\sigma').
$$
Since $\delta(\sigma_2) = \delta(\sigma'_2)$, by Definition~\ref{foot-canc}  this implies
$\delta(\sigma)=\delta(\sigma')$, so that $\fdiff$ is well-defined.
\qed

\smallskip

For $\RAM$, the $\fdiff$ operation is defined as follows. For $l_1, l_2 \in
\Foot(\RAM)$, take any $\sigma_1, \sigma_2 \in \RAM$ such that $\delta(\sigma_1)
= l_1$ and $\delta(\sigma_2) = l_2$. Then $l_2 \fdiff l_1 = \{\sigma \mid
\dom(\sigma) \uplus \dom(\sigma_1) = \dom(\sigma_2) \}$, if $\dom(\sigma_1)
\subseteq \dom(\sigma_2)$; otherwise, $l_1 \fdiff l_2$ is undefined.

We say that a footprint $l_1$ is {\bf \em smaller} than $l_2$,
written $l_1 \preceq l_2$, when $l_2\fdiff l_1$ is defined. 
The $\circ$ and $\fdiff$ operations on footprints satisfy an analogue of
Proposition~\ref{prop-diff}.
\begin{prop}\label{prop-delta}
For all $l_1,l_2,l_3\in\Foot(\Sigma)$, if $l_1 \circ l_2$ and $l_1\fdiff l_3$
are defined, then 
$$
(l_1 \circ l_2)\fdiff l_3 = (l_1 \fdiff l_3) \circ l_2.
$$
\end{prop}




\noindent In the rest of the paper, we fix a separation algebra $\Sigma$ with the $*$
operation cancellative on footprints.

\section{Linearizability with Ownership Transfer\label{sec:observ}}

In the following, we consider descriptions of computations of a library
providing several methods to a multithreaded client. We fix a set $\ThreadID$
of thread identifiers and a set $\MethodName$ of method names.  As we
explained in Section~\ref{sec:intro}, a good definition of linearizability has
to allow replacing a concrete library implementation with its abstract version
while keeping client behaviours reproducible. For this, it should require that
the two libraries have similar client-observable behaviours. Such behaviours are
recorded using {\em histories}, which we now define in our setting.
\begin{defi}\label{def:intact}
An {\bf \em interface action} $\psi$ is an expression of the form
$(t, \call\ m(\sigma))$ or $(t,\ret\ m(\sigma))$,
where $t\in\ThreadID$, $m\in\MethodName$ and $\sigma\in\Sigma$.
We denote the sets of all call and return actions by $\CallAct$ and $\RetAct$,
and the set of all interface actions by $\CallRetAct$.
\end{defi}

An interface action records a call to or a return from a library method $m$ by
thread $t$. The component $\sigma$ in $(t, \call\ m(\sigma))$ specifies the part
of the state transferred upon the call from the client to the library; $\sigma$
in $(t, \ret\ m(\sigma))$ is transferred in the other direction.  For example,
in the algebra $\RAM$, the annotation $\sigma = [42 : 0]$ implies the transfer
of the cell at the address $42$ storing $0$. In the algebra $\RAM_\pi$, $\sigma =
[42 : (0, 1/2)]$ implies the transfer of a read permission for this cell.

\begin{defi}
  A {\bf \em history} $H$ is a finite sequence of interface actions such that
  for every thread $t$, its projection $H|_t$ to actions by $t$ is a sequence of
  alternating call and return actions over matching methods that starts from a
  call action.
\end{defi}

In the following, we use the standard notation for sequences: $\varepsilon$ is
the empty sequence, $\tau(i)$ is the $i$-th element of a sequence $\tau$,
$\tau\pref_k$ is the prefix of $\tau$ of length $k$, and $|\tau|$ is the
length of $\tau$.

Not all histories make intuitive sense with respect to the ownership transfer
reading of interface actions. For example, let $\Sigma = \RAM$ and consider the
history in Figure~\ref{fig:hist}(a). The history is meant to describe {\em
  all}\/ the interactions between the library and the client.  According to the
history, the cell at the address $10$ was first owned by the client, and then
transferred to the library by thread $1$. However, before this state was
transferred back to the client, it was again transferred from the client to the
library, this time by thread $2$. This is not consistent with the intuition of
ownership transfer, as executing the second action requires the cell to be owned
both by the library and by the client, which is impossible in $\RAM$.

As we show in this paper, histories that do not respect the notion of ownership,
such as the one above, cannot be generated by any program, and should not be
taken into account when defining linearizability.  We now use the notion of
footprints of states from Section~\ref{sec:prelim} to characterise formally the
set of histories that respect ownership.
A finite history $H$ induces a partial function $\dbf{H}:\Foot(\Sigma)
\rightharpoonup \Foot(\Sigma)$, which tracks how a computation with the history
$H$ changes the footprint of the library state:
$$
\begin{array}[t]{@{}r@{\ }c@{\ }l@{\ \ }l@{}}
\dbf{\varepsilon}l 
&=& 
l;
\\[2pt]
\dbf{H \psi}l
& = & 
\dbf{H}l \circ \delta(\sigma),
&
\mbox{if } \psi = (\_, \call\ \_(\sigma)) \wedge {(\dbf{H}l \circ \delta(\sigma))\fdef};
\\[2pt]
\dbf{H \psi}l
& = & 
\dbf{H}l \fdiff \delta(\sigma),
&
\mbox{if } \psi = (\_, \ret\ \_(\sigma)) \wedge {(\dbf{H}l\fdiff\delta(\sigma))\fdef};
\\[2pt]
\dbf{H \psi}l 
& = & 
\mbox{undefined},
&
\mbox{otherwise}.
\end{array}
$$
Using this function, we characterise histories respecting the notion of
ownership as follows.
\begin{defi}\label{defn:balance}
  A history $H$ is {\bf \em balanced} from $l\in\Foot(\Sigma)$ if $\dbf{H}(l)$
  is defined. We call subsets of $\BHistory = \{(l,H) \mid \mbox{$H$ is balanced
    from $l$}\}$ {\bf\em interface sets}.
\end{defi}
An interface set can be used to describe all the behaviours of a library
relevant to its clients. In the following, $\cH$ ranges over interface sets.

To keep client behaviours reproducible when replacing a concrete library by an
abstract one, we do not need to require the latter to reproduce the histories of
the former exactly: the histories generated by the two libraries can be
different in ways that are irrelevant for their clients. We now introduce a {\em
  linearizability relation} that matches a history of a concrete library with
that of the abstract one that yields the same client-observable behaviour.
\begin{defi}\label{lin}
  The {\bf\em linearization relation} $\sqsubseteq$ on histories is defined as
  follows: $H \sqsubseteq H'$ holds if there exists a bijection $\rho \colon
  \{1,\ldots, |H|\} \to \{1,\ldots, |H'|\}$ such that 
$$
\begin{array}{@{}l@{}l@{}}
\forall i, j.\, 
(H(i) = H'(\rho(i))) \wedge ((i < j \wedge 
 ( & (\exists t.\, H(i) = (t, \_) \wedge H(j) = (t,\_))
\vee {} \\[2pt]
& (H(i) = (\_, \ret\ \_) \wedge H(j) = (\_,\call\ \_)))) {\implies} \rho(i) < \rho(j)).
\end{array}
$$

We lift $\sqsubseteq$ to $\BHistory$ as follows: $(l,H) \sqsubseteq (l',H')$
holds if $l' \preceq l$ and $H \sqsubseteq H'$.

Finally, we lift $\sqsubseteq$ to interface sets as follows: $\cH_1 \sqsubseteq
\cH_2$ holds if
$$
\forall (l_1,H_1) \in\cH_1.\, \exists (l_2,H_2) \in \cH_2.\, (l_1,H_1)
\sqsubseteq (l_2, H_2).
$$
\end{defi}
Thus, a history $H$ is linearized by a history $H'$ when the latter is a
permutation of the former preserving the order of actions within threads and
non-overlapping method invocations. The duration of a method invocation is
defined by the interval from the method call action to the corresponding return
action (or to the end of the history if there is none). An interface set $\cH_1$
is linearized by an interface set $\cH_2$, if every history in $\cH_1$ may be
reproduced in a linearized form by $\cH_2$ without requiring more memory. We now
discuss the definition in more detail.

Definition~\ref{lin} treats parts of memory whose ownership is passed between
the library and the client in the same way as parameters and return values in
the classical definition~\cite{linearizability}: they are required to be the
same in the two histories. In fact, the setting of the classical definition can
be modelled in ours if we pass parameters and return values via the heap. Let
$\Sigma = \RAM$ and let us fix distinct locations $\myarg_t \in \Loc$ for $t \in
\ThreadID$ meant for the transfer of parameters and return values. Then
histories of the classical definition are represented in our setting by
histories where all actions are of the form
$$
(t, \call\ m([\myarg_t : {\it param}])) \mbox{ or } (t,\ret\ m([\myarg_t : {\it
  retval}])), \mbox{ where }
{\it param}, {\it retval} \in \Val.
$$
The novelty of our definition lies in restricting the histories considered to
balanced ones, which are the only ones that can be produced by programs (we
formalise this fact in Section~\ref{sec:semantics}). The notion of balancedness
also plays a key role in proving the Abstraction Theorem in the presence of
ownership transfer (Section~\ref{sec:rearr}).

\begin{figure}
\leftline{(a):}

\bigskip

\includegraphics[scale=.34, trim= 0 16cm 0 0cm]{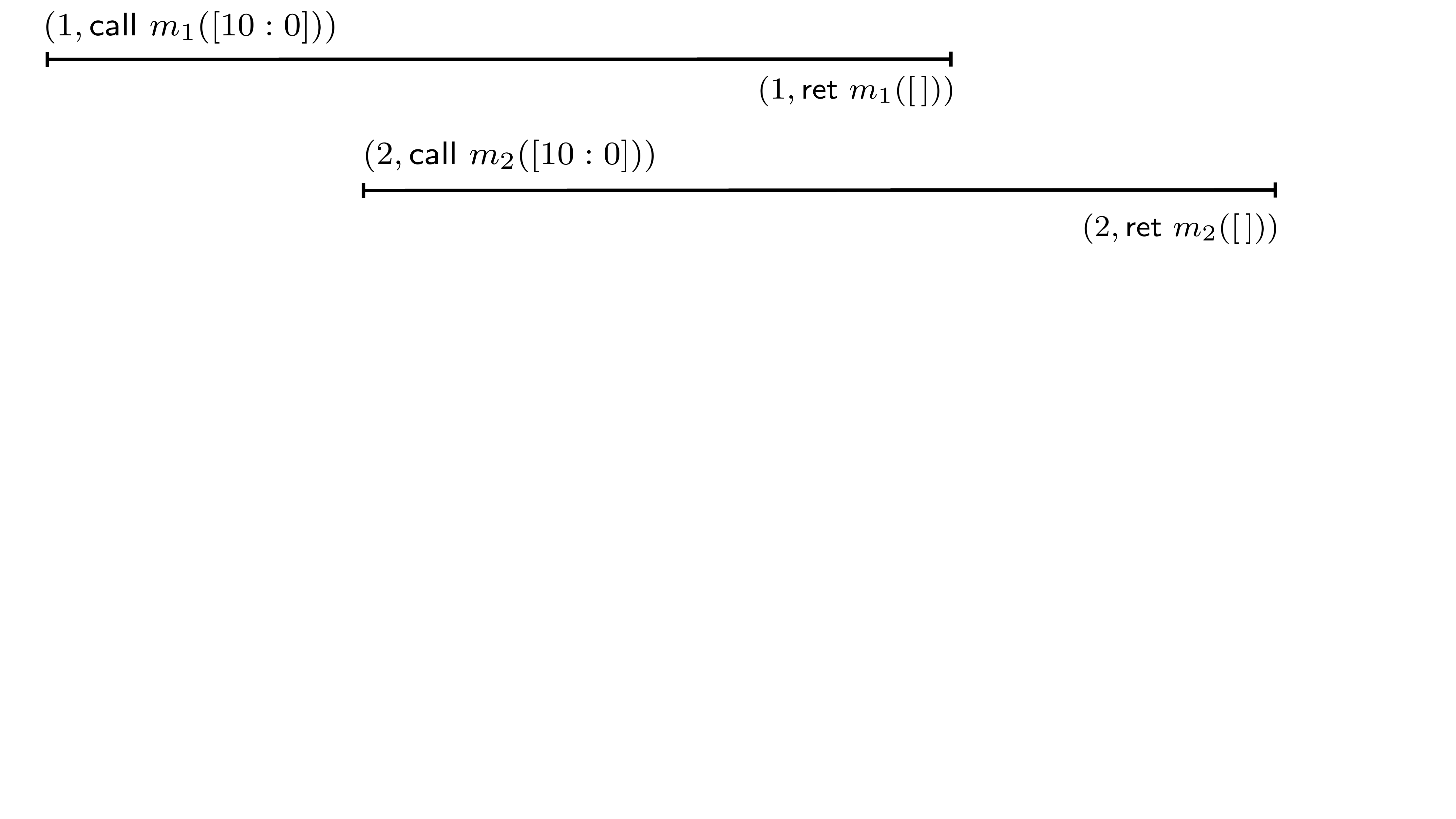}

\leftline{(b):}

\bigskip

\includegraphics[scale=.34, trim= 0 16cm 0 0cm]{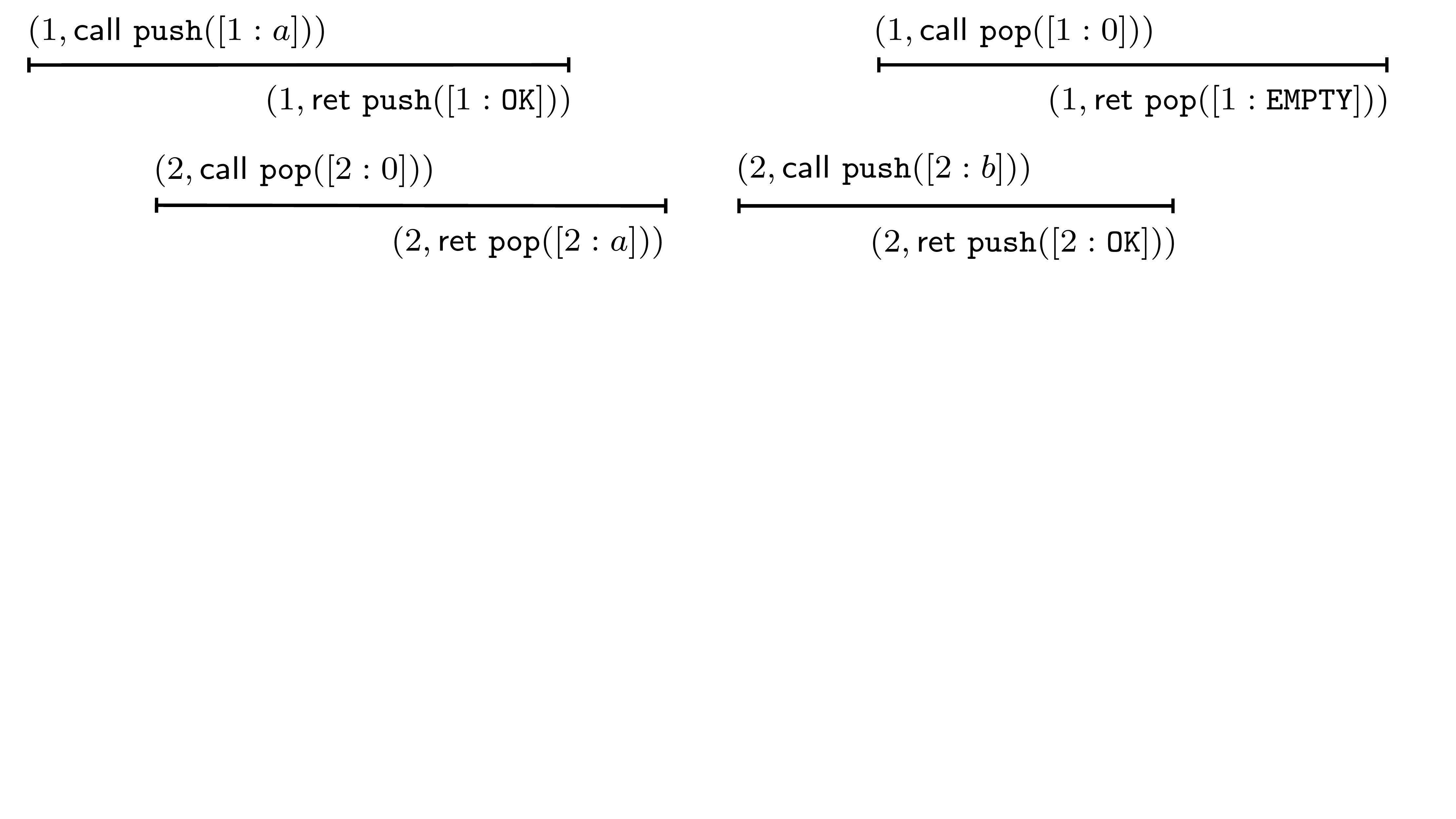}

\leftline{(c):}

\bigskip

\includegraphics[scale=.34, trim= 0 19cm 0 0cm]{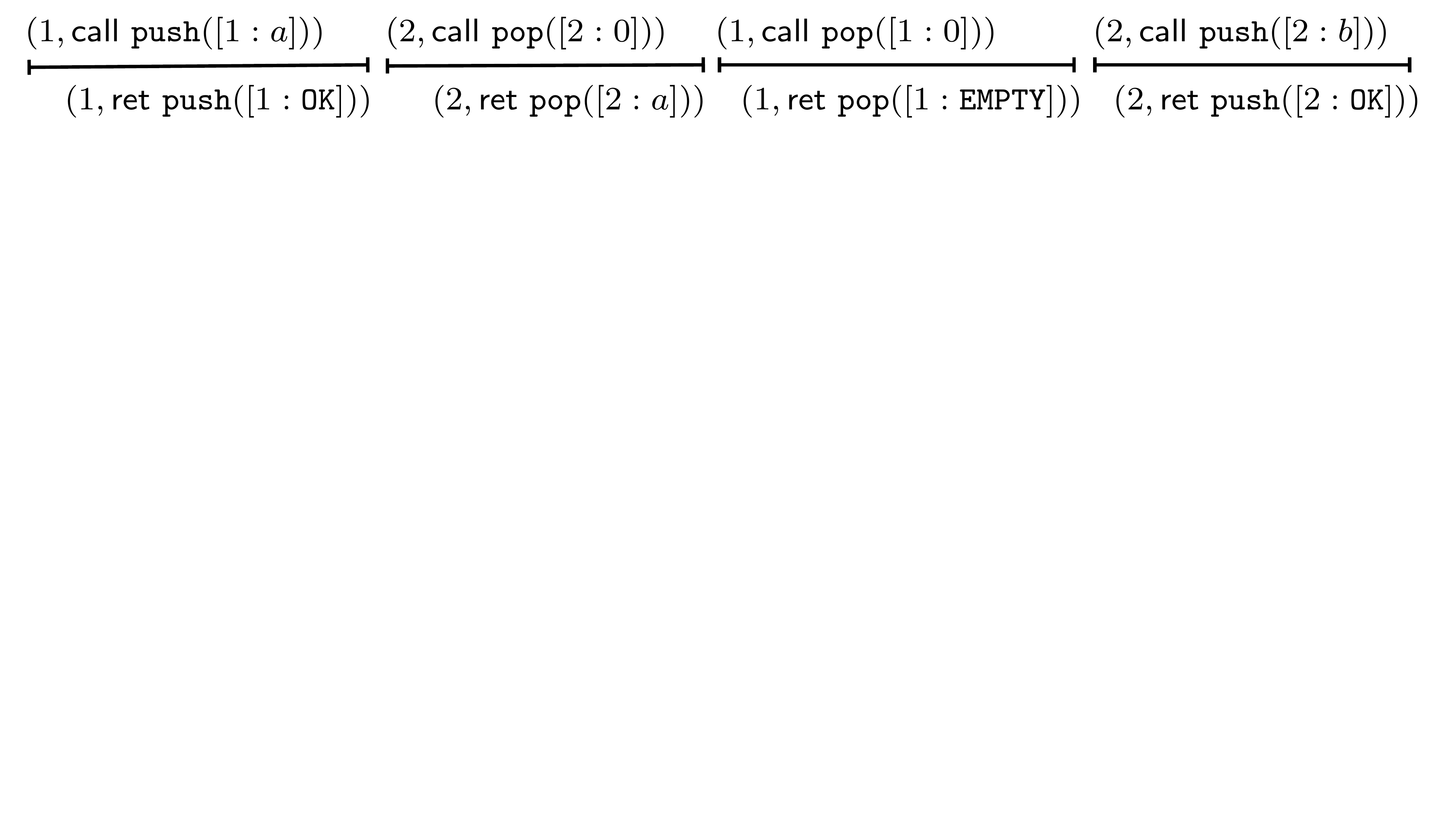}

\leftline{(d):}

\bigskip

\includegraphics[scale=.34, trim= 0 16cm 0 0cm]{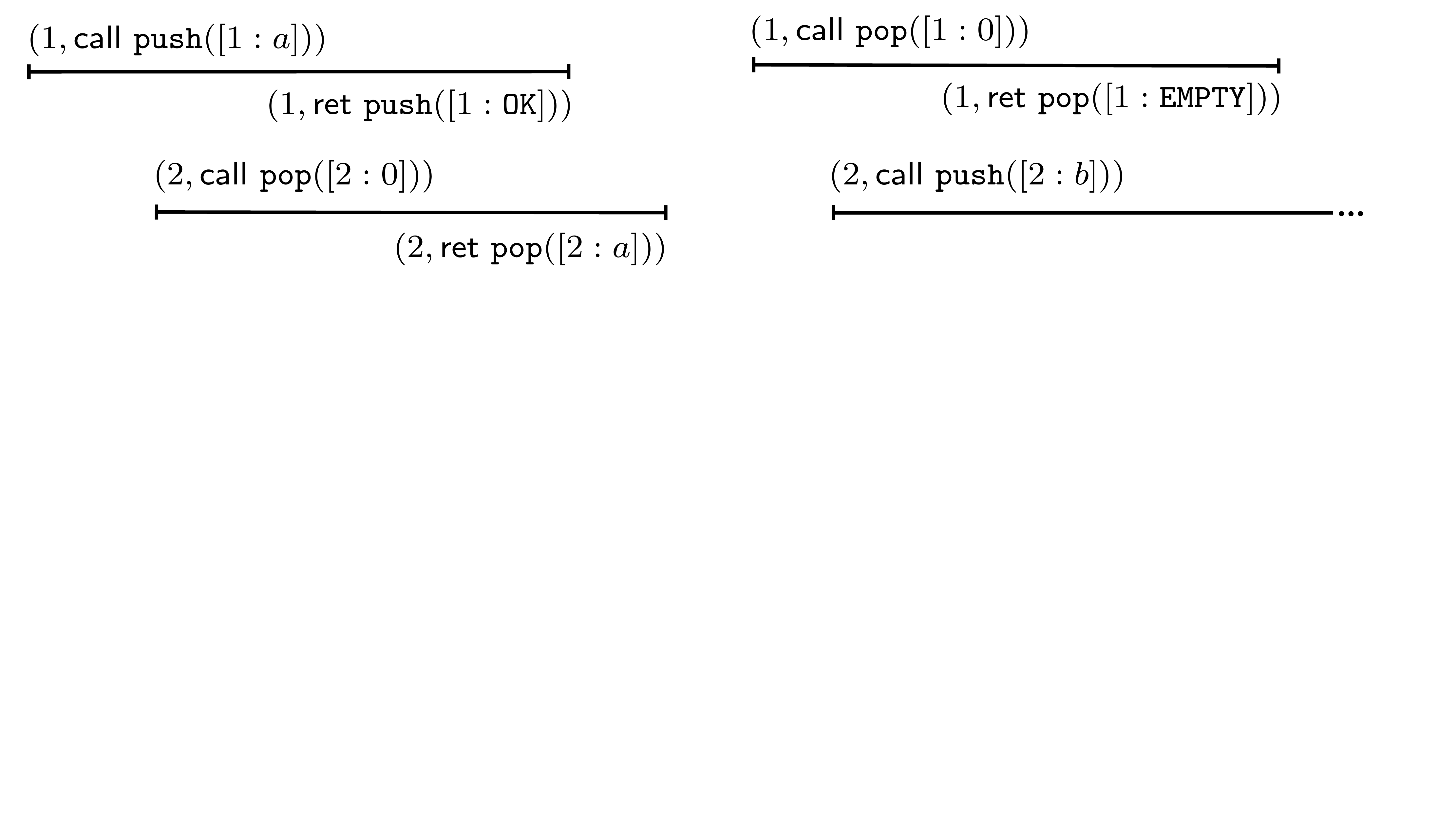}

\leftline{(e):}

\bigskip

\includegraphics[scale=.34, trim= 0 16cm 0 0cm]{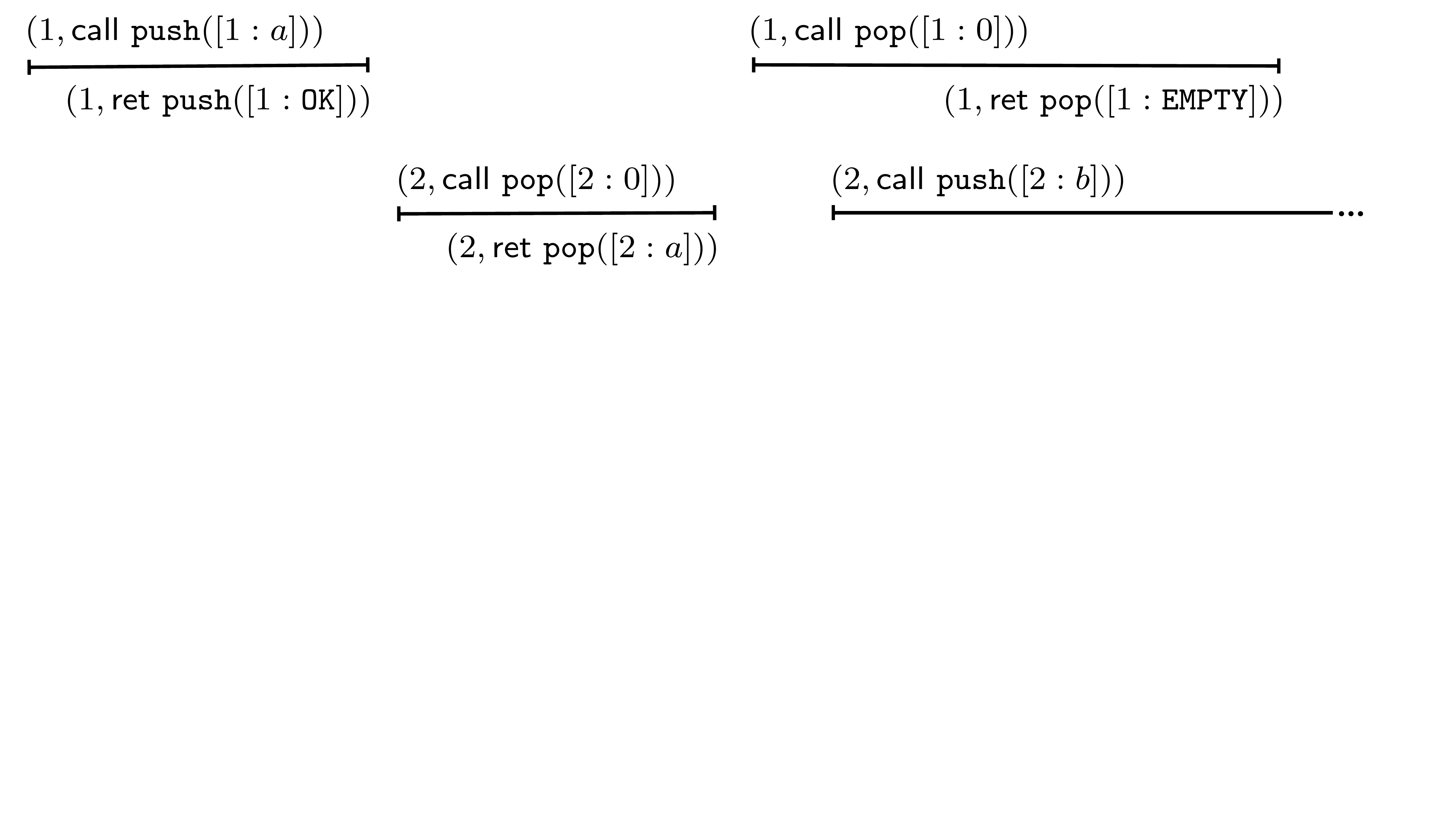}
\caption{Example histories}
\label{fig:hist}
\end{figure}

The fact that the linearizability relation allows us to permute actions by
different threads lets us arrange method invocations into a linear sequence.
For example the history in Figure~\ref{fig:hist}(b) is linearized by that in
Figure~\ref{fig:hist}(c). The former might correspond to a concurrent stack
implementation, where threads $1$ and $2$ pass parameters and return values via
locations $1$ and $2$, respectively. Histories such as the one in
Figure~\ref{fig:hist}(c), where a call to every method is immediately followed
by the corresponding return, are called {\bf\em sequential}. Sequential
histories correspond to abstract libraries with every method implemented
atomically. When the histories in Figures~\ref{fig:hist}(b)
and~\ref{fig:hist}(c) are members of interface sets defining the behaviour of
concurrent and atomic stack implementations, the linearizability relationship
between them allows us to justify that the call to {\tt pop} by thread $1$ in
Figure~\ref{fig:hist}(b) can return ${\tt EMPTY}$, since this behaviour can be
witnessed by the valid history of the atomic implementation in
Figure~\ref{fig:hist}(c). (In fact, the {\tt pop} would also be allowed to
return $b$, since the resulting history would be linearized by a sequential
history with the {\tt push} of $b$ preceding the {\tt pop}.)

The requirement that the order of non-overlapping method invocations be
preserved is inherited from the classical notion of
linearizability~\cite{linearizability}. As shown by Filipovi\'c et
al.~\cite{tcs10}, this requirement is essential to validate an Abstraction
Theorem for clients that can communicate via shared client-side variables. For
example, since in Figure~\ref{fig:hist}(b) the {\tt push} of $a$ returns before
the {\tt push} of $b$ is called, the order between these method invocations has
to stay the same in any linearizing history, such as the one in
Figure~\ref{fig:hist}(c).

Following Filipovi\'c et al.~\cite{tcs10}, we do not require the abstract
history $H'$ to be sequential, like in the classical definition of
linearizablity. This allows our definition to compare behaviours of two
concurrent library implementations. We also allow a concrete history to contain
calls without matching returns, arising, e.g., because the corresponding method
invocation did not terminate. In this case, we require the same behaviour to be
reproduced in the abstract history~\cite{icalp}, which is possible because the
latter does not have to be sequential. For example, the history in
Figure~\ref{fig:hist}(d) is linearized by that in Figure~\ref{fig:hist}(e). This
yields a simpler treatment of non-terminating calls than the use of completions
in the classical definition of linearizability~\cite{linearizability}.


Definition~\ref{lin} requires that the initial footprint of an abstract history
$H'$ be smaller than that of the concrete history $H$. This requirement is
standard in data refinement~\cite{blame}: it ensures that, when we replace a
concrete library by an abstract one in a program, the library-owned memory stays
disjoint from the client-owned one. It does not pose problems in practice, as
the abstract library generating $H'$ usually represents some of the data
structures of the concrete library abstractly and, hence, more concisely.

So far we have defined the notion of linearizability on interface sets without
taking into account library implementations that generate them. In the rest of
the paper, we develop this notion for libraries written in a particular
programming language and prove an Abstraction Theorem, which guarantees that a
library can be replaced by another library linearizing it when we reason about
its client program.

\section{Programming Language\label{sec:prog}}

We consider a simple concurrent programming language:
$$
\begin{array}{@{}l@{\ \ }c@{\ \ }l@{}}
C & ::=&
  c \mid m \mid  C; C \mid  C + C \mid  C^*
\\[2pt]
\cL & ::=&  \{m = C;\ \ldots;\ m = C \}
\\[2pt]
\cS & ::=& 
  {\sf let}\ \cL\ {\sf in}\ C \parallel \ldots \parallel C  
\end{array}
$$
A program consists of a single {\bf\em library} $\cL$ implementing methods
$m\in\MethodName$ and its {\bf\em client} $C_1 \parallel \ldots
\parallel C_n$, given by a parallel composition of threads. 
The language is parameterised by a set of {\bf\em primitive commands} $c \in
\PComm$, meant to execute atomically. Commands also include method calls
$m\in\MethodName$, sequential composition $C;C'$, nondeterministic choice $C+C'$
and finite iteration $C^*$. We use $+$ and ${}^*$ instead of conditionals and
while loops for theoretical simplicity: as we show below, the latter can be
defined in the language as syntactic sugar. Methods do not take arguments and do
not return values, as these can be passed via special locations on the heap
associated with the identifier of the thread calling the method
(Section~\ref{sec:observ}). We disallow nested method calls. We also assume that
every method called in the program is defined by the library, and thus call
$\cS$ a {\bf\em complete program}. An {\bf \em open program} is a library $\cL$
without a client, or a client $\Cc$ without a library implementation:
$$
\begin{array}{@{}l@{\ \ }c@{\ \ }l@{}}
\Cc  & ::=&   {\sf let}\ [-]\ {\sf in}\ C\parallel\ldots\parallel C
\\[2pt]
\cO &::=& \cS \mid \Cc \mid \cL  
\end{array}
$$
In $\Cc$, we allow the client to call methods that are not defined in the
program (but belong to the missing library).  An open program represents a library
or a client considered in isolation.  The novelty of the kind of open programs
we consider here is that we allow them to communicate with their environment via
ownership transfers. We now define a way to specify a contract this
communication follows.

\subsection{Method Specifications\label{sec:mspecs}}

A {\bf \em predicate} is a set of states from $\Sigma$, and a {\bf \em
  parameterised predicate} is a mapping from thread identifiers to
predicates. We use the same symbols $p,q,r$ for ordinary and parameterised
predicates; it should always be clear from the context which one we mean.  When
$p$ is a parameterised predicate, we write $p_t$ for the predicate obtained by
applying $p$ to a thread $t$. Both kinds of predicates can be described
syntactically, e.g., using separation logic assertions~\cite{lics02}.

We define possible ownership transfers between components with the aid of {\bf
  \em method specifications} $\Gamma$, which are sets of Hoare triples $\{p\}\
m\ \{q\}$, at most one for each method. Here $p$ and $q$ are parameterised
predicates such that $p_t$ describes pieces of state transferred when thread $t$
calls the method $m$, and $q_t$, those transferred at its return. Note that the
intention of the pre- and postconditions in method specifications is only to
identify the areas of memory transferred; in other words, they describe the
``type'' of the returned data structure, but not its ``value''. As usual for
concurrent algorithms, a complete specification of a library is given by its
abstract implementation (Section~\ref{sec:refinement}). 


For example, as we discussed in Section~\ref{sec:intro}, when programmers store
pointers in a concurrent container, they often intend to transfer the ownership
of the data structures these pointers identify at calls to and returns from the
container's methods. In Figure~\ref{fig:impl}(a) we give an example of such a
container---a bounded stack represented by an array. For readability, we write
examples in C instead of the minimalistic language introduced above.
The library protects the array by a lock; more complicated algorithms allow a
higher degree of concurrency~\cite{herlihy-book}.
Take $\Sigma = \RAM$ and for $x \in \Loc$ let
$\Obj(x)\subseteq \RAM$ denote the set of states representing all well-formed
data structures of a certain type allocated at the address $x$. For example, for
objects with a single integer field we have $\Obj(x) = \{[x : y] \mid y \in
\Val\}$. Then the specification of the stack when it stores pointers to such
data structures can be given as follows:
\be\label{container}
\begin{array}{l}
\{\exists x.\, \myarg_t \mapsto x * {\sf Obj}(x)
\}\ {\tt push}\ \{\myarg_t \mapsto {\tt OK} \vee 
(\myarg_t \mapsto {\tt FULL} * {\sf Obj}(x))\};
\\[2pt]
\{\myarg_t \mapsto \_\}\ {\tt pop}\ 
\{\exists x.\, \myarg_t \mapsto x * ((x = {\tt EMPTY} \wedge \emp) \vee 
(x \not= {\tt EMPTY} \wedge {\sf Obj}(x)))\}.
\end{array}
\ee
Here we use the separation logic syntax to describe predicates parameterised by
the thread identifier $t$. Thus, $\emp$ denotes the empty heap $[\,]$, $x
\mapsto y$ the heap $[x : y]$, and $*$ the combination of 
heaps with disjoint domains. In the following
we also use the assertion $x..y \mapsto \_$ for $x \leq y$, 
denoting all the heaps with the domain $\{x, x+1, \ldots, y\}$. 
We use distinguished locations $\myarg_t$, $t \in \ThreadID$
to pass parameters and return values. According to the specification, the stack
gets the ownership of an object when a pointer to it is pushed, and gives it up
when the pointer is popped.

\begin{figure}
{\small
\begin{tabular}{@{}l@{\qquad\quad}|@{\qquad\quad}l@{}}
{}
\begin{source}[basicstyle = \ttfamily,language=C,escapeinside={(*}{*)}]
void *stack[SIZE];
int count = 0; // count of
// elements stored
Lock array_lock; // protects 
// the array and the count

int push(void *arg) {
  lock(array_lock);
  if (count == SIZE) {
    unlock(array_lock);
    return FULL;
  }
  stack[count++] = arg;
  unlock(array_lock);
  return OK;
}

void *pop() {
  lock(array_lock);
  if (count == 0) {
    unlock(array_lock);
    return EMPTY;
  }
  void *obj = stack[--count];
  unlock(array_lock);
  return obj;
}
\end{source}
&
\begin{source}[basicstyle = \ttfamily,language=C,escapeinside={(*}{*)}]
struct Node { Node *prev, *next; };
Node *free_list; // a cyclic doubly-
// linked list with a sentinel node
Lock *list_lock; // protects the list

void free(void *arg) {
  Node *block = (Node*)arg;
  lock(list_lock);
  block->prev = free_list;
  block->next = free_list->next;
  free_list->next->prev = block;
  free_list->next = block;
  unlock(list_lock);
}

void *alloc() {
  lock(list_lock);
  if (free_list->next == free_list) {
    unlock(list_lock);
    return NULL;
  }
  Node *block = free_list->next;
  free_list->next = block->next;
  block->next->prev = free_list;
  unlock(list_lock);
  return block;
}
\end{source}
\\
\\
\qquad  \qquad \qquad \qquad  (a)
& \quad \qquad \qquad \qquad  \qquad  (b)
\end{tabular}
}
\caption{Example concurrent library implementations: 
(a) a bounded stack storing pointers to objects;
(b) a memory allocator managing memory blocks of a fixed size}
\label{fig:impl}
\end{figure}

Now take ${\it size} \in \mathbb{N}$ and let 
$$
\Obj(x) = \big\{\sigma \mid \dom(\sigma) = \{x, \ldots, x+{\it size}-1\}\big\}.
$$
We specify an allocator managing blocks of {\it size} memory cells as follows:
\be\label{alloc}
\begin{array}{l}
\{\exists x.\, \myarg_t \mapsto x * (x..(x+{\it size}-1) \mapsto \_) \}\ 
{\tt free}\ 
\{\myarg_t \mapsto \_\};
\\[2pt]
\{\myarg_t \mapsto \_\}\ {\tt alloc}\ 
\{\exists x.\, \myarg_t \mapsto x * ((x = 0 \wedge \emp) \vee 
(x \not= 0 \wedge (x..(x+{\it size}-1) \mapsto \_)))\}.
\end{array}
\ee
The specification corresponds to the ownership transfer reading of allocator
calls explained in Section~\ref{sec:intro}. In Figure~\ref{fig:impl}(b) we give
an example allocator corresponding to the specification (we have omitted
initialisation code from the figure). Note that, unlike the stack, the allocator
does access the blocks of memory transferred to it by the client, since it
stores free-list pointers inside them.

To define the semantics of ownership transfers unambiguously
(Section~\ref{sec:semantics}), we require pre- and postconditions in method
specifications to be {\em precise}~\cite{seplogic-concurrent}.
\begin{defi}
  A predicate $r\in\power{\Sigma}$ is {\bf \em precise} if for every state
  $\sigma$ there exists at most one substate $\sigma_1$ satisfying $r$, i.e.,
  such that $\sigma_1 \in r$ and $\sigma = \sigma_1 * \sigma_2$ for some
  $\sigma_2$. 
\end{defi}
Since the $*$ operation is cancellative, when such a substate $\sigma_1$ exists,
the corresponding substate $\sigma_2$ is unique and is denoted by $\sigma \diff
r$. A parameterised predicate $r$ is precise if so is $r_t$ for every $t$.

Informally, a precise predicate carves out a unique piece of the heap. For
example, assuming the algebra $\RAM$, the predicate $\{[42 : 0]\}$ and those
used in the allocator specification are precise. However, the predicate $\{[42 :
0], [\,]\}$ is not: when $\sigma = [42 : 0]$ we can take either $\sigma_1 = [42
: 0]$ and $\sigma_2 = [\,]$, or $\sigma_1 = [\,]$ and $\sigma_2 = [42 : 0]$.

A {\bf\em specified open program} is of the form $\Gamma \vdash \Cc$ or $\cL :
\Gamma$. In the former, the specification $\Gamma$ describes all the methods
that $\Cc$ may call.  In the latter, $\Gamma$ provides specifications for the
methods in the open program that can be called by its external environment. In
both cases, $\Gamma$ specifies the type of another open program that can fill in
the hole in $\Cc$ or $\cL$. When we are not sure which form a program has, we
write $\Gamma \vdash \cO: \Gamma'$. In this case, if $\cO$ does not have a
client, then $\Gamma$ is empty; if $\cO$ does not have a library, then $\Gamma'$
is empty; and if $\cO$ is complete, then both $\Gamma$ and $\Gamma'$ are empty.
For specified open programs 
$$
\Gamma \vdash \Cc \;=\; 
\Gamma \vdash {\sf let}\ [-]\ {\sf in}\ C_1\parallel\ldots\parallel C_n
$$
and
$$
\cL : \Gamma 
$$ 
agreeing on the specification $\Gamma$ of library methods, 
we denote by $\Cc(\cL)$ the complete program 
$$
{\sf let}\ \cL\ {\sf in}\ C_1\parallel\ldots\parallel C_n.
$$

\subsection{Primitive Commands\label{sec:prim}} 

We now discuss primitive commands in more detail. Consider the set
$\power{\Sigma} \cup \{\top\}$ of subsets of $\Sigma$ with a special element
$\top$ used to denote an error state, resulting, e.g., from dereferencing an
invalid pointer. We extend the $*$ operation on $\power{\Sigma}$ to
$\power{\Sigma}\cup \{\top\}$ by letting $\top*p=p*\top=\top*\top = \top$ for
all $p\in\power{\Sigma}$.

We assume an interpretation of every primitive command $c \in \PComm$ as a
transformer $f^t_c:\Sigma \to (\power{\Sigma} \cup \{\top\})$, which maps
pre-states to states obtained when thread $t \in\ThreadID$ executes $c$ from a
pre-state. The fact that our transformers are parameterised by $t$ allows atomic
accesses to areas of memory indexed by thread identifiers. This idealisation
simplifies the setting in that it lets us do without special thread-local or
method-local storage for passing method parameters and return values.

Some typical primitive commands are:
$$
{\sf skip},
\qquad
[E] = E',
\qquad
{\sf assume}(E), 
$$
where expressions $E$ are defined as follows:
$$
E \;::=\;
\mathbb{Z} 
\mid
\thread 
\mid
[E] 
\mid
 E+E 
\mid
{-}E 
\mid
{!E} 
\mid
\ldots
$$
Here ${\sf tid}$ refers to the identifier of the thread executing the command,
$[E]$ returns the contents of the address $E$ in memory, and $!E$ is the C-style
negation of an expression $E$---it returns $1$ when $E$ evaluates to $0$, and
$0$ otherwise. The $\assume(E)$ command filters out all the input states where $E$
evaluates to $0$. Hence, after $\assume(E)$ is executed, $E$ always has
a non-zero value. Using it, the standard commands for conditionals and loops 
can be defined in our language
as follows:
\begin{gather}
(\mathsf{if}\ E\ \mathsf{then}\ C_{1}\ \mathsf{else}\ C_{2})
\ =\ (\assume(E);C_{1}) + (\assume(!E);C_{2}), \label{if-assm}
\\
(\mathsf{while}\ E\ \mathsf{do}\ C)
\ =\ 
(\assume(E);C)^*;\, \assume(!E).\notag
\end{gather}

\noindent For the above commands $c$ and $t\in\ThreadID$, we define $f^t_c: \RAM \to
\power{\RAM} \cup \{\top\}$ using the transition relation $\leadsto_t: \RAM
\times (\RAM \cup \{\top\})$ in Figure~\ref{fig:transfer}:
$$
f_c^t(\sigma)
\ =\ 
\mbox{if ($\sigma, c \leadsto_t \top$) then }
\top \mbox{ else } \bigcup\big\{ \sigma' \mid 
\sigma, c \leadsto_t  \sigma' \big\}.
$$
In the figure, $\db{E}_{\sigma, t} \in \Val \cup \{\top\}$ denotes the result of
evaluating the expression $E$ in the state $\sigma$ with the current thread
identifier $t$. When this evaluation dereferences illegal memory addresses, it
results in the error value $\top$.  We define $f^t_c: \RAM_\pi \to
\power{\RAM_\pi} \cup \{\top\}$ similarly, but using the transition relation
$\leadsto_t: \RAM_\pi \times (\RAM_\pi \cup \{\top\})$ in
Figure~\ref{fig:transfer-permission1}. The transformers formalise the semantics
of permissions explained in Section~\ref{sec:prelim}: permissions less than $1$
allow reading, and the full permission $1$ additionally allows writing. Note
that $\assume$ yields an empty set of post-states when its condition evaluates
to zero, leading to the program getting stuck. Thus, even though, when executing
the {\sf if} statement~(\ref{if-assm}), both branches of the non-deterministic
choice will be explored, only the branch where the assume condition evaluates to
true will proceed further.

\begin{figure}[t]
$$
\begin{array}{@{}l@{\quad}l@{\quad}l@{\qquad}l@{}}
\sigma, {\sf skip}
& \leadsto_t & 
\sigma
\\[2pt]
\sigma, [E] = E'
& \leadsto_t &
\sigma{[\db{E}_{\sigma, t} : \db{E'}_{\sigma,t}]},
&
\mbox{if}\
\db{E}_{\sigma,t} \in \dom(\sigma), \db{E'}_{\sigma,t}\in \Val
\\[2pt]
\sigma, {[E]} = E'
& \leadsto_t & 
\top,
&
\mbox{if}\ 
\db{E}_{\sigma,t} \not\in \dom(\sigma)
\mbox{ or } \db{E'}_{\sigma,t} = \top
\\[2pt]
\sigma, \assume(E)
& \leadsto_t & 
\sigma,
&
\mbox{if}\
\db{E}_{\sigma, t} \in \Val - \{0\}
\\[2pt]
\sigma, \assume(E)
& \leadsto_t &
\top,
&
\mbox{if}\
\db{E}_{\sigma, t} = \top
\end{array}
$$
\caption{\label{fig:transfer}
Transition relation for sample primitive commands in
the $\RAM$ model.
The result $\top$ indicates that the command faults.}
\end{figure}
\begin{figure}[t]
$$
\begin{array}{@{}l@{\quad}l@{\quad}l@{\quad \ \ }l@{}}
\sigma, {\sf skip}
& \leadsto_t & 
\sigma
\\[2pt]
\sigma, {[E]} = E'
& \leadsto_t &
\sigma{[\db{E}_{\sigma, t} : (\db{E'}_{\sigma,t},1)]},
&
\mbox{if}\ \sigma(\db{E}_{\sigma,t})  = (\_, 1),
\db{E'}_{\sigma,t} \in \Val
\\[2pt]
\sigma, {[E]} = E'
& \leadsto_t & 
\top,
&
\mbox{if the above condition does not hold}
\\[2pt]
\sigma, \assume(E)
& \leadsto_t & 
\sigma,
&
\mbox{if}\
\db{E}_{\sigma, t} \in \Val - \{0\} 
\\[2pt]
\sigma, \assume(E)
& \leadsto_t &
\top,
&
\mbox{if}\
\db{E}_{\sigma, t} = \top
\end{array}
$$
\caption{\label{fig:transfer-permission1}
Transition relation for sample primitive commands
in the $\RAM_\pi$ model. The evaluation of expressions $\db{E}_{\sigma, t}$ ignores
permissions in $\sigma$.}
\end{figure}

For our results to hold, we need to place some restrictions on the transformers
$f_c^t$ for every primitive command $c \in \PComm$ and thread $t \in \ThreadID$:
\begin{description}
\item[Footprint Preservation] 
$\forall \sigma, \sigma'\in \Sigma.\, \sigma' \in f^t_c(\sigma) \implies
\delta(\sigma')=\delta(\sigma)$. 

\item[Strong Locality]
$\!\forall \sigma_1, \sigma_2 \,{\in}\, \Sigma.\, 
{(\sigma_1\,{*}\,\sigma_2)\fdef} \wedge f^t_c(\sigma_1) \,{\not=}\, \top {\implies}
f^t_c(\sigma_1\,{*}\,\sigma_2) \,{=}\, f^t_c(\sigma_1)\,{*}\,\{\sigma_2\}$.
\end{description}
Footprint Preservation prohibits primitive commands from allocating or
deallocating memory. This does not pose a problem, since in the context of
linearizability, an allocator is just another library and should be treated as such.
The Strong Locality of $f_c^t$ says that, if a command $c$ can be safely executed from a state $\sigma_1$, then when
executed from a bigger state $\sigma_1*\sigma_2$, it does not change the
additional state $\sigma_2$ and its effect depends only on the state $\sigma_1$
and not on the additional state $\sigma_2$. 

The Strong Locality is a strengthening of the locality property in separation
logic~\cite{asl}:
$$
\forall \sigma_1, \sigma_2 \in \Sigma.\, 
{(\sigma_1 * \sigma_2)\fdef} \wedge {f^t_c(\sigma_1) \not= \top} \implies
f^t_c(\sigma_1 * \sigma_2) \subseteq f^t_c(\sigma_1) * \{\sigma_2\}.
$$
Locality rules out commands that can check if a cell is allocated in the heap
other than by trying to access it and faulting if it is not allocated. For
example, let $\Sigma=\RAM$ and consider the following transformer $f^t:
\RAM\rightarrow \power{\RAM}\cup\{\top\}$:
$$
f^t(\sigma) \ =\ 
\text{if $\sigma(1)\fdef$ then } \{\sigma[1 : 0]\} \mbox{ else }
\{\sigma\}.
$$
The transformer $f^t$ defines the denotation of a `command'
that writes $0$ to the cell at the address $1$ if it is allocated and acts as
a no-op if it is not. This violates locality. Indeed, 
take $\sigma_1=[\,]$ and $\sigma_2=[1 : 1]$. Then
$$
f^t(\sigma_1 * \sigma_2)= f^t([1:1]) =\{[1 : 0]\}
$$ 
and
$$
f^t(\sigma_1) * \{\sigma_2\}=f^t([\,]) * \{[1:1]\}=
\{[\,]\} * \{[1:1]\}=\{[1:1]\}.
$$
Hence, $f^t(\sigma_1* \sigma_2) \subseteq f^t(\sigma_1) * \{\sigma_2\}$ does not
hold.

While locality prohibits the command from changing the additional state, it
permits the effect of the command to depend on this state~\cite{blame}. 
The Strong Locality forbids such dependencies. To see this,
consider another `command' defined by the following transformer $f^t:
\RAM\rightarrow \power{\RAM}\cup\{\top\}$:
$$
\begin{array}{@{}l@{\ }l@{}}
f^t(\sigma)  \ =\ & 
\mbox{if $\sigma(1)\fundef$ then } \top 
\\[2pt]
&\mbox{else if  $({\sigma(2)\fdef})$ then $\{\sigma[1 :  0]\}$}
\\[2pt]
%
%
&\mbox{else $\{\sigma[1 : 0], \sigma[1 : 1]\}$}.
\end{array}
$$
The command does not access the cell at the address $2$, since it does not fault
if the cell is not allocated. However, when the cell is allocated, the effect of
the command depends on its value. It is easy to check that $f^t$ is
local. However, it is not strongly local, since for $\sigma_1 = [1 : 0]$ and
$\sigma_2 = [2 : 0]$, we have
$$
f^t(\sigma_1 * \sigma_2) = f^t([1 : 0, 2 : 0]) = \{[1: 0, 2:0]\}
$$
and
$$
f^t(\sigma_1)  * \{\sigma_2\} = f^t([1 : 0]) * \{[2:0]\} = \{[1:0], [1:1]\} *
\{[2:0]\} = \{[1:0, 2:0], [1:1, 2:0]\},
$$
so that $f^t(\sigma_1* \sigma_2) = f^t(\sigma_1) * \{\sigma_2\}$ does not
hold.
The property of Strong Locality subsumes the one of contents independence used
in situations similar to ours in previous work on data refinement in a
sequential setting~\cite{blame}.

The transformers for standard commands, except memory (de)allocation, satisfy
the conditions of Footprint Preservation and Strong Locality.

\section{Client-Local and Library-Local Semantics\label{sec:semantics}}

We now give the semantics to complete and open programs. In the latter case, we
define component-local semantics that include all behaviours of an open program
under any environment satisfying the specification associated with it. In
Section~\ref{sec:refinement}, we use these to lift linearizability to libraries
and formulate the Abstraction Theorem.

Programs in our semantics denote sets of {\em traces}, recording every step in a
computation. These include both internal actions by program components and calls
and returns. We define program semantics in two stages. First, given a program,
we generate the set of the possible execution traces of the program. This is
done solely based on the structure of its statements, without taking into
account restrictions arising from the semantics of primitive commands or
ownership transfers.
%
%
The next step filters out traces that are not consistent with the above
restrictions using a trace evaluation process and, for open programs, annotates
calls and returns appropriately.

\subsection{Traces\label{sec:traces}}  

Traces consist of {\em actions}, which include primitive commands performed
internally by a component and calls or returns, possibly annotated with states.
Thus actions, include all interface actions $\psi$ from Definition~\ref{def:intact}.
\begin{defi}
The set of {\bf \em actions} is defined as follows: 
$$
\varphi  \in \Act \ ::= \ 
\psi \mid
(t, c) \mid
 (t, \call\ m) \mid 
(t,\ret\ m),
$$
where $t\in\ThreadID$, $m\in\MethodName$ and $c \in \PComm$.
\end{defi}
\begin{defi}
  A {\bf \em trace} $\tau$ is a finite sequence of actions such that for every
  thread $t$, the projection of $\tau$ to $t$'s call and return actions is a
  sequence of alternating call and return actions over matching methods that
  starts from a call action.
\end{defi}

We classify actions in a trace as those performed by the client and the library
based on whether they happen inside a method.
\begin{defi}
For a trace $\tau$ and an index $i \in \{1,\ldots,|\tau|\}$, an action 
$\tau(i)$ is a {\bf \em client action} if $\tau(i) = (t,c)$ for
some thread $t$ and a primitive command $c$ and 
$$
\forall j.\, j < i \wedge \tau(j) = (t, \call\ \_) \implies 
\exists k.\, j < k < i \wedge \tau(k) = (t, \ret\ \_).
$$
An action $\tau(i)$ is a {\bf \em library action} if $\tau(i) = (t,c)$ for
some $t$ and $c$ but $\tau(i)$ is not a client action, that is,
$$
\exists j.\, j < i \wedge \tau(j) = (t, \call\, \_) \wedge 
\neg\exists k.\, j < k < i \wedge \tau(k) = (t, \ret\, \_).
$$
A trace is a {\bf \em client trace}, if all of its actions of the form $(t,c)$
are client actions; it is a {\bf \em library trace}, if they all of them are
library actions.
\end{defi}
In the following, $\kappa$ denotes client traces, $\lambda, \zeta, \alpha,
\beta$ library traces, and $\tau$ arbitrary ones. We write $\client(\tau)$ for
the projection of $\tau$ to client, call and return actions, $\lib(\tau)$ for
that to library, call and return actions, and $\history(\tau)$ for that to call
and return actions.

%

\subsection{Trace Sets\label{sec:tsets}}

Consider a program $\Gamma \vdash \cO : \Gamma'$ and let
$M\subseteq \MethodName$ be the set of methods implemented by its library or
called by its client. We define the trace set $\dbp{\Gamma \vdash \cO : \Gamma'} \in
\power{\WTrace}$ of $\cO$ in Figure~\ref{fig:trace}.  We first define the trace
set $\dbp{C}_t\eta$ of a command $C$, parameterised by the identifier $t$ of
the thread executing it and a mapping $\eta \in M
\times \ThreadID \to \power{\WTrace}$ giving the trace set of the body of every
method that $C$ can call when executed by a given thread.  The trace set of a
client $\dbp{C_1\parallel\ldots\parallel C_n} \eta$ is obtained by
interleaving traces of its threads.
\begin{figure}[t]
$$
\begin{array}{@{}c@{}}
\begin{array}{@{}r@{}c@{}l@{}}
\dbp{c}_{t}\eta & {} = {} & \{(t, c)\}
\\[2pt]
\dbp{C_1 + C_2}_t \eta & = & {\dbp{C_1}_t \eta} \cup {\dbp{C_2}_t \eta}
\\[2pt]
\dbp{C^*}_{t}\eta & = & \big(\dbp{C}_t \eta\big)^*
\\[2pt]
\dbp{m}_t \eta & = & 
\{(t, \call\ m)\, \tau\, (t, \ret\ m) \mid \tau \in \eta(m,t)\}
\\[2pt]
\dbp{C_1; C_2}_t \eta & = & 
\{\tau_1 \tau_2 \mid \tau_1 \in \dbp{C_1}_t \eta \wedge 
\tau_2 \in \dbp{C_2}_t \eta\}
\\[2pt]
\dbp{C_1 \parallel \ldots \parallel C_n} \eta & = & 
\bigcup\{\tau_1 \parallel \ldots \parallel \tau_n \mid \forall t \in \{ 1,\ldots,n\}.\, \tau_t \in \dbp{C_t}_t \eta\}
\end{array}
\\
\\
\begin{array}{@{}r@{}c@{}l@{}}
\dbp{ {\sf let}\ \{m = C_m \mid m \in M\}\ {\sf in}\ 
C_1 \parallel\ldots \parallel C_n} & = &
\prefix(\dbp{C_1 \parallel\ldots \parallel C_n} (\mylambda (m,t).\, \dbp{C_m}_t (\_)))
\\[2pt]
\dbp{\Gamma : {\sf let}\ [-]\ {\sf in}\ C_1\parallel\ldots\parallel C_n} & {} = {} &
\prefix(\dbp{C_1\parallel\ldots\parallel C_n} 
(\mylambda (m,t).\, \{\varepsilon\}))
\\[2pt]
\dbp{\{m = C_m \mid m \in \{m_1,\ldots,m_j\}\} : \Gamma} & = & {} 
\\[2pt]
\multicolumn{3}{@{}l@{}}{
\hfill
\begin{array}{@{}r@{}}
\prefix(\bigcup\nolimits_{k\ge 1} \dbp{
C_\cmgc \parallel \ldots (\mbox{$k$ times}) \ldots
\parallel C_\cmgc}(\mylambda (m,t).\, \dbp{C_m}_t (\_)))
\\[2pt]
(\mbox{where } C_\cmgc = (m_1 + \ldots + m_j)^*)
\end{array}
}
\end{array}
\end{array}
$$
\caption{Trace sets of commands and programs. Here $\prefix(T)$ is the prefix
  closure of $T$ and $\tau \in \tau_1 \parallel \ldots \parallel \tau_n$ if and
  only if every action in $\tau$ is done by a thread $t\in\{1,\ldots, n\}$ and
  for all such $t$, we have $\tau|_t = \tau_t$.  We use $\mylambda$ for
  functions, in contrast to $\lambda$ for library traces.  }
\label{fig:trace}
\end{figure}

The trace set $\dbp{\Cc(\cL)}$ of a complete program is that of its client
computed with respect to a mapping $\mylambda (m,t).\, \dbp{C_m}_t (\_)$
associating every method $m$ with the trace set of its body $C_m$. Since we
prohibit nested method calls, $\dbp{C_m}_t\eta$ does not depend on $\eta$. We
prefix-close the resulting trace set to take into account incomplete
executions. In particular, this allows the thread scheduler to be unfair: a
thread can be preempted and never scheduled again.

A program $\Gamma \vdash \Cc$ generates client traces $\dbp{\Gamma \vdash \Cc}$,
which do not include internal library actions. This is achieved by associating
an empty trace with every library method. Finally, a program $\cL : \Gamma'$
generates all possible library traces $\dbp{\cL : \Gamma'}$. This is achieved by
running the library under its {\em most general client}, where every thread
executes an infinite loop, repeatedly invoking arbitrary library methods.

\subsection{Evaluation\label{sec:eval}}  

\begin{figure}
\leftline{$\db{\Gamma \vdash \cO : \Gamma' } : \Sigma \to (2^\Trace \cup \{\top\})$:}
\vspace{5pt}
\centerline{$
\begin{array}{@{}r@{\,}c@{\,}l@{}}
\db{\Gamma \vdash \cO : \Gamma' }\sigma & {} = {} &
\mbox{if } \exists \tau \in \dbp{\Gamma \vdash
  \cO : \Gamma'}.\, \db{\Gamma \vdash \tau : \Gamma'}\sigma = \top 
\mbox{ then } \top 
\\[2pt]
&& \mbox{else } 
\{\tau' \mid \exists \tau \in \dbp{\Gamma \vdash
  \cO : \Gamma'}.\,  (\_, \tau') \in 
\db{\Gamma \vdash \tau : \Gamma'}\sigma\}
\end{array}
$}

\bigskip

\leftline{$\db{\Gamma \vdash \tau : \Gamma'} : \Sigma \rightarrow 
(\power{\Sigma \times \Trace} \cup \{\top\})$:}
\vspace{5pt}
\centerline{$
\begin{array}{@{}r@{\,}c@{\,}l@{}}
\db{\Gamma \vdash \varepsilon : \Gamma'}\sigma & {} = {}& \{(\sigma,\varepsilon)\}
\\[2pt]
\db{\Gamma \vdash  \tau\varphi : \Gamma'}\sigma & {} = {} &
\begin{array}[t]{@{}l@{}}
\mbox{if $(\db{\Gamma \vdash  \tau : \Gamma'}\sigma = \top)$ then $\top$}
\\[2pt]
\mbox{else if $(\exists (\sigma',\_) \in \db{\Gamma \vdash  \tau : \Gamma'}\sigma.\,
\db{\Gamma \vdash \varphi : \Gamma'}\sigma' = \top)$ then $\top$} 
\\[2pt]
\mbox{else $\{(\sigma'',\tau'\varphi') \mid\exists
  \sigma'.\,(\sigma',\tau') \in \db{\Gamma \vdash \tau : \Gamma'}\sigma \wedge
  (\sigma'',\varphi') \in \db{\Gamma \vdash \varphi : \Gamma'}\sigma'\}$}
\end{array}
\end{array}
$}

\bigskip

\leftline{$\db{\Gamma \vdash \varphi : \Gamma'} : \Sigma 
\rightarrow (\power{\Sigma \times \Act} \cup \{\top\})$:}
\vspace{5pt}
\centerline{$
\begin{array}{@{}r@{\,}c@{\,}l@{}}
\db{\Gamma \vdash (t,c) : \Gamma'}\sigma 
& {} =  {}&
\mbox{if } (f^t_c(\sigma) = \top) \mbox{ then } \top \mbox{ else } 
\{(\sigma',(t,c)) \mid \sigma' \in f^t_c(\sigma)\}
\\[2pt]
\db{(t,\call\ m)}\sigma & {} = {}& \{(\sigma, (t,\call\ m))\}
\\[2pt]
\db{(t,\ret\ m)}\sigma & {} = {}& \{(\sigma, (t,\ret\ m))\}
\\[2pt]
\db{ (t,\call\ m) : (\{p\}\ {m}\ \{q\}), \Gamma'}\sigma
& {} = {}& \{(\sigma * \sigma_p, (t,\call\ m(\sigma_p))) \mid \sigma_p \in p_t 
\wedge (\sigma * \sigma_p)\fdef \}
\\[2pt]
\db{(t,\ret\ m) : (\{p\}\ {m}\ \{q\}), \Gamma'}\sigma & {} = {}&
\mbox{if } (\sigma \diff q_t)\fundef \mbox{ then } \top
\mbox{ else } \{(\sigma \diff q_t, (t,\ret\ m(\sigma \diff {(\sigma \diff q_t)})))\}
\\[2pt]
\db{ (\{p\}\ {m}\ \{q\}), \Gamma \vdash (t,\call\ m)}\sigma & {} = {}&
\mbox{if } (\sigma \diff p_t)\fundef \mbox{ then } \top
\mbox{ else } \{(\sigma \diff p_t, (t,\call\ m(\sigma \diff {(\sigma \diff p_t)})))\}
\\[2pt]
\db{ (\{p\}\ {m}\ \{q\}), \Gamma \vdash (t,\ret\ m)}\sigma
& {} = {}& \{(\sigma * \sigma_q, (t,\ret\ m(\sigma_q))) \mid \sigma_q \in q_t \wedge (\sigma * \sigma_q)\fdef \}
\end{array}
$}
\caption{Semantics of programs}\label{fig:eval}
\end{figure}

The set of traces generated using $\dbp{\cdot}$ may include those not consistent
with the semantics of primitive commands or expected ownership transfers. We
therefore define the meaning of a program $\db{\Gamma \vdash \cO : \Gamma'} \in
\State \to (\power{\WTrace} \cup \{\top\})$ by evaluating every trace in
$\dbp{\Gamma \vdash \cO : \Gamma'}$ from a given initial state to determine
whether it is feasible. For open programs, this process also annotates calls and
returns in a trace with states transferred. 

The formal definition of $\db{\Gamma \vdash \cO : \Gamma'}$ is given in
Figure~\ref{fig:eval} with the aid of a trace evaluation function $\db{\Gamma
  \vdash \tau : \Gamma'} : \Sigma \rightarrow \power{\Sigma \times \Trace} \cup
\{\top\}$. Given an initial state, this either yields multiple final states and
annotated traces, or fails and produces $\top$. If the resulting set of
state-trace pairs is empty, then the trace is infeasible and is discarded.  If
the evaluation produces $\top$ on any trace from $\dbp{\Gamma \vdash \cO :
  \Gamma'}$, then the program has no semantics for the given initial state and
its denotation is defined to be $\top$. The evaluation of $\tau$ is defined
inductively on its length using a function $\db{\Gamma \vdash \varphi : \Gamma'}
: \Sigma \rightarrow \power{\Sigma \times \Act} \cup \{\top\}$ that evaluates a
single action $\varphi$. We explain this function by considering separately the
cases of a complete program, open program with a library and open program with a
client.

The evaluation $\db{\varphi}$ for an action in a complete program has the
standard semantics, with the effects of primitive commands computed using
their transformers from Section~\ref{sec:prog}. In this case, calls and returns
are left unannotated, since no ownership transfers to or from the external
environment are performed.

The function $\db{\varphi : \Gamma'}$ gives a {\bf \em library-local} semantics
to the program $\cL:\Gamma'$, in the sense that it generates library traces
under any client respecting $\Gamma'$. When a method $m$ from $\Gamma'$ is
called by thread $t$, the library receives the ownership of any state consistent
with the method precondition $p_t$. This state has to be compatible with that of
the library.  After the method returns, the library has to give up the piece of
state satisfying its postcondition. Since $q_t$ is precise, this piece of state
is determined uniquely. The evaluation faults if the state to be transferred is
not available; thus, a library has no semantics if it violates the contract with
its client given by $\Gamma'$. This also ensures that the histories produced by
a library are balanced.
\begin{prop}\label{safe2wb}
  If $\lambda \in \db{\cL : \Gamma'}\sigma$, then $\history(\lambda)$ is balanced from
  $\delta(\sigma)$.
\end{prop}

The function $\db{\Gamma \vdash \varphi}$ gives a {\bf \em client-local}
semantics to $\Gamma \vdash \Cc$, in the sense that it generates traces of this
client assuming any behaviour of the library consistent with $\Gamma$. When a
thread $t$ calls a method $m$ in $\Gamma$, it transfers the ownership of a piece
of state satisfying the method precondition $p_t$ to the library being called.
As before, this piece is defined uniquely, because preconditions are precise.
When such a piece of state is not available, the evaluation faults. This ensures
that client respects the method specifications of the libraries it uses. When
the method returns, the client receives the ownership of an arbitrary piece of
state satisfying its postcondition $q_t$, compatible with the current state of
the client.

\subsection{Connection between Local and Global Semantics\label{sec:decomp}}
We now formulate a lemma, used in the proof of the Abstraction Theorem
(Section~\ref{sec:refinement}), that states the connection between the
library-local and client-local semantics on one side and the semantics of
complete programs on the other. We start by introducing some auxiliary
definitions.

\begin{defi}
  A program $\Gamma \vdash \cO : \Gamma'$ is {\bf \em safe} at $\sigma$, if
  $\db{\Gamma \vdash \cO : \Gamma'}\sigma \not= \top$; $\cO$ is safe for
  $\cI\subseteq \Sigma$, if it is safe at $\sigma$ for all $\sigma \in \cI$.
\end{defi}

For a set of initial states $\cI\subseteq\Sigma$, let
$$
\db{(\Gamma \vdash \cO : \Gamma'),\cI}=\{(\sigma, \tau) \mid
\sigma\in\cI \wedge \tau\in\db{\Gamma \vdash \cO : \Gamma'}\sigma\}.
$$
We define an operator $\otimes : \power{\Sigma \times \WTrace} \times
\power{\Sigma \times \WTrace} \rightarrow \power{\Sigma \times \WTrace}$
combining the resulting sets $X$ and $Y$ of state-trace pairs produced by 
the client-local and library-local semantics into a set corresponding to the
complete program:
$$
X \otimes Y 
= \{(\sigma * \sigma', \tau) \,\mid\, 
\exists \kappa,\lambda.\,
(\sigma,\kappa) \in X \wedge (\sigma',\lambda) \in Y
\wedge {(\sigma * \sigma')\fdef} \wedge \cover(\tau,\kappa,\lambda)\},
$$
where
$$
\cover(\tau,\kappa,\lambda)
\iff
\history(\kappa) = \history(\lambda) \wedge
\client(\tau) = \erase(\kappa) \wedge
\lib(\tau) = \erase(\lambda)
$$
and $\erase$ is a function on traces that erases the state annotations from
their interface actions.
\begin{lem}
\label{prop:sem:local-global}
Assume $\Gamma \vdash \Cc$ and $\cL : \Gamma$ safe for $\cI_0$ and $\cI_1$,
respectively. Then $\Cc(\cL)$ is safe for $\cI_0*\cI_1$ and
$$
\db{\fillin{\cL}{\Cc}, \cI_0 * \cI_1} = \db{\Gamma \vdash \Cc,\cI_0} \otimes \db{\cL : \Gamma,\cI_1}.
$$
\end{lem}

The lemma shows that the set of traces produced by $\Cc(\cL)$ can be obtained by
combining pairs of traces with the same history produced by $\Cc$ and $\cL$.
Note that, since the semantics of $\Cc(\cL)$ does not annotate calls and returns
with the states transferred, in $\cover$ we have to erase these annotations from
the local traces $\kappa$ or $\lambda$ before comparing the traces with
$\tau$. Unpacking the definition of $\otimes$ and using the fact that
$$
\forall \kappa,\lambda.\,
\history(\kappa) = \history(\lambda)
\implies
\exists \tau.\,\cover(\tau,\kappa,\lambda),
$$
from Lemma~\ref{prop:sem:local-global} we get the following two corollaries.
\begin{cor}[Decomposition]\label{lemma-clientlocal}
  Assume $\Gamma \vdash \Cc$ and $\cL : \Gamma$ safe for $\cI_0$ and $\cI_1$,
  respectively. Then $\Cc(\cL)$ is safe for $\cI_0*\cI_1$ and
\begin{multline*}
\forall (\sigma, \tau) \in \db{\fillin{\cL}{\Cc}, \cI_0 * \cI_1}.\,
\exists (\sigma_0,\kappa)\in \db{\Gamma \vdash \Cc,\cI_0}.\,
\exists (\sigma_1,\lambda) \in \db{\cL : \Gamma,\cI_1}.\, \\
\sigma= \sigma_0*\sigma_1 \wedge 
\cover(\tau,\kappa,\lambda).
\end{multline*}
\end{cor}
\begin{cor}[Composition]\label{lemma-clientlocal2}
If
$\Gamma \vdash \Cc$ 
and 
$\cL : \Gamma$  are
safe for $\cI_0$ and $\cI_1$, respectively, then
\begin{multline*}
\forall (\sigma_1,\kappa) \in \db{\Gamma : \Cc,\cI_0}.\, 
\forall (\sigma_2,\lambda) \in \db{\cL : \Gamma,\cI_1}.\,
({(\sigma_0 * \sigma_1)\fdef} \wedge 
\history(\kappa) = \history(\lambda)) \implies {}
\\
\exists \tau.\, (\sigma_0 *\sigma_1, \tau) \in \db{\fillin{\cL}{\Cc}, \cI_0 *
  \cI_1} 
\wedge \cover(\tau,\kappa,\lambda).
\end{multline*}
\end{cor}

\noindent Corollary~\ref{lemma-clientlocal} can be viewed as carrying over properties of
the local semantics, such as safety, to the global one, and in this sense is the
statement of the soundness of the former with respect to the latter. The
corollary also confirms that the client defined by $\dbp{\cL : \Gamma'}$ and
$\db{\lambda : \Gamma'}$ is indeed most general, as it reproduces library
behaviours under any possible clients. Corollary~\ref{lemma-clientlocal2}
carries over properties of the global semantics to the local ones, stating the
adequacy of the latter.


Lemma~\ref{prop:sem:local-global} is proved in Appendix~\ref{sec:proofs}. Most
of the proof deals with maintaining a splitting of the state of
$\fillin{\cL}{\Cc}$ into the parts owned by $\cL$ and $\Cc$, which changes
during ownership transfers. The proof relies crucially on the safety of the
client and the libraries and the Strong Locality property of primitive commands.
In more detail, safety is defined by considering executions of a component in
the library-local or the client-local semantics. These execute the component
code only on the memory it owns, whose amount only changes with ownership
transfers to and from its environment according to method
specifications. Because of the Strong Locality property, commands fault when
accessing memory cells that are not present in the state they are run from, and
their execution does not depend on any additional memory that might be present
in the state. Hence, when we use a component inside a complete program, its
safety guarantees that the component code does not touch the part of the heap
belonging to other components in the program, and its execution is not affected
by the state of such components. This guarantees that the behaviour a component
produces as part of the complete program can be reproduced when we execute it in
isolation and vice versa, allowing us to establish
Lemma~\ref{prop:sem:local-global}. In practice, the safety of a program can be
established using existing program logics, such as separation
logic~\cite{seplogic-concurrent,rgsep-thesis}.

\section{Abstraction Theorem\label{sec:refinement}}

We are now in a position to define the notion of linearizability on libraries
and prove the central technical result of this paper---the Abstraction Theorem.
We define linearizability between specified libraries $\cL : \Gamma$, together
with their sets of initial states $\cI$. First, using the library-local
semantics of Section~\ref{sec:semantics}, we define the interface set describing
all the behaviours of a library $\cL$ when run from initial states in $\cI$:
$$
\interf(\cL : \Gamma,\cI) = \{(\delta(\sigma_0),\history(\tau)) \mid
(\sigma_0,\tau) \in \db{(\cL : \Gamma),\cI}\} \subseteq \BHistory.
$$
\begin{defi}\label{lin2}
  Consider $\cL_1 : \Gamma$ and $ \cL_2 : \Gamma$ safe for $\cI_1$ and $\cI_2$,
  respectively.  We say that $(\cL_1 : \Gamma,\cI_1)$ \textbf{\em is linearized by}
  $(\cL_2 : \Gamma,\cI_2)$, written $(\cL_1 : \Gamma,\cI_1) \sqsubseteq (\cL_2 :
  \Gamma,\cI_2)$, if, according to Definition~\ref{lin},
$$
\interf(\cL_1 : \Gamma,\cI_1) \sqsubseteq \interf(\cL_2 : \Gamma,\cI_2).
$$
For an interface set $\cH_2$ we say that $(\cL_1 : \Gamma,\cI_1)$ \textbf{\em is
  linearized by} $\cH_2$, written $(\cL_1 : \Gamma,\cI_1) \sqsubseteq \cH_2$, if
$$
\interf(\cL_1 : \Gamma,\cI_1) \sqsubseteq \cH_2.
$$
\end{defi}
Thus, $(\cL_1 : \Gamma,\cI_1)$ is linearized by $(\cL_2 : \Gamma,\cI_2)$ if
every history generated by the library-local semantics of the former may be
reproduced in a linearized form by the library-local semantics of the latter
without requiring more memory. The relation $(\cL_1 : \Gamma,\cI_1) \sqsubseteq
(\cL_2 : \Gamma,\cI_2)$ allows us to specify a library by another piece of code,
but possibly simpler than the original one. For example, the stack and the
allocator from Figure~\ref{fig:impl} with method
specifications~(\ref{container}) and~(\ref{alloc}) can be specified by the
libraries in Figure~\ref{fig:spec}. The libraries replace the array and the
linked list in the implementations by the abstract data types of a sequence and
a set (we assume a trivial extension of the $\RAM$ algebra from
Section~\ref{sec:prelim} to allow memory cells to store values of such
types). Thus, the abstract libraries use less memory than the concrete ones.
Instead of using locking, all operations on the abstract data types are done
atomically; formally, we assume primitive commands corresponding to the code in
the atomic blocks.

\begin{figure}
{\small
\begin{tabular}{@{}l@{\qquad\quad}|@{\qquad\quad}l@{}}
{}
\begin{source}%[basicstyle = \ttfamily,numbers=left, numberstyle=\tiny,language=C,escapeinside={(*}{*)}]
Sequence<void*> stack;

int push(void *arg) {
  atomic { 
    if (nondet()) { return FULL; }
    else { 
      add_to_head(stack, arg); 
      return OK; 
    }
  }
}

void *pop() {
  atomic {
    if (!isEmpty(stack)) {
      void *obj = head(stack);
      stack = tail(stack);
      return obj; 
    } else { return EMPTY; }
  }
}
\end{source}
&
\begin{source}%[basicstyle = \ttfamily,numbers=left, numberstyle=\tiny,language=C,escapeinside={(*}{*)}]
Set<void*> free_list;

void free(void *arg) {
  atomic {
    add(free_list, arg);
  }
}

void *alloc() {
  atomic {
    if (!isEmpty(free_list)) {
      Node *block = 
          (Node*)take(free_list);
      block->next = nondet();
      block->prev = nondet();
      return block;
    } else {
      return 0; 
    }
  }
}
\end{source}
\\
\\
\qquad  \qquad \qquad \qquad  (a)
&  \qquad \qquad \qquad  \qquad  (b)
\end{tabular}
}
\caption{Specifications corresponding to the implementations in
  Figure~\ref{fig:impl}: (a) a bounded stack storing pointers to objects; (b) a
  memory allocator managing memory blocks of a fixed size. The {\tt Node}
  structure is defined in Figure~\ref{fig:impl}(b).}
\label{fig:spec}
\end{figure}

The other relation $(\cL_1 : \Gamma,\cI_1) \sqsubseteq \cH_2$ introduced in
Definition~\ref{lin2} allows us to specify a library directly by an interface
set, without fixing a piece of code generating it. The interface set $\cH_2$ can
still be simpler than that of $(\cL_1 : \Gamma,\cI_1)$, e.g., containing only
sequential histories. Even though the two forms of defining linearizability may
seem very similar, as we show in Section~\ref{sec:rearr}, their mathematical
properties are fundamentally different.

We now formulate two variants of the Abstraction Theorem, corresponding to the
two ways of specifying libraries (we prove them in Section~\ref{sec:rearr}).
\begin{thm}[Abstraction---specification by code]\label{thm2}
If 
\begin{iteMize}{$\bullet$}
\item $\cL_1 : \Gamma$,
$ \cL_2 : \Gamma$,
$\Gamma \vdash \Cc$
are safe for $\cI_1$, $\cI_2$, $\cI$, respectively, and
\item
$(\cL_1 : \Gamma,\cI_1) \sqsubseteq  (\cL_2 : \Gamma,\cI_2)$, 
\end{iteMize}
then
\begin{iteMize}{$\bullet$}
\item
$\Cc(\cL_1)$ and $\Cc(\cL_2)$ are safe for $\cI*\cI_1$ and $\cI *\cI_2$,
respectively, and
\item
$
\forall (\sigma_1,\tau_1) \in \db{\fillin{\cL_1}{\Cc}, \cI * \cI_1}.\,
\exists (\sigma_2,\tau_2) \in \db{\fillin{\cL_2}{\Cc}, \cI *\cI_2}.\,
\client(\tau_1) = \client(\tau_2).
$
\end{iteMize}
\end{thm}
Thus, when reasoning about a client $\fillin{\cL_1}{\Cc}$ of a library $\cL_1$,
we can soundly replace $\cL_1$ by a library $\cL_2$ linearizing it: if a
safety property over client traces holds of $\fillin{\cL_2}{\Cc}$, it will also
hold of $\fillin{\cL_1}{\Cc}$. In practice, we are usually interested in {\bf
  \em atomicity abstraction}, a special case of this transformation when methods
in $\cL_2$ are atomic. An instance is replacing one of the libraries from
Figure~\ref{fig:impl} by its specification from Figure~\ref{fig:spec}. The
requirement that $\Cc$ be safe in the theorem restricts its applicability to
well-behaved clients that do not access memory owned by the library: you cannot
replace a library by another one if the client can access its internal data
structures and thereby ``look inside the box''. Similarly, the safety of the
libraries ensures that they cannot corrupt the data structures owned by the
client. 
%
%

The other version of the Abstraction Theorem, allowing library specification by
an interface set, guarantees that replacing a library by its specification
leaves all the original client behaviours reproducible modulo the following
notion of trace equivalence.
\begin{defi}
  Client traces $\kappa$ and $\kappa'$ are {\bf\em equivalent}, written $\kappa \sim
  \kappa'$, if $\kappa|_t = \kappa'|_t$ for all $t \in \ThreadID$ and the projections
  of $\kappa$ and $\kappa'$ to non-interface actions are identical.
\end{defi}
\begin{thm}[Abstraction---specification by an interface set]\label{thm-spec}
If 
\begin{iteMize}{$\bullet$}
\item $\cL_1 : \Gamma$ and
$\Gamma \vdash \Cc$
are safe for $\cI_1$ and $\cI$, respectively, and
\item
$(\cL_1,\cI_1) \sqsubseteq \cH_2$, 
\end{iteMize}
then
\begin{iteMize}{$\bullet$}
\item
$\Cc(\cL_1)$ is safe for $\cI*\cI_1$ and
\item
$
\forall (\sigma,\tau_1) \in \db{\fillin{\cL_1}{\Cc}, \cI * \cI_1}.\,
\exists \kappa, l.\, 
(\sigma', \kappa) \in \db{\Gamma \vdash \Cc, \cI} \wedge
(l, \history(\kappa)) \in \cH_2 \wedge
{(\delta(\sigma') \circ l)\fdef} \wedge
\client(\tau_1) \sim \erase(\kappa).
$
\end{iteMize}
\end{thm}

\noindent The theorem shows that client behaviours of $\Cc(\cL_1)$ can be reproduced by
the client-local semantics of $\Cc$ projected to histories in $\cH_2$ with
initial footprints compatible with initial client states. Note that
$\client(\tau_1) = \client(\tau_2)$ in Theorem~\ref{thm2} implies that
$\history(\tau_1) = \history(\tau_2)$, i.e., $\Cc(\cL_2)$ can reproduce the
history of $\Cc(\cL_1)$ exactly. In contrast, Theorem~\ref{thm-spec} does not
guarantee this, since $\kappa \sim \kappa'$ does not imply $\history(\kappa) =
\history(\kappa')$; we only know that the projection to non-interface actions is
reproduced. We discuss the reason for this discrepancy below.


\subsection{The Rearrangement Lemma and the Proof of the Abstraction
  Theorem\label{sec:rearr}}

The key component used for establishing Theorem~\ref{thm2} is the Rearrangement
Lemma: if $H \sqsubseteq H'$, then every execution trace of a library producing
$H'$ can be transformed into another trace of the same library that differs from
the original one only in the order of interface actions and produces $H$,
instead of $H'$. Hence, the library specification can simulate any behaviour of
its implementation the client can expect.
\begin{lem}[Rearrangement---library]\label{cor-rearr}
  If $(\delta(\sigma),H) \sqsubseteq (\delta(\sigma'),H')$ and $\cL : \Gamma$ is
  safe at $\sigma'$, then
$$
\forall \lambda' \in\db{\cL : \Gamma}\sigma'.\, 
\history(\lambda') = H' \implies 
\exists \lambda\in\db{\cL : \Gamma}\sigma'.\, \history(\lambda) = H.
$$
\end{lem}
The proof of the lemma is highly non-trivial and is a subject of
Section~\ref{sec:rearr-proof}. We point out Lemma~\ref{cor-rearr} would not hold
had we included unbalanced histories in our definition of linearizability. To
show this, take $\Sigma = \RAM$ and consider the histories in
Figure~\ref{fig:cex}.  In Figure~\ref{fig:cex}(b) the library receives the cell
$10$ from the client, then returns it and then receives it again. Even though
the history in Figure~\ref{fig:cex}(a) is linearized by that in
Figure~\ref{fig:cex}(b), the former is not balanced, and by
Proposition~\ref{safe2wb}, cannot be produced by $\cL$.  This shows that
Lemma~\ref{cor-rearr} does not hold for unbalanced $H$.

\begin{figure}[t]
\leftline{(a):}

\bigskip

\hspace{1cm}\includegraphics[scale=.34, trim= 0 16cm 0 0cm]{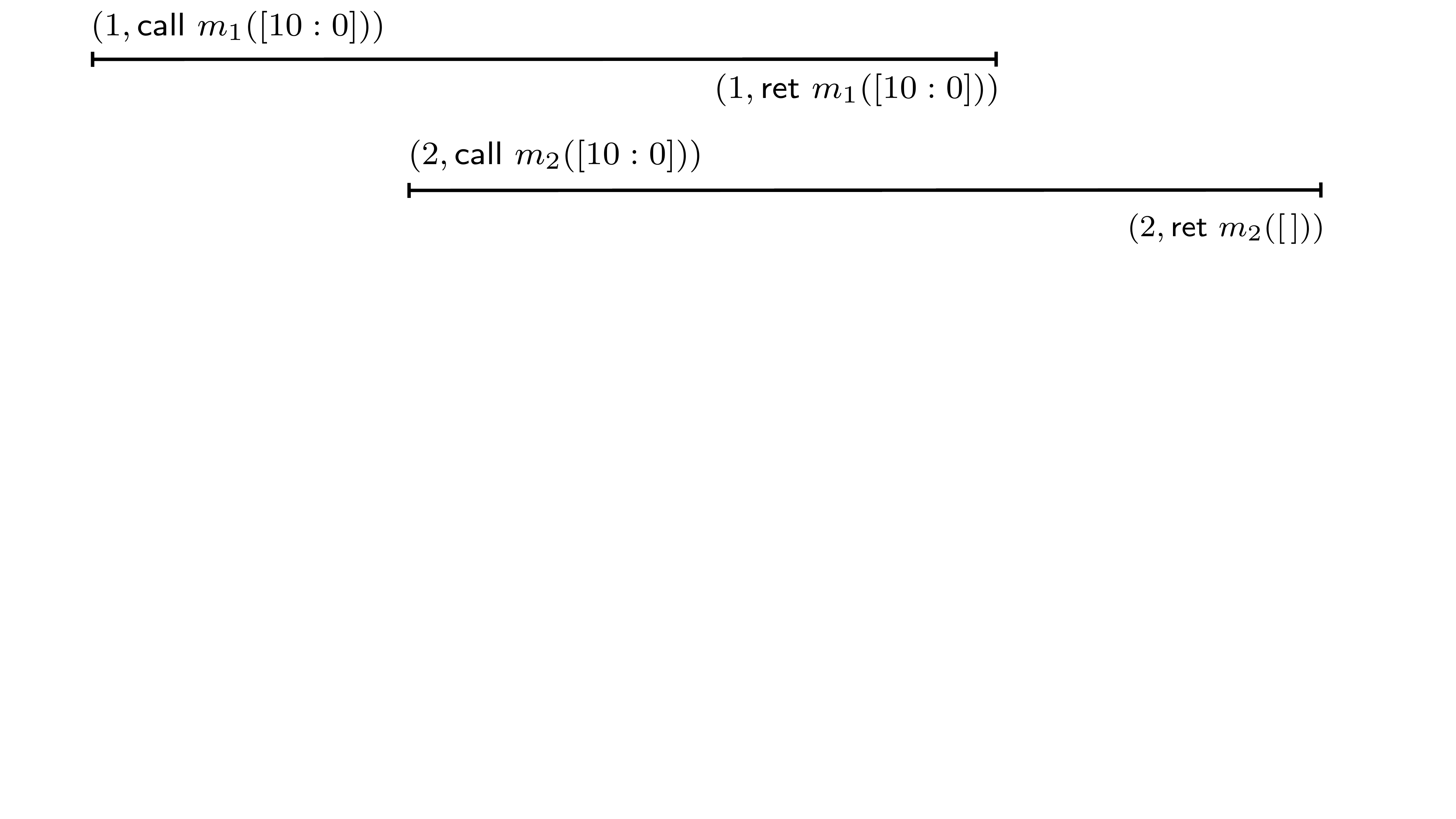}

\leftline{(b):}

\bigskip

\includegraphics[scale=.34, trim= 0 21cm 0 0cm]{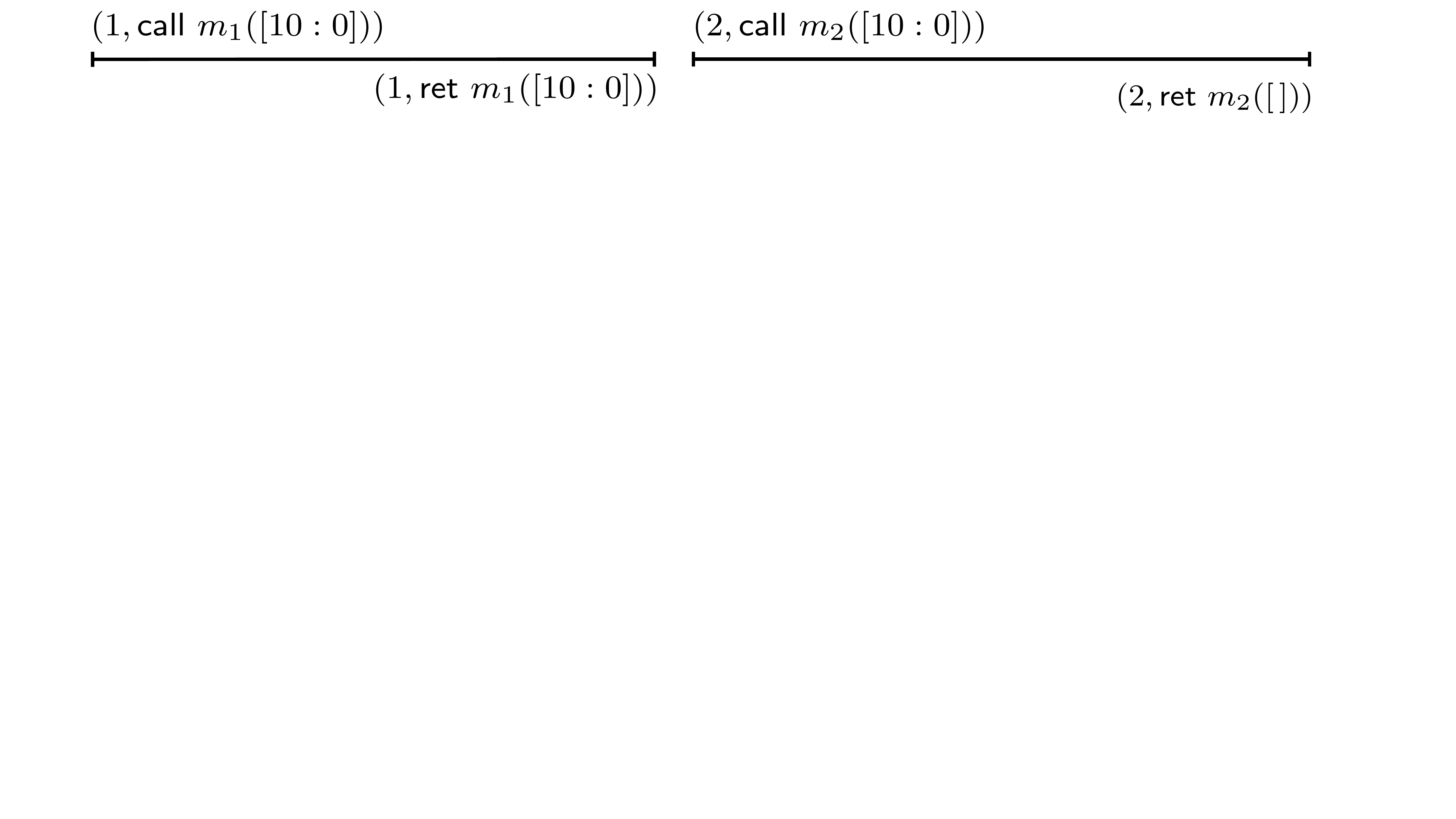}
\caption{Counterexample showing the need for history balancedness in
  Lemma~\ref{cor-rearr}} 
\label{fig:cex}
\end{figure}

Note that we have $H\sqsubseteq H$ for any history $H$. As a consequence, from
Lemma~\ref{cor-rearr} we obtain the following surprising result, stating that
linearizability between libraries is equivalent to inclusion between the sets of
histories they produce.
\begin{cor}\label{cor:subset}
If $\cL_1 : \Gamma$ and $\cL_2 : \Gamma$ are safe for $I_1$ and $I_2$, respectively, then
\begin{multline*}
(\cL_1 : \Gamma, \cI_1) \sqsubseteq (\cL_2 : \Gamma, \cI_2) \iff {}\\
\forall (l, H) \in \interf(\cL_1 : \Gamma, \cI_1).\, \exists l'.\, l' \preceq l \wedge
(l',H) \in \interf(\cL_2 : \Gamma, \cI_2).
\end{multline*}
\end{cor}

\noindent This fact is not a consequence of ownership transfer and also holds for the
classical notion of linearizability. Intuitively, Lemma~\ref{cor-rearr}, and
hence, Corollary~\ref{cor:subset} hold due to a closure property of the
semantics of the language from Section~\ref{sec:prog}. Namely, in this and other
programming languages, there may always be a delay between the point when a
library method is called and when it starts executing and, conversely, when it
ends executing and when the control returns to the client. For example, when
executing the code in Figure~\ref{fig:spec}(a), there may be delays between a
call to {\tt push}, the execution of the atomic block and the return from {\tt
  push}. Hence, this library can produce both of the histories in
Figures~\ref{fig:hist}(b) and~\ref{fig:hist}(c).

Due to this property of the program semantics, a trace from $\dbp{\cL}$, e.g.,
one producing the history in Figure~\ref{fig:hist}(c), will stay valid if we
execute some of the calls in it earlier and returns later, like in
Figure~\ref{fig:hist}(b). This is also (usually) safe given the ownership
transfer reading of calls and returns in the library-local semantics defined by
$\db{\cL: \Gamma}$: it just means that the library receives state from the
client earlier and gives it up later. The proof of Lemma~\ref{cor-rearr} uses
such transformations on a history to ``de-linearize'' it, e.g., transforming the
history in Figure~\ref{fig:hist}(c) into that in Figure~\ref{fig:hist}(b). The
above closure property is also the reason for Theorem~\ref{thm2} guaranteeing
that $\Cc(\cL_2)$ can reproduce the history of $\Cc(\cL_1)$ exactly.

Given that replacing a library $\cL_1$ by its linearization $\cL_2$ does not
simplify its interface set, can Theorem~\ref{thm2} really simplify reasoning
about a complete program $\Cc(\cL_1)$?  Fortunately, the answer is yes, since
the point of the theorem is to simplify {\em the code} of this program. For
example, replacing the library in Figure~\ref{fig:impl}(a) by the one in
Figure~\ref{fig:spec}(a) allows us to pretend in reasoning about a complete
program that changes to the library state, shared between different threads, are
atomic, and thus consider fewer possible thread interleavings. Calls and returns
in such a complete program are merely thread-local operations that do not
complicate reasoning. In Section~\ref{sec:example}, we discuss an example of
using Theorem~\ref{thm2} to simplify proofs of complicated algorithms.

Specifying a library by an interface set $\cH_2$ instead of code, as in
Theorem~\ref{thm-spec}, does not allow us to get results such as
Lemma~\ref{cor-rearr} and Corollary~\ref{cor:subset}, since the set $\cH_2$ is
not guaranteed to satisfy any closure properties. For example, it might contain
only sequential histories, where every call is immediately followed by the
corresponding return without a delay. In fact, $\cH_2$ has to be simpler than
the interface set of $\cL_1$ for Theorem~\ref{thm-spec} to be useful, since this
is what the theorem replaces $\cL_1$ by. Fortunately, to prove
Theorem~\ref{thm-spec} we can exploit a closure property of the client-local
semantics, formalised by the following variant of the Rearrangement Lemma: any
client trace can be transformed into an equivalent one with a given history
linearizing the history of the original one.
\begin{lem}[Rearrangement---client]\label{thm}
  If $(l,H) \sqsubseteq (l',H')$ and $\Gamma \vdash \Cc$ is safe at $\sigma$, then
$$
\forall \kappa \in\db{\Gamma \vdash \Cc}\sigma.\, 
{(\delta(\sigma) \circ l)\fdef} \wedge
\history(\kappa) = H \implies \exists
\kappa'\in\db{\Cc}\sigma.\, \history(\kappa') = H' \wedge \kappa \sim \kappa'.
$$
\end{lem}

\noindent Intuitively, the lemma holds because, in the client-local semantics, it is safe
to execute calls later and returns earlier. Like Lemma~\ref{cor-rearr}, this
lemma would not hold if we allowed $H'$ to be unbalanced.


In summary, when a library is specified by the code of its abstract
implementation, the ability to linearize a concrete history while looking for a
matching abstract one allowed by Definition~\ref{lin} is not strictly needed.
However, it is indispensable when the library is specified directly by a set of
histories. We were able to obtain this insight into the original definition of
linearizability by formalising the guarantees the linearizability of a library
provides to its clients as Abstraction Theorems.

Using Lemmas~\ref{cor-rearr} and~\ref{thm}, we now prove the two versions of
the Abstraction Theorem.

\paragraph{\em Proof of Theorem~\ref{thm2}.} The safety of $\Cc(\cL_1)$ and
$\Cc(\cL_2)$ follows from Corollary~\ref{lemma-clientlocal}.  Take $(\sigma,
\tau_1)\in \db{\fillin{\cL_1}{\Cc}, \cI * \cI_1}$.  We transform the trace
$\tau_1$ of $\Cc(\cL_1)$ into a trace $\tau_2$ of $\Cc(\cL_2)$ with the same
client projection using the local semantics of $\cL_1$, $\cL_2$ and
$\Cc$. Namely, we first apply Corollary~\ref{lemma-clientlocal} to generate a pair
$(\sigma_{\bf l}^1,\lambda_1)\in \db{\cL_1 : \Gamma, \cI_1}$ of a library-local
initial state and a trace and a client-local pair $(\sigma_{\bf c},\kappa) \in
\db{\Gamma \vdash \Cc, \cI}$, such that 
\be\label{eq-decomp} 
\sigma = \sigma_{\bf c}*\sigma_{\bf l}^1 \ \wedge\ 
\client(\tau_1)=\erase(\kappa) \wedge\ 
\history(\kappa) =\history(\lambda_1).  
\ee 
Since $(\cL_1 : \Gamma, \cI_1)
\sqsubseteq (\cL_2 : \Gamma, \cI_2)$, for some $(\sigma_{\bf l}^2,\lambda_2) \in
\db{\cL_2 : \Gamma, \cI_2}$, we have
$$
(\delta(\sigma_{\bf l}^1), \history(\lambda_1)) \sqsubseteq 
(\delta(\sigma_{\bf l}^2), \history(\lambda_2)),
$$
which implies $\delta(\sigma_{\bf l}^2) \preceq \delta(\sigma_{\bf l}^1)$.
By Lemma~\ref{cor-rearr}, $\lambda_2$ can be transformed
into a trace $\lambda'_2$ such that 
$$((\sigma_{\bf l}^2,\lambda'_2)\in \db{\cL_2 : \Gamma,\cI_2}) \ \wedge\
(\history(\lambda'_2)=\history(\lambda_1) = \history(\kappa)).
$$
Since $\delta(\sigma_{\bf l}^2) \preceq \delta(\sigma_{\bf l}^1)$ and
$(\sigma_{\bf c}*\sigma_{\bf l}^1)\fdef$, we have $(\sigma_{\bf c}*\sigma_{\bf l}^2)\fdef$.  We then use
Corollary~\ref{lemma-clientlocal2} to compose the library-local trace $\lambda'_2$ with
the client-local one $\kappa$ into a trace $\tau_2$ such that
$$
((\sigma_{\bf c}*\sigma_{\bf l}^2,\tau_2) \in \db{\Cc(\cL_2), \cI * \cI_2})\ \wedge\
(\client(\tau_2)=\erase(\kappa) = \client(\tau_1)).\eqno{\qEd}
$$\vspace{-2 pt}

\noindent The above proof scheme can be described mnemonically as `decompose, rearrange,
compose'. We reuse its first two steps to prove Theorem~\ref{thm-spec}.

\paragraph{\em Proof of Theorem~\ref{thm-spec}.} Take $(\sigma, \tau_1)\in
\db{\fillin{\cL_1}{\Cc}, \cI * \cI_1}$.  Like in the proof of
Theorem~\ref{thm2}, we apply Corollary~\ref{lemma-clientlocal} to generate
$(\sigma_{\bf l}^1,\lambda_1)\in \db{\cL_1, \cI_1}$ and $(\sigma_{\bf c},\kappa) \in
\db{\Gamma \vdash \Cc,\cI}$ such that~(\ref{eq-decomp}) holds.
Since $(\cL_1 : \Gamma, \cI_1) \sqsubseteq \cH_2$, for some $(l_2,H_2) \in
\cH_2$, we have  
$$
(l_2 \preceq \delta(\sigma_{\bf l}^1))\ \wedge\
((\delta(\sigma_{\bf l}^1), \history(\kappa))=
(\delta(\sigma_{\bf l}^1), \history(\lambda_1)) \sqsubseteq 
(l_2,H_2)).
$$
Then by Lemma~\ref{thm}, $\kappa$ can be transformed into a trace $\kappa'$, such that
$$
(\sigma_{\bf c},\kappa') \in \db{\Gamma \vdash\Cc,\cI}
\ \wedge\
\history(\kappa') = H_2\ \wedge\
\kappa \sim \kappa'
$$
Since $\client(\tau_1) = \erase(\kappa)$, we thus have
$\client(\tau_1) \sim \erase(\kappa')$.
Furthermore, since $l_2 \preceq \delta(\sigma_{\bf l}^1)$ and $(\sigma_{\bf c} *
\sigma_{\bf l}^1)\fdef$, we have $(\delta(\sigma_{\bf c}) \circ l_2)\fdef$. Hence, $\kappa'$
and $l_2$ are the required trace and footprint.\qed

\subsection{Establishing and Using Linearizability with Ownership
  Transfer\label{sec:example}}

Our preliminary investigations show that linearizability with ownership transfer
can be established by generalising existing proof systems for proving classical
linearizability based on separation logic~\cite{rgsep-thesis}. The details of
such a generalisation are out of the scope of this paper; we plan to report on
it in the future.

The Abstraction Theorem is not just a theoretical result: it enables
compositional reasoning about complex concurrent algorithms that are challenging
for existing verification methods. For example, the theorem can be used to
justify Vafeiadis's compositional proof~\cite[Section 5.3]{rgsep-thesis} of the
multiple-word compare-and-swap (MCAS) algorithm implemented using an auxiliary
operation called RDCSS~\cite{mcas} (the proof used an abstraction of the kind
enabled by Theorem~\ref{thm2} without justifying its correctness). If the MCAS
algorithm were verified together with RDCSS, its proof would be extremely
complicated. Fortunately, we can consider MCAS as a client of RDCSS, with the two
components performing ownership transfers between them. The Abstraction Theorem
then makes the proof tractable by allowing us to verify the linearizability of
MCAS assuming an atomic specification of the inner RDCSS algorithm.

\section{Proof of the Rearrangement Lemma\label{sec:rearr-proof}}

We only give the proof of Lemma~\ref{cor-rearr}, as that of Lemma~\ref{thm} is
completely symmetric.

The proof transforms $\lambda'$ into $\lambda$ by repeatedly swapping adjacent actions
in it according to a certain strategy to make the history of the trace equal to
$H$. The most subtle place in the proof is swapping
$$
(t_1,\ret\ m_1(\sigma_1))\, (t_2,\call\ m_2(\sigma_2))
$$
to yield 
$$
(t_2,\call\ m_2(\sigma_2))\, (t_1,\ret\ m_1(\sigma_1)),
$$
where $t_1\not=t_2$. This case is subtle for the following reason.
Let the state of the library $\cL$ before the return action be $\theta$; then
$(\theta \diff \sigma_1)\fdef$ and the state of the library after executing the
return and the call is $(\theta \diff \sigma_1)*\sigma_2$. For the swapping to
be  possible, we need $(\theta * \sigma_2)\fdef$; then by
Proposition~\ref{prop-diff}
\be\label{theta-swap}
(\theta \diff \sigma_1) * \sigma_2 = (\theta * \sigma_2)\diff \sigma_1,
\ee
which can be used to establish that the resulting trace is still produced by $\cL$.
However, $(\theta * \sigma_2)\fdef$ is not
guaranteed if the history $H$ is arbitrary. For example, take $\Sigma =
\RAM$ and let $H$ and $H'$ be defined by Figures~\ref{fig:cex}(a) and~\ref{fig:cex}(b).
Since $H$ is unbalanced, it cannot be produced by any library, and 
hence, we cannot swap 
$$
(1, \ret\ m_1([10:0]))\,(2, \call\ m_2([10:0]))
$$
in $H'$. In our proof we use the fact that the history $H$ is balanced to
show that a situation in which we cannot swap a return followed by a call while
transforming $\lambda'$ into $\lambda$ cannot happen. This is non-trivial, as the
problematic situation can potentially happen midway through the transformation.
We only know that the target history $H$ of $\lambda$ is balanced, but this does not
straightforwardly imply that the histories of the intermediate traces obtained
while transforming $\lambda'$ into $\lambda$ are, since these histories might be quite
different from $H$. Inferring their balancedness from that of $H$ represents the
most challenging part of the proof.



We therefore first do the proof under an assumption that allows swapping a
return followed by a call easily and consider the general case later. This lets
us illustrate the overall idea of the proof, which is then reused in the
additional part of the proof dealing with the challenge presented by the general
case.  Namely, we make the following assumption:
\be\label{assm}
\begin{minipage}{13cm}
  $\Sigma = \RAM$ and for any $\zeta \in \db{\cL : \Gamma}\sigma'$ and
  interface actions $\psi_1 = (t_1, \_\ \_(\sigma_1))$ and $\psi_2 = (t_2, \_\
  \_(\sigma_2))$ in $\zeta$, if $t_1\not=t_2$, then $\dom(\sigma_1) \cap
  \dom(\sigma_2) = \emptyset$.
\end{minipage}
\ee 
For example, this holds when states transferred between the client and the
library are always thread-local.  It is easy to check that in $\RAM$, if
$(\theta \diff \sigma_1)\fdef$, $((\theta \diff \sigma_1) * \sigma_2)\fdef$ and
$(\sigma_1 * \sigma_2)\fdef$, then $(\theta * \sigma_2)\fdef$
and thus~(\ref{theta-swap}) holds. Hence,~(\ref{assm}) allows us to justify
swapping a return followed by a call in a trace easily.  We now proceed to prove
Lemma~\ref{cor-rearr} under this assumption.  In our proof, we use the
assumption in a single place, which we note explicitly; the rest of the proof is
independent from it.

Below we sometimes write $\sqsubseteq_\rho$ instead of $\sqsubseteq$ to make the
bijection $\rho$ used to establish the relation between histories in
Definition~\ref{lin} explicit. For a bijection $\rho$ between histories $H$ and
$H'$, we write $\id_k(\rho)$ if $\rho$ is an identity on the first $k$ actions
in $H$.

Take $\sigma,\sigma'\in\State$ and consider a trace $\lambda' \in
\db{\cL : \Gamma}\sigma'$. Assume histories $H, H'$ such that $\history(\lambda')=H'$ and
$(\delta(\sigma), H) \sqsubseteq (\delta(\sigma'), H')$, so that $H$ is balanced
from $\delta(\sigma)$ and $H'$ from $\delta(\sigma')$. We prove that there
exists a trace $\lambda \in \db{\cL : \Gamma}\sigma'$ such that $\history(\lambda)=H$.  To this
end, we define a finite sequence of steps that transforms $\lambda'$ into a such a
trace $\lambda$. The main idea of the transformation is to make progressively longer
prefixes of the trace have histories coinciding with prefixes of $H$. Namely,
the transformation is done in stages, and on stage $k = 0,1,2,\ldots, |H|$ we
obtain a trace $\alpha_k \in \db{\cL : \Gamma}\sigma'$, where $\alpha_0=\lambda'$. Every one
of these traces is such that for some prefix $\beta_k$ of $\alpha_k$ we have:
\begin{gather*}
\history(\beta_k) = H\pref_k; 
\\
\exists \rho.\, ((\delta(\sigma), H) \sqsubseteq_\rho 
(\delta(\sigma'), \history(\alpha_k))) \wedge \id_k(\rho);
\\
\forall j.\, 0 \le j < k \implies 
(\beta_j \mbox{ is a prefix of } \beta_k).
\end{gather*}
We let $\beta_0 = \varepsilon$, so that the above conditions are initially satisfied.
Thus, during the transformation, 
progressively longer prefixes $\beta_k$ of $\alpha_k$ have histories
coinciding with prefixes of $H$, while the linearizability relation between the
history $H$ and that of $\alpha_k$ is preserved. We then take $\alpha_{|H|}$ as
the desired trace $\lambda$.

The trace $\alpha_{k+1}$ is constructed from the trace $\alpha_k$ by applying
the following lemma for $\lambda_1=\beta_k$, $\lambda_1\lambda_2=\alpha_k$,
$H_1=H\pref_k$, $H_1\psi H_2=H$, $\alpha_{k+1} = \lambda_1\lambda'_2\lambda''_2$
and $\beta_{k+1} = \lambda_1\lambda'_2$.

\begin{lem}\label{lemma}
  Assume~(\ref{assm}) holds.  Consider a history $H_1\psi H_2$ and a
  trace $\lambda_1\lambda_2\in \db{\cL : \Gamma}\sigma'$ such that 
\begin{gather}
\history(\lambda_1) = H_1;\label{1}\\
\exists \rho.\,  ((\delta(\sigma), H_1\psi H_2) \sqsubseteq_\rho 
(\delta(\sigma'), \history(\lambda_1\lambda_2))) \wedge \id_{|H_1|}(\rho).\label{2}
\end{gather}
Then there exist traces $\lambda'_2$ and
$\lambda''_2$ such that $\lambda_1\lambda'_2\lambda''_2 \in \db{\cL : \Gamma}\sigma'$ and
\begin{gather}
\history(\lambda_1\lambda'_2)=H_1\psi;\label{4}\\
\exists \rho'.\,  ((\delta(\sigma), H_1\psi H_2) \sqsubseteq_{\rho'}
(\delta(\sigma'), \history(\lambda_1\lambda'_2\lambda''_2)))
\wedge \id_{|H_1\psi|}(\rho').\label{5}
\end{gather}
\end{lem}

To prove Lemma~\ref{lemma}, we convert $\lambda_1\lambda_2$ into $\lambda_1\lambda'_2\lambda''_2$ by
swapping adjacent actions in the trace a finite number of times while preserving
its properties of interest. These transformations are described by the following
proposition, which formalises the closure properties of the library-local
semantics we alluded to in Section~\ref{sec:rearr}.
\begin{prop}\label{prop}
Let $\cL : \Gamma$ be safe at $\sigma_0$ and consider
$\zeta \in \db{\cL : \Gamma}\sigma_0$ and a history $S$ 
such that $S \sqsubseteq_{\rho} \history(\zeta)$.
Then swapping any two adjacent actions $\varphi_1\varphi_2$ in $\zeta$
executed by different threads such that
{\renewcommand{\labelenumi}{(\theenumi)}
\begin{enumerate}[\em(i)]
\item
$\varphi_1\in\Act - \ERetAct$, $\varphi_2\in\ECallAct$; or
\item
$\varphi_1 \in \ERetAct$, $\varphi_2\in\ECallAct$, 
$\varphi_2$ precedes $\varphi_1$ in $S$, and~(\ref{assm}) holds; or
\item
$\varphi_1\in \ERetAct$, $\varphi_2\in\Act - \ECallAct$
\end{enumerate}}%
\noindent yields a trace $\zeta' \in \db{\cL : \Gamma}\sigma_0$ such that $S
\sqsubseteq_{\rho'} \history(\zeta')$ for the bijection $\rho'$ defined as
follows.  If $\varphi_1\not\in\ECallRetAct$ or $\varphi_2\not\in\ECallRetAct$,
then $\rho'=\rho$.  Otherwise, let $i$ be the index of $\varphi_1$ in
$S$.  Then $\rho'(i+1)=\rho(i)$, $\rho'(i)=\rho(i+1)$ and
$\rho'(k)=\rho(k)$ for $k\not\in \{i,i+1\}$.
\end{prop}
Since, in the library-local semantics, the library gains state at a call and
gives it up at a return, intuitively, the transformation in the proposition
allows the library to gain state earlier (i, ii) and give it up later (iii). The
assumption that $\varphi_2$ precede $\varphi_1$ in case (ii) is needed to ensure
that the transformation does not violate the linearizability relation. The proof
of case (ii) is the only place where the assumption~(\ref{assm}) is used.

\paragraph{\em Proof sketch for Proposition~\ref{prop}.}
Consider $\zeta = \zeta_1 \varphi_1 \varphi_2 \zeta_2 \in \db{\cL :
  \Gamma}\sigma_0$ and let $\zeta' = \zeta_1 \varphi_2 \varphi_1 \zeta_2$. The
proof of the required linearizability relationship is trivial. It therefore
remains to show that $(\_, \zeta') \in \db{\cL : \Gamma}\sigma_0$. We know that
for some $\alpha$ we have $\alpha \in \dbp{\cL}$ and $(\_, \zeta) \in
\db{\alpha}\sigma_0$. Let $\alpha = \alpha_1 \varphi'_1 \varphi'_2 \alpha_2$ and
$\alpha' = \alpha_1 \varphi'_2 \varphi'_1 \alpha_2$, where $\varphi'_1$ and
$\varphi'_2$ correspond to $\varphi_1$ and $\varphi_2$. It is easy to see that
$\alpha' \in \dbp{\cL}$. It therefore remains to show that $\zeta' \in
\db{\alpha'}\sigma_0$. The proof proceeds by case analysis on the kind of
actions $\varphi_1$ and $\varphi_2$. The justification of the case when
$\varphi_1$ is a return and $\varphi_2$ is a call follows from~(\ref{assm}) by
the argument given earlier. Out of the remaining cases, we only consider a
single illustrative one: $\varphi_1 = (t_1, c)$ and $\varphi_2 = (t_2, \call\
m_2(\sigma))$ for $t_1\not=t_2$.

Assume $(\theta', \zeta_1\varphi_1\varphi_2) \in
\db{\alpha_1\varphi'_1\varphi'_2}\sigma_0$. Then for some $\theta$ we have
$(\theta, \zeta_1) \in \db{\alpha_1}\sigma_0$, $f_c^{t_1}(\theta) \not= \top$
and $\theta' \in f_c^{t_1}(\theta)*\{\sigma\}$. By the Footprint Preservation
property, we get $(\theta*\sigma)\fdef$. Then by the Strong Locality property,
$$
\theta' \in f_c^{t_1}(\theta)*\{\sigma\} = f_c^{t_1}(\theta*\sigma).
$$
Hence, $(\theta', \zeta_1\varphi_2\varphi_1) \in
\db{\alpha_1\varphi'_2\varphi'_1}\sigma_0$.  
Thus, Footprint Preservation and Strong Locality guarantee that a call can
be safely executed earlier than a primitive command.
\qed

\paragraph{\em Proof of Lemma~\ref{lemma}.}
From~(\ref{1}) and~(\ref{2}) it follows that $\lambda_2=\lambda_3\psi\lambda_4$ for some
traces $\lambda_3$ and $\lambda_4$, where $\rho$ maps $\psi$ in $H_1\psi H_2$ to the
$\psi$ action shown in $\lambda_2$; see Figure~\ref{fig:diag}. We consider two
cases.

\begin{figure}
\includegraphics[scale=.32, trim= 0 18cm 0 0cm]{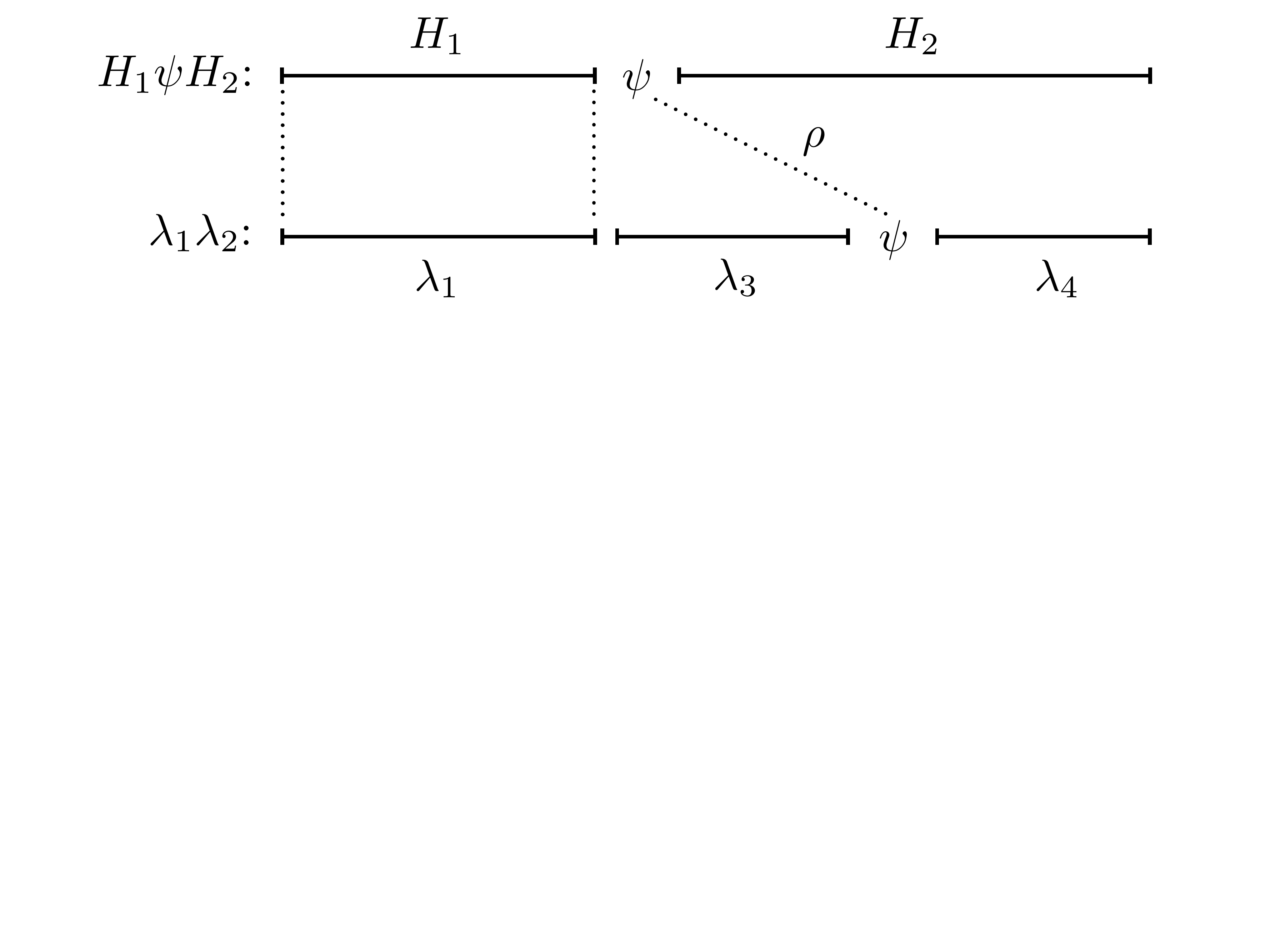}
\caption{Illustration of the proof of Lemma~\ref{lemma}}
\label{fig:diag}
\end{figure}

\medskip

1. $\psi\in\ECallAct$. Let $t$ be the thread executing $\psi$.
By~(\ref{2}) we have $(H_1\psi H_2)|_t = (\history(\lambda_1\lambda_2))|_t$. Then, since
$\lambda_1\lambda_2$ is a library trace,~(\ref{1}) implies that there are no actions by
thread $t$ in $\lambda_3$. Furthermore, for any return action $\varphi$ in $\lambda_3$,
the action in $H_1\psi H_2$ corresponding to it according to $\rho$ is in
$H_2$. Thus, we can move the action $\psi$ to the position between $\lambda_1$ and
$\lambda_3$ by swapping it with adjacent actions a finite number of times as
described in Proposition~\ref{prop}(i, ii).  As a result, we obtain the trace
$\lambda_1\psi\lambda_3\lambda_4 \in \db{\cL : \Gamma}\sigma'$.  Conditions~(\ref{4})--(\ref{5}) then
follow from Proposition~\ref{prop}(i, ii) for $\lambda'_2=\psi$ and
$\lambda''_2=\lambda_3\lambda_4$.

\medskip

2. $\psi\in\ERetAct$. Assume that $\lambda_3$ contains a call action
$\varphi$, so that it precedes the return action $\psi$ in
$\history(\lambda_1\lambda_2)$. Then by~(\ref{1}) and~(\ref{2}) the action in $H_1\psi
H_2$ corresponding to $\varphi$ according to $\rho$ is in $H_2$ and thus follows
$\psi$ in $H_1\psi H_2$. This violates the preservation of the order of
non-overlapping method invocations required by~(\ref{2}). Hence, there are no
call actions in $\lambda_3$. Since $\lambda_1\lambda_2$ is a library trace, 
this implies that for any action $\varphi=(t,\ret\
\_)$ in $\lambda_3$ there are no actions by the thread $t$ in $\lambda_3$ following
$\varphi$.  Thus, we can move all return actions in the subtrace $\lambda_3$ of
$\lambda_1\lambda_2$ to the position between $\psi$ and $\lambda_4$ by swapping them with
adjacent actions a finite number of times as described in
Proposition~\ref{prop}(iii).  We thus obtain the trace
$\lambda_1\lambda'_3\psi\lambda'_4\lambda_4 \in \db{\cL : \Gamma}\sigma'$, where $\lambda'_4$ consists of
all return actions in $\lambda_3$, and $\lambda'_3$ of the rest of actions in the
subtrace; in particular, $\lambda'_3$ does not contain any interface actions.
Conditions~(\ref{4})--(\ref{5}) then follow from Proposition~\ref{prop}(iii) for
$\lambda'_2=\lambda'_3\psi$ and $\lambda''_2=\lambda'_4\lambda_4$.  \qed

\smallskip

This completes the proof of Lemma~\ref{cor-rearr} under the
assumption~(\ref{assm}). Let us now lift this assumption and consider the only
place in the proof of the lemma that relies on it---that when we swap a return
followed by a call using Proposition~\ref{prop}(ii) in case 1 in the proof of
Lemma~\ref{lemma}. Let us now identify precise conditions under which this
situation happens. Let $\psi = (t_1,\call\ m_1(\sigma_1))$ and let the adjacent
return action with which we are trying to swap it be 
$(t_2,\ret\ m_2(\sigma_2))$.
Let $S^1$ be the target history $H_1\psi H_2$ from Lemma~\ref{lemma} and $S^2$
be the history of the trace in which we are trying to swap the return and the call.
Then the two histories are of the form 
\be\label{2hist}
\begin{array}{r@{\ }c@{\ }l}
S^1 &=& S\, (t_1,\call\ m_1(\sigma_1)) \,S_1\, (t_2,\ret\ m_2(\sigma_2))\, S_2;\\[2pt]
S^2 &=& S S'_1\, (t_2,\ret\ m_2(\sigma_2))\, (t_1,\call\ m_1(\sigma_1))\, S'_2
\end{array}
\ee
for some $S, S_1, S_2, S'_1, S'_2$.
Furthermore, from the conditions of Lemma~\ref{lemma}, we have:
{\renewcommand{\labelenumi}{(\theenumi)}
\begin{enumerate}[(i)]
\item\label{cond0} $t_1 \not= t_2$;
\item
$S^1$ and $S^2$ are balanced from some $l_1$ and $l_2$, respectively, such
that $l_2 \preceq l_1$; and
\item $S^1 \sqsubseteq_\rho S^2$, where $\id_{|S|}(\rho)$ 
  and $\rho$ maps $(t_1,\call\ m_1(\sigma_1))$ and $(t_2,\ret\ m_2(\sigma_2))$
  in $S^1$ to the corresponding actions shown in $S^2$.
\end{enumerate}
}

As the following proposition shows, we can always do the desired transformation
if the history
\be\label{hist1}
S S'_1\, (t_1,\call\ m_1(\sigma_1))\, (t_2,\ret\ m_2(\sigma_2)) \,S'_2
\ee
resulting from swapping the return and the call in $S^2$ is balanced from $l_2$.
\begin{prop}\label{cor-call-ret}
Let $\cL : \Gamma$ be safe at $\sigma_0$. Consider traces
$$
\zeta =\zeta_1\,(t_2,\ret\ m_2(\sigma_2))\, (t_1,\call\ m_1(\sigma_1))\,\zeta_2 \in
\db{\cL : \Gamma}\sigma_0
$$
and
$$
\zeta ' =\zeta_1\, (t_1,\call\ m_1(\sigma_1))\, (t_2,\ret\ m_2(\sigma_2))\, \zeta_2.
$$
If $\history(\zeta')$ is balanced from $\delta(\sigma_0)$, then $\zeta' \in
\db{\cL : \Gamma}\sigma_0$.
\end{prop}


Thus, the only problematic case we have is when the history~(\ref{hist1}) is not
balanced from $l_2$. We summarise all the conditions under which such case can
happen in the following definition.
\begin{defi}\label{def-confl}
Histories $S^1$ and $S^2$ of the form~(\ref{2hist}) are {\bf \em conflicting} if
the conditions (i)--(iii) above are satisfied and the history~(\ref{hist1}) is
not balanced from $l_2$.
\end{defi}
Given Proposition~\ref{cor-call-ret}, the only case when the transformation in
the proof of Lemma~\ref{lemma} can fail to convert the trace is when $H_1\psi
H_2$ and the history currently being transformed are conflicting. Thus, with the
assumption~(\ref{assm}) lifted, Lemma~\ref{lemma} turns into
\begin{lem}\label{lemma2}
  Consider a history $H_1\psi H_2$, and a trace $\lambda_1\lambda_2 \in
  \db{\cL : \Gamma}\sigma'$ such that (\ref{1}) and (\ref{2}) hold. Then
  either $H_1\psi H_2$ and another history composed of actions from
  $\history(\lambda_1\lambda_2)$ are conflicting, or there exist traces $\lambda'_2$ and
  $\lambda''_2$ such that $\lambda_1\lambda'_2\lambda''_2 \in \db{\cL : \Gamma}\sigma'$
  and~(\ref{4}) and~(\ref{5}) hold.
\end{lem}
We now show that no conflicting pairs of histories exist, hence guaranteeing
that Lemma~\ref{lemma2} can always be used to construct $\alpha_{k+1}$ from
$\alpha_k$ in transforming $\lambda'$ into $\lambda$. This completes the proof of
Lemma~\ref{cor-rearr} in the general case.

We first discuss the main idea of the proof.  The history $S^1$ in~(\ref{2hist})
is similar to~(\ref{hist1}) in that the call precedes the return. We would like
to use the fact that $S^1$ is balanced to prove that so is~(\ref{hist1}),
thereby yielding a contradiction. As we noted at the beginning of this section,
this is not straightforward due to the differences in the form of the histories
$S^1$ and $S^2$ other than the precedence of the two call and return actions. We
resolve this problem by adjusting the strategy we used above to transform
$\lambda'$ into $\lambda$ under the assumption~(\ref{assm}) to iron out the
differences between $S^1$ and $S^2$. In particular, we use a variant of the
transformation that, when the process of moving the call action $\psi$ to the
left in $\lambda_1\lambda_2$ gets stuck (Figure~\ref{fig:diag}), leaves the
corresponding action in $H_1\psi H_2$ unmatched and continues bringing the rest
of the trace $\lambda_1\lambda_2$ in sync with the target history.

\begin{lem}\label{call-ret}
There are no conflicting pairs of histories.
\end{lem}
\proof
Consider histories $S^1$ and $S^2$ satisfying the conditions in
Definition~\ref{def-confl}. Since $S^2$ is balanced from $l_2$, 
$$
\db{S S'_1\, (t_2,\ret\ m_2(\sigma_2))}^\sharp l_2
=(\db{S S'_1}^\sharp l_2)\fdiff\delta(\sigma_2)
$$
is defined. Assume
$$
\db{S S'_1\, (t_1, \call\ m_1(\sigma_1))}^\sharp l_2=
(\db{S S'_1}^\sharp l_2)\circ \delta(\sigma_1)
$$
is defined. Since $(\db{S S'_1}^\sharp l_2)\fdiff \delta(\sigma_2)$ is defined, by 
Proposition~\ref{prop-delta}, we have:
\begin{multline*}
\db{S S'_1\, (t_1, \call\ m_1(\sigma_1))\, (t_2,\ret\ m_2(\sigma_2))}^\sharp l_2 
= ((\db{S S'_1}^\sharp l_2) \circ  \delta(\sigma_1)) \fdiff \delta(\sigma_2) = {}
\\
 ((\db{S S'_1}^\sharp l_2) \fdiff  \delta(\sigma_2)) \circ  \delta(\sigma_1)
= \db{S S'_1\, (t_2,\ret\ m_2(\sigma_2))\, (t_1, \call\ m_1(\sigma_1))}^\sharp l_2.
\end{multline*}
  Then~(\ref{hist1}) is balanced from $l_2$, contradicting our
assumptions. Hence, ${((\db{S S'_1}^\sharp l_2)\circ \delta(\sigma_1))}\fundef$.

A call action in $S'_1$ cannot be in $S_2$: in this case it would follow
$(t_2,\ret\ m_2(\sigma_2))$ in $S^1$, but precede it in $S^2$, contradicting
$S^1 \sqsubseteq S^2$.
Hence, all call actions in $S'_1$ are in $S_1$. Let $S_1 = S_3 S_4$, where
$S_3$ is the minimal prefix of $S_1$ containing all call actions from $S'_1$. Then
$$
\begin{array}{r@{\ }c@{\ }l}
S^1 &=& S\, (t_1,\call\ m_1(\sigma_1))\, S_3 S_4\, (t_2,\ret\ m_2(\sigma_2)) \,S_2;\\[2pt]
S^2 &=& S S'_1\, (t_2,\ret\ m_2(\sigma_2))\, (t_1,\call\ m_1(\sigma_1)) \,S'_2.
\end{array}
$$
If $S_3$ is non-empty, any return action in it precedes its last call action,
which is also in $S'_1$. Since $S^1 \sqsubseteq S^2$, such a return action
also has to be in $S'_1$. Thus, all return actions in $S_3$ are in $S'_1$.


The traces $S^1$ and $S^2$ are of the following more general form, obtained by
letting $S_0 = S\, (t_1,\call\ m_1(\sigma_1))$ and $S'_0 = S$:
$$
\begin{array}{r@{\ }c@{\ }l}
S^1 &=& S_0 S_3 S_4\, (t_2,\ret\ m_2(\sigma_2)) \,S_2;\\[2pt]
S^2 &=& S'_0 S'_1\, (t_2,\ret\ m_2(\sigma_2))\, (t_1,\call\ m_1(\sigma_1)) \,S'_2,
\end{array}
$$
where
\begin{iteMize}{$\bullet$}
\item $S^1$ and $S^2$ are balanced from some $l_1$ and $l_2$, respectively,
  such that $l_2 \preceq l_1$;
\item
$S_0$ and $S'_0$ are identical, except $S_0$ may have some extra call actions;
\item
$S^1 \sqsubseteq_\rho S^2$;
\item
$\rho^{-1}$ maps all call actions in $S'_1$ to actions in $S_3$;
\item
$\rho$ maps all return actions in $S_3$ to actions in $S'_1$;
\item $\rho^{-1}$ maps actions in $S'_0$ to those in $S_0$, in particular $(t_1,\call\
  m_1(\sigma_1))$ to an action in $S_0$, and $(t_2,\ret\ m_2(\sigma_2))$ to the
  same action shown in $S^1$; and
\item ${((\db{S'_0 S'_1}^\sharp l_2)\circ \delta(\sigma_1))}\fundef$.
\end{iteMize}
We denote this form by ({\bf F}). The additional call actions in $S_0$ are the
ones for which the transformation in Lemma~\ref{lemma} failed.  The conditions
relating $S_3$ and $S'_1$ imply that $S_3$ may have more calls than $S'_1$, and
$S'_1$ more returns than $S_3$. Thus, intuitively, $S_0S_3$ gains more state
than $S'_0S'_1$, including that transferred by $(t_1,\call\ m_1(\sigma_1))$, and
$S'_0S'_1$ gives up more than $S_0S_3$. In the following, we use this and the
fact that $S^1$ is balanced from a bigger footprint than $S^2$ to show that
${((\db{S'_0 S'_1}^\sharp l_2)\circ \delta(\sigma_1))}\fdef$, thereby yielding a
contradiction.

To this end, we describe a process that transforms the histories $S^1$ and
$S^2$ into another pair of histories satisfying the conditions above, but such
that $S_3$ is strictly smaller. Repeatedly applying this process, we can make
$S_3$ empty, obtaining histories satisfying ({\bf F}):
\be\label{hist-temp}
\begin{array}{l}
S_0 S_4\, (t_2,\ret\ m_2(\sigma_2)) \,S_2;\\[2pt]
S'_0 S'_1\, (t_2,\ret\ m_2(\sigma_2))\, (t_1,\call\ m_1(\sigma_1)) \,S'_2.
\end{array}
\ee
In particular, ${((\db{S'_0 S'_1}^\sharp l_2)\circ \delta(\sigma_1))}\fundef$.
Before describing the transformation process, we show that, given the above pair
of histories, we can obtain a contradiction. We use the following simple
proposition, proved in Appendix~\ref{sec:proof-foot}.
\begin{prop}\label{prop-foot}
  Assume $S$ is identical to $S'$, except it may have extra calls, and $S$ and $S'$
  are balanced from $l_1$ and $l_2$, respectively, such that $l_2 \preceq
  l_1$. Then the $\circ$-combination $l_c$ of footprints of states transferred
  at the extra call actions in $S$ is defined, $S'$ is balanced from $l_1$ and
$$
{\db{S}^\sharp l_1 = (\db{S'}^\sharp l_1) \circ l_c}
\;\wedge\;
{\db{S'}^\sharp l_2 \preceq \db{S'}^\sharp l_1}.
$$
\end{prop}

\noindent Consider the histories in~(\ref{hist-temp}).
Since all calls from $S'_1$ are in $S_3 = \varepsilon$, $S'_1$ contains only
returns. Since the histories are balanced from $l_1$ and $l_2$, respectively,
$\db{S_0}^\sharp l_1$ and $\db{S'_0}^\sharp l_2$ are defined.  The history $S_0$
is identical to $S'_0$, except it may have extra calls. By
Proposition~\ref{prop-foot}, the $\circ$-combination of footprints of states
transferred at the extra call actions in $S_0$ is defined. Since an action
$(t_1,\call\ m_1(\sigma_1))$ is in $S_0$, but not in $S'_0$, this combination is
of the form $\delta(\sigma_1)\circ l_c$ for some $l_c$; hence,
$$
\db{S_0}^\sharp l_1 = (\db{S'_0}^\sharp l_1)\circ \delta(\sigma_1)\circ l_c
\;\wedge\; \db{S'_0}^\sharp l_2 \preceq \db{S'_0}^\sharp l_1.
$$
Therefore, $(\db{S'_0}^\sharp l_2)\circ \delta(\sigma_1)$ is defined.
Since $S'_1$ contains only return actions, 
$$\db{S'_0S'_1}^\sharp l_2 = (\db{S'_0}^\sharp l_2)\fdiff l',$$
where $l'$ is the $\circ $-combination of the footprints of states transferred
at these actions. This implies
$$
\db{S'_0}^\sharp l_2  = (\db{S'_0S'_1}^\sharp l_2)\circ l',
$$
where both expressions are defined. But then so is
$$
(\db{S'_0}^\sharp l_2)\circ \delta(\sigma_1)=
(\db{S'_0S'_1}^\sharp l_2)\circ l'\circ \delta(\sigma_1).
$$
Hence, $(\db{S'_0S'_1}^\sharp l_2)\circ \delta(\sigma_1)$ is defined, contradicting
the opposite fact established above. This contradiction implies that a
conflicting pair of histories does not exist.

\bigskip

Now assume arbitrary histories $S^1$ and $S^2$ satisfying ({\bf F}):
$$
\begin{array}{r@{\ }c@{\ }l}
S^1 &=& S_0 S_3 S_4\, (t_2,\ret\ m_2(\sigma_2)) \,S_2;\\[2pt]
S^2 &=& S'_0 S'_1\, (t_2,\ret\ m_2(\sigma_2))\, (t_2,\call\ m_1(\sigma_1)) \,S'_2,
\end{array}
$$
We show that from these we can construct another pair of histories
satisfying ({\bf F}), but with $S_3$ strictly smaller. We use the same
transformation as in the proof of Lemma~\ref{lemma}. When this transformation
gets stuck, we obtain another pair of histories of the form ({\bf F}), but again
with a smaller $S_3$.

Let us make a case split on the next action in $S_3$.
\begin{iteMize}{$\bullet$}
\item $S_3 = (t, \call\ m(\sigma))\, S_5$, such that the action corresponding to
  $(t, \call\ m(\sigma))$ according to $\rho$ is not in $S'_1$. 
  In this case we let $S_0 := S_0\, (t, \call\ m(\sigma))$ and $S_3:=S_5$. Thus, the
  call action $(t, \call\ m(\sigma))$ unmatched in $S'_1$ becomes part of $S_0$.
\item $S_3 = (t, \call\ m(\sigma))\, S_5$, such that the action
 corresponding to $(t, \call\ m(\sigma))$ according to $\rho$ is in
  $S'_1$. Let $S'_1 = S'_3\, (t, \call\ m(\sigma)) S'_4$, so that
$$
\begin{array}{r@{\ }c@{\ }l}
S^1 &=& S_0\, (t, \call\ m(\sigma))\, S_5 S_4\, (t_2,\ret\ m_2(\sigma_2))\, S_2;\\[2pt]
S^2 &=& S'_0 S'_3\, (t, \call\ m(\sigma))\, S'_4\, (t_2,\ret\ m_2(\sigma_2))\,
(t_2,\call\ m_1(\sigma_1)) \,S'_2.
\end{array}
$$
Using the transformations from case 1 in the proof of Lemma~\ref{lemma}, we can
try to move the action $(t, \call\ m(\sigma))$ to the position between 
$S'_0$ and $S'_3$ while
preserving the balancedness of the history. If this succeeds,
we construct a new pair of histories of the form ({\bf F}) by letting
$S_0 := S_0\, (t, \call\ m(\sigma))$, $S_3 := S_5$, 
$S'_0 := S'_0\, (t, \call\ m(\sigma))$ and $S'_1:=S'_3 S'_4$. 
Otherwise, we get a pair of conflicting histories, which are of the form ({\bf
  F}) but with a smaller $S_3$. Again, in this case the unmatched call action
$(t, \call\ m(\sigma))$ becomes part of $S_0$.
\item $S_3 = (t, \ret\ m(\sigma)) \,S_5$. Then the action corresponding to $(t,
  \ret\ m(\sigma))$ according to $\rho$ is also in $S'_1$, so that $S'_1 = S'_3\,
  (t, \ret\ m(\sigma))\, S'_4$:
$$
\begin{array}{r@{\ }c@{\ }l}
S^1 &=& S_0\, (t, \ret\ m(\sigma)) \,S_5 S_4\, (t_2,\ret\ m_2(\sigma_2)) \,S_2;\\[2pt]
S^2 &=& S'_0 S'_3\, (t, \ret\ m(\sigma))\, S'_4\, (t_2,\ret\ m_2(\sigma_2))\,(t_2,\call\ m_1(\sigma_1))\, S'_2.
\end{array}
$$
Using the transformations from case 2 in the proof of Lemma~\ref{lemma}, we can
move the return action to the position between $S'_0$ and $S'_3$ while
preserving the balancedness of the history. We thus obtain a pair of histories:
$$
\begin{array}{r@{\ }c@{\ }l}
S^1 &=& S_0\, (t, \ret\ m(\sigma))\, S_5 S_4\, (t_2,\ret\ m_2(\sigma_2))\, S_2;\\[2pt]
S^2 &=& S'_0\, (t, \ret\ m(\sigma))\, S'_3 S'_4\, (t_2,\ret\ m_2(\sigma_2))\,(t_2,\call\ m_1(\sigma_1)) \,S'_2.
\end{array}
$$
Then we can let $S_0 := S_0\, (t, \ret\ m(\sigma))$, $S_3 := S_5$, 
$S'_0 := S'_0\, (t, \ret\ m(\sigma))$ and $S'_1:=S'_3 S'_4$.\qed
\end{iteMize}


\section{Frame Rule for Linearizability\label{sec:frame}}

Libraries such as concurrent containers are used by clients to transfer the
ownership of data structures, but do not actually access their contents. We show
that for such libraries, the classical linearizability implies linearizability
with ownership transfer.
\begin{defi}\label{def:gamma2}
  A method specification $\Gamma' = \{\{r^m\}\ m\ \{s^m\} \mid m \in M\}$
  {\bf\em extends} a specification $\Gamma = \{\{p^m\}\ m\ \{q^m\} \mid m \in
  M\}$, if\/ $\forall t.\, r^m_t \subseteq p^m_t*\Sigma \wedge s^m_t \subseteq
  q^m_t*\Sigma$.
\end{defi}
For example, the method specification~(\ref{container}) of the stack in
Figure~\ref{fig:impl}(a) extends the following specification:
\be\label{container-val}
\begin{array}{l}
\{\exists x.\, \myarg_t \mapsto x \}\ 
{\tt push}\ \{\myarg_t \mapsto {\tt OK} \vee 
\myarg_t \mapsto {\tt FULL}\};
\\[2pt]
\{\exists y.\, \myarg_t \mapsto y\}\ {\tt pop}\ 
\{\exists x.\, \myarg_t \mapsto x\}.
\end{array}
\ee
According to this specification, {\tt push} just receives an arbitrary pointer
$x$ as a parameter; in contrast, the specification~(\ref{container})
additionally mandates that the object the pointer identifies be transferred to
the library. We now identify conditions under which the linearizability between
a pair of libraries satisfying $\Gamma$ entails that of the same libraries
satisfying an extended method specification $\Gamma'$. This yields a result
somewhat analogous to the frame rule of separation logic~\cite{lics02}.

We start by introducing some auxiliary definitions. We first define operations
for mapping between histories corresponding to extended and non-extended method
specifications. For the method specification $\Gamma$ from
Definition~\ref{def:gamma2}, we define operations $\llfloor \cdot
\rrfloor_{\Gamma}$ and $\llceil \cdot \rrceil_{\Gamma}$ on interface actions as
follows:
$$
\begin{array}{l@{\ }c@{\ }l@{}}
\llfloor(t, \call\ m(\sigma))\rrfloor_{\Gamma} &=& 
(t, \call\ m(\sigma \diff {(\sigma\diff p^m_t)}));
\\[2pt]
\llfloor(t, \ret\ m(\sigma))\rrfloor_{\Gamma} &=& 
(t, \ret\ m(\sigma \diff {(\sigma\diff q^m_t)}));
\\[2pt]
\llceil (t, \call\ m(\sigma)) \rrceil_{\Gamma}
&=& (t, \call\ m(\sigma\diff p^m_t));
\\[2pt]
\llceil (t, \ret\ m(\sigma)) \rrceil_{\Gamma}
&=& (t, \ret\ m(\sigma\diff q^m_t));
\end{array}
$$
otherwise, the result is undefined. Thus, $\llfloor \psi \rrfloor_{\Gamma}$
selects the part of the state in $\psi$ that is required by $\Gamma$ and
$\llceil \psi \rrceil_{\Gamma}$ the extra piece of state not required by it.
We then lift $\llfloor \cdot \rrfloor_\Gamma$ and $\llceil \cdot \rrceil_\Gamma$
to traces by applying them to every interface action.

Given a history $H_0$ produced by a library $\cL : \Gamma'$, we need to be able
to check that the library does not modify the extra pieces of state not required
by the original method specification $\Gamma$, which are given by $\llceil H_0
\rrceil_\Gamma$. To this end, we define an evaluation function similar to
$\dba{\cdot}^\sharp$ from Section~\ref{sec:observ}, which is meant to be applied
to $\llceil H_0 \rrceil_\Gamma$. For an interface action $\psi$ we define
$\dba{\psi} : \Sigma \rightarrow (\Sigma \cup \{\top\})$ as follows:
$$
\begin{array}{l@{\ }c@{\ }l}
\dba{(t,\call\;m(\sigma_0))}\sigma
&=& 
\mbox{if } {(\sigma*\sigma_0)\fdef} \mbox{ then } \sigma * \sigma_0
\mbox{ else } \top;
\\[2pt]
\dba{(t,\ret\;m(\sigma_0))}\sigma
&=& 
\mbox{if } {(\sigma \diff \sigma_0)\fdef} \mbox{ then } \sigma \diff \sigma_0 
\mbox{ else } \top.
\end{array}
$$
We then define the evaluation $\dba{H}: \Sigma \rightarrow
(\Sigma \cup \{\top\})$ of a history $H$ as follows: 
$$
\dba{\varepsilon} \sigma  \ =\   \sigma; \qquad
\dba{H \psi} \sigma  \ =\ 
\mbox{if } (\dba{H} \sigma \not= \top) 
\mbox{ then }\dba{\psi} (\dba{H} \sigma)
\mbox{ else }\top.
$$
Thus, if the evaluation does not fail, then the history respects the notion of
ownership and the pieces of state transferred to the library are returned to the
client unmodified.

\begin{thm}[Frame rule]\label{thm:frame}
  Assume 
\begin{enumerate}[\em(1)]
\item $\Gamma'$ extends $\Gamma$; 
\item
 for all $i \in \{1,2\}$, $\cL_i : \Gamma$ and $\cL_i: \Gamma'$
  are safe for $\cI_i$ and $\cI_i * I$, respectively; 
\item
  $(\cL_1 : \Gamma, \cI_1) \sqsubseteq (\cL_2 : \Gamma, \cI_2)$; and
\item for every $\sigma_0 \in I_1$, $\sigma'_0\in  I$
and $\lambda \in \db{\cL_1:\Gamma'}(\sigma_0*\sigma'_0)$, 
we have $\dba{\llceil\history(\lambda)\rrceil_{\Gamma}}\sigma'_0 \not= \top$.
\end{enumerate}
 Then $(\cL_1 : \Gamma', \cI_1*\cI) \sqsubseteq (\cL_2 : \Gamma', \cI_2*\cI)$.
\end{thm}

The theorem allows us to establish the linearizability relation with respect to
the extended specification $\Gamma'$ given the relation with respect to
$\Gamma$. This enables the use of the Abstraction Theorem for clients performing
ownership transfer. However, the theorem does not guarantee the safety of the
libraries with respect to $\Gamma'$ for free, because it is not implied by their
safety with respect to $\Gamma$. Intuitively, this is because $\Gamma'$ extends
both preconditions and postconditions in $\Gamma$; hence, not only does it
guarantee to the library that the client will provide extra pieces of state at
calls, but it also requires the library to provide (possibly different) extra
pieces of state at returns. For example, $\Gamma$ might assign the specification
$\{\myarg_t \mapsto \_\}\ m\ \{\myarg_t \mapsto \_\}$ to every method $m$, and
$\Gamma'$, the specification $\{\myarg_t \mapsto \_\}\ m\ \{\exists x.\,
\myarg_t \mapsto x * x \mapsto \_\}$. Unless a library already has all the
memory required by the postconditions in $\Gamma'$ in its initial state, it has
no way of satisfying $\Gamma'$. 

This situation is in contrast to the frame rule of separation
logic~\cite{lics02}, which guarantees the safety of a piece of code with respect
to an extended specification. However, the frame rule requires the latter
specification to extend both pre- and postconditions with the same piece of
state, so that the code returns it immediately after termination. In our
setting, a library can return the extra state to its client after a different
method invocation and, possibly, in a different thread.

Finally, condition (4) in Theorem~\ref{thm:frame} ensures that the extra memory
required by postconditions in $\Gamma'$ comes from the extra memory provided in
its preconditions and the extension of the initial state, not from the memory
transferred according to $\Gamma$.

It can be shown that for the library $\cL_1$ in Figure~\ref{fig:impl}(a), the
library $\cL_2$ in Figure~\ref{fig:spec}(a) and method specification $\Gamma$
defined by~(\ref{container-val}), we have $(\cL_1 : \Gamma, I_1) \sqsubseteq
(\cL_2 : \Gamma, I_2)$. It is also not difficult to prove (e.g., using
separation logic) that $\cL_1 : \Gamma'$ and $\cL_2 : \Gamma'$ are
safe. Condition (4) in Theorem~\ref{thm:frame} is satisfied, since the proof of
safety of $\cL_1 : \Gamma'$ would use only the extra state provided in the
preconditions of $\Gamma'$ to provide the extra state required by its
postconditions. Hence, by Theorem~\ref{thm:frame} we have $(\cL_1 : \Gamma',
I_1) \sqsubseteq (\cL_2 : \Gamma', I_2)$.

%

However, Theorem~\ref{thm:frame} is not applicable to the memory allocators in
Figures~\ref{fig:impl}(b) and~\ref{fig:spec}(b): since the allocator
implementation in Figure~\ref{fig:impl}(b) stores free-list pointers inside the
memory blocks, it is unsafe with respect to the variant of the method
specification~(\ref{alloc}) that does not transfer their ownership.

The proof of Theorem~\ref{thm:frame} relies on the following two lemmas, proved
in Appendices~\ref{proof:frame-out} and~\ref{proof:frame-in}, that convert
between library traces corresponding to extended and original method
specifications. The first lemma shows that for a trace $\lambda$ produced by $\cL :
\Gamma'$, the trace $\llfloor \lambda \rrfloor_{\Gamma}$ can be produced by $\cL :
\Gamma$. The safety of the library with respect to $\Gamma$ and condition (4)
from Theorem~\ref{thm:frame} guarantee that the smaller preconditions specified
by $\Gamma$ are enough for the library to execute safely and that the extra
pieces of state in $\Gamma'$ do not influence its execution.
\begin{lem}\label{frame-out}
  If $\lambda \in \db{\cL : \Gamma'}(\sigma_0*\sigma'_0)$, $\cL:\Gamma$ is safe at
  $\sigma_0$, and $\dba{\llceil \history(\lambda) \rrceil_{\Gamma}}\sigma'_0 \not=
  \top$, then $\llfloor \lambda \rrfloor_{\Gamma} \in \db{\cL:\Gamma}\sigma_0$.
\end{lem}

The other lemma gives conditions under which we can conclude that a trace $\lambda$
is produced by $\cL : \Gamma'$ given that $\llfloor \lambda \rrfloor_\Gamma$ is
produced by $\cL : \Gamma$.
\begin{lem}\label{frame-in}
  Assume $\llfloor \lambda \rrfloor_{\Gamma} \in\db{\cL:\Gamma}\sigma_0$,
  $(\sigma_0*\sigma'_0)\fdef$, $\history(\lambda)$ is balanced from
  $\delta(\sigma''_0*\sigma'_0)$ for $\delta(\sigma_0) \preceq
  \delta(\sigma''_0)$, and $\dba{\llceil \history(\lambda)
    \rrceil_{\Gamma}}\sigma'_0\not=\top$.  Then $\lambda \in
  \db{\cL:\Gamma'}(\sigma_0*\sigma'_0)$.
\end{lem}

\paragraph{\em Proof of Theorem~\ref{thm:frame}.}  Consider a trace $\lambda_1
\in \db{\cL_1 : \Gamma'}(\sigma_1*\sigma)$, where $\sigma_1 \in \cI_1$ and
$\sigma\in\cI$.  Then by (4) we have $\dba{\llceil \history(\lambda_1)
  \rrceil_{\Gamma}}\sigma\not= \top$, and hence by Lemma~\ref{frame-out} we have
$\llfloor \lambda_1 \rrfloor_{\Gamma} \in \db{\cL_1:\Gamma}\sigma_1$. Since
$(\cL_1 : \Gamma, \cI_1) \sqsubseteq (\cL_2 : \Gamma, \cI_2)$, for some
$\sigma_2\in\cI_2$ and $\lambda_2\in\db{\cL_2:\Gamma}\sigma_2$ we have
$$
(\delta(\sigma_1),\history(\llfloor \lambda_1 \rrfloor_{\Gamma})) \sqsubseteq
(\delta(\sigma_2),\history(\lambda_2)).
$$
By Lemma~\ref{cor-rearr}, there exists $\lambda'_2\in\db{\cL_2:\Gamma}\sigma_2$
such that $\history(\lambda'_2) = \history(\llfloor \lambda_1 \rrfloor_{\Gamma})$.  Let
$\lambda''_2$ be the trace $\lambda'_2$ with its interface actions replaced so that they
form the history $\history(\lambda_1)$. Then $\llfloor \lambda''_2 \rrfloor_{\Gamma} =
\lambda'_2$ and $\llceil \history(\lambda''_2) \rrceil_{\Gamma} = \llceil
\history(\lambda_1)\rrceil_{\Gamma}$. Since $\delta(\sigma_2) \preceq
\delta(\sigma_1)$, we have $(\sigma_2*\sigma)\fdef$.  Hence, by
Lemma~\ref{frame-in}, $\lambda''_2 \in \db{\cL_2:\Gamma'}(\sigma_2*\sigma)$, from
which the required follows.\qed


\section{Related Work\label{sec:relwork}}

The original definition of linearizability~\cite{linearizability} was proposed
in an abstract setting that did not consider a particular programming language
and implicitly assumed a complete isolation between the states of the client and
the library. Furthermore, at the time it was not clear what the linearizability
of a library entails for its clients. Filipovi\'c et al.~\cite{tcs10} were the
first to observe that linearizability implies a form of contextual refinement;
technically, their result is similar to our Lemma~\ref{thm}, but formulated over
a highly idealistic semantics. In a previous work~\cite{icalp}, we generalised
their result to a compositional proof method, formalised by an Abstraction
Theorem, that allows one to replace a concrete library by an abstract one in
reasoning about a complete program.

This paper is part of our recent push to propose notions of concurrent library
correctness for realistic programming languages. So far we have developed such
notions together with the corresponding Abstraction Theorems for supporting
reasoning about liveness properties~\cite{icalp} and weak memory
models~\cite{lintso,tso2sc-disc12,cpp-popl13}. All these results assumed that
the library and its client operate in disjoint address spaces and, hence, are
guaranteed not to interfere with each other and cannot communicate via the
heap. Lifting this restriction is the goal of the present paper. Although the
basic proof structure of Theorems~\ref{thm2} and~\ref{thm-spec} is the same as
in~\cite{icalp,lintso}, the formulations and proofs of the Abstraction Theorem
and the required lemmas here have to deal with technical challenges posed by
ownership transfer that did not arise in previous work. First, their
formulations rely on the novel forms of client-local and library-local semantics
(Section~\ref{sec:semantics}) that allow a component to communicate with its
environment via ownership transfers. Proving Lemma~\ref{prop:sem:local-global}
then involves a delicate tracking of a splitting between the parts of the state
owned by the library and the client, and how ownership transfers affect
it. Second, the key result needed to establish the Abstraction Theorem is the
Rearrangement Lemma (Lemmas~\ref{cor-rearr} and~\ref{thm}). What makes the proof
of this lemma difficult in our case is the need to deal with subtle interactions
between concurrency and ownership transfer that have not been considered in
previous work. Namely, changing the history in the lemma requires commuting
ownership transfer actions; justifying the correctness of these transformations
is non-trivial and relies on the notion of history balancedness that we propose.
These differences notwithstanding, we hope that techniques for handling
ownership transfer proposed in this paper can be combined with the ones for
handling other types of client-library interactions considered so
far~\cite{icalp,lintso,tso2sc-disc12,cpp-popl13}.

Recently, there has been a lot of work on verifying linearizability of common
algorithms; representative papers
include~\cite{daphna-linearizability,turkish10,rgsep-thesis}.  All of them
proved classical linearizability, where libraries and their clients exchange
values of a given data type and do not perform ownership transfers. This
includes even libraries, such as concurrent containers, that are actually used
by client threads to transfer the ownership of data structures. The frame rule
for linearizability we propose (Theorem~\ref{thm:frame}) justifies that
classical linearizability established for concurrent containers entails
linearizability with ownership transfer. This makes our Abstraction Theorem
applicable, enabling compositional reasoning about their clients.

Turon and Wand~\cite{aaron} have proposed a logic for establishing refinements
between concurrent modules, likely equivalent to
linearizability~\cite{tcs10}. Their logic considers libraries and clients
residing in a shared address space, but not ownership transfer.  As a result,
they do not support separate reasoning about a library and its client in
realistic situations of the kind we consider.

Elmas et al.~\cite{turkish09,turkish10} have developed a system for verifying
concurrent programs based on repeated applications of atomicity
abstraction. They do not use linearizability to perform the
abstraction. Instead, they check the commutativity of an action to be
incorporated into an atomic block with {\em all} actions of other threads.  In
particular, to abstract a library implementation in a program by its atomic
specification, their method would have to check the commutativity of every
internal action of the library with all actions executed by the client code of
other threads. Thus, the method of Elmas et al.  does not allow decomposing the
verification of a program into verifying libraries and their clients separately.
In contrast, our Abstraction Theorem ensures the atomicity of a library under
{\em any} safe client.

The most common approach of decomposing the verification of concurrent programs
is using {\em thread-modular} reasoning methods, which consider every thread in
the program in isolation under some assumptions on its
environment~\cite{Jones83,Pnueli}. However, a single thread would usually make
use of multiple program components. This work goes further by allowing a
finer-grain {\em intrathread-modular} reasoning: separating the verification of
a library and its client, the code from both of which may be executed by a
single thread. Note that this approach is complementary to thread-modular
reasoning, which can still be used to carry out the verification subtasks, such
as establishing the linearizability of libraries and proving the safety of
clients. Thread-modular techniques do enable a restricted form of
intrathread-modular reasoning, since they allow reasoning about the control of a
thread in a program while ignoring the possibility of its interruption by the
other threads. Hence, they allow considering a library method called by the
thread in isolation, e.g., by using the standard proof rules for procedures.
However, such a decomposition is done under fixed assumptions on the environment
of the thread and thus does not allow, e.g., increasing the atomicity of the
environment's actions. As the example of MCAS shows (Section~\ref{sec:example}),
this is necessary to deal with complex algorithms.


Ways of establishing relationships between different sequential implementations
of the same library have been studied in {\em data
  refinement}~\cite{data,Reynolds83}, including cases of interactions via
ownership transfer~\cite{blame,data-aplas,data-anindya}. Our results can be
viewed as generalising data refinement to the concurrent setting. Moreover, when
specialised to the sequential case, they provide a more flexible method of
performing it in the presence of the heap and ownership transfer than previously
proposed ones. In more detail, the way we define client safety
(Section~\ref{sec:semantics}) is more general some of the ways used in data
refinement~\cite{blame}. There, it is typical to fix a (precise) invariant of a
library and check that the client does not access the area of memory fenced off
by the invariant.  Here we do not require an explicit library invariant, using
the client-local semantics instead: since primitive commands fault when
accessing non-existent memory cells, the safety of the client in this semantics
ensures that it does not access the internals of the library. We note that the
approach requiring an invariant for library-local data structures does not
generalise to the concurrent setting: while a precise invariant for the data
structures {\em shared} among threads executing library code is not usually
difficult to find, the state of data structures {\em local} to the threads
depends on their program counters. Thus, an invariant insensitive to program
positions inside the library code often does not exist. Such difficulties are
one of reasons for using client- and library-local semantics in this paper.


Finally, we note that the applicability of our results is not limited to proving
existing programs correct: they can also be used in the context of formal
program development. In this case, instead of {\em abstracting} an existing
library to an atomic specification while proving a complete program, the
Abstraction Theorem allows {\em refining} an atomic library specification to a
concrete concurrent implementation while developing a program
top-down~\cite{back81,jones-tcs07}. Our work thus advances the method of
atomicity refinement to a setting with concurrent components sharing an address
space and communicating via ownership transfers.



\section*{Acknowledgements}
We would like to thank 
 Anindya Banerjee,
 Josh Berdine,
 Xinyu Feng,
 Hongjin Liang, 
 Victor Luchangco,
 David Naumann,
 Peter O'Hearn,
 Matthew Parkinson,
 Noam Rinetzky and
 Julles Villard
for helpful comments.
Gotsman was supported by the EU FET project ADVENT.
Yang was supported by EPSRC.

\appendix

\section{Additional proofs\label{sec:proofs}}

\subsection{Proof of Proposition~\ref{prop-foot}\label{sec:proof-foot}}

We prove the required by induction on the length of $S$. If $S$ is empty, then
so is $S'$ and $l_c = \delta(e)$. Assume the statement of the
proposition is valid for all histories $S$ of length less than $n>0$. Consider a
history $S = S_0 \psi$ of length $n$ and a corresponding history $S'$
satisfying the conditions of the proposition. We now make a case split on the
type of the action $\psi$.
\begin{iteMize}{$\bullet$}
\item
$\psi$ is a call transferring $\sigma_0$ that is not in $S'$. Then $S_0$ and
$S'$ are identical except $S_0$ may have extra calls. Hence, by the induction
hypothesis for $S_0$ and $S'$, $S'$ is balanced from $l_1$,
${\db{S'}^\sharp l_2 \preceq \db{S'}^\sharp l_1}$
and
$$
\db{S}^\sharp l_1 = \db{S_0\psi}^\sharp l_1 = 
(\db{S_0}^\sharp l_1) \circ \delta(\sigma_0) = (\db{S'}^\sharp l_1) \circ  (l_c \circ \delta(\sigma_0)).
$$
\item
$\psi$ is a call transferring $\sigma_0$ also present in $S'$. Then $S' = S'_0
\psi$, where $S'_0$ and $S_0$ are identical except $S_0$ may have extra calls. Hence,
by the induction hypothesis for $S_0$ and $S'_0$, we have
\begin{multline*}
\db{S}^\sharp l_1=
\db{S_0\psi}^\sharp l_1 = (\db{S_0}^\sharp l_1) \circ  \delta(\sigma_0) = {}\\
(\db{S'_0}^\sharp l_1) \circ l_c \circ  \delta(\sigma_0) = 
(\db{S'_0\psi}^\sharp l_1) \circ l_c = 
(\db{S'}^\sharp l_1) \circ l_c.
\end{multline*}
In particular, $S'$ is balanced from $l_1$. By the induction hypothesis for
$S_0$ and $S'_0$, we also have
$\db{S'_0}^\sharp l_2 \preceq \db{S'_0}^\sharp l_1$. From this we get
$$
\db{S'}^\sharp l_2=
\db{S'_0\psi}^\sharp l_2 = (\db{S'_0}^\sharp l_2) \circ \delta(\sigma_0) 
 \preceq (\db{S'_0}^\sharp l_1) \circ \delta(\sigma_0) = \db{S'_0\psi}^\sharp
 l_1 = \db{S'}^\sharp l_1.
$$
\item
$\psi$ is a return transferring $\sigma_0$. Then it is also present in $S'$,
so that $S' = S'_0\psi$, where $S_0$ and $S'_0$ are identical except $S_0$ may have
extra calls. Then by the induction hypothesis for $S_0$ and $S'_0$, we have:
$$
\db{S}^\sharp l_1=
\db{S_0\psi}^\sharp l_1 = 
(\db{S_0}^\sharp l_1) \fdiff \delta(\sigma_0) = 
((\db{S'_0}^\sharp l_1) \circ l_c) \fdiff \delta(\sigma_0).
$$
Since $S'=S'_0\psi$ is balanced from $l_2$, 
$(\db{S'_0}^\sharp l_2) \fdiff \delta(\sigma_0)$ is defined. Furthermore, 
by the induction hypothesis for $S_0$ and $S'_0$, we also have
$\db{S'_0}^\sharp l_2 \preceq \db{S'_0}^\sharp l_1$. 
Hence, $(\db{S'_0}^\sharp l_1) \fdiff \delta(\sigma_0)$ is defined as well.
By Proposition~\ref{prop-delta}, we then have:
\begin{multline*}
\db{S}^\sharp l_1
=((\db{S'_0}^\sharp l_1) \circ  l_c) \fdiff\delta(\sigma_0) = {}\\
((\db{S'_0}^\sharp l_1) \fdiff \delta(\sigma_0)) \circ l_c =
(\db{S'_0\psi}^\sharp l_1) \circ  l_c = 
(\db{S'}^\sharp l_1) \circ  l_c.
\end{multline*}
In particular, $S'$ is balanced from $l_1$. 
From $\db{S'_0}^\sharp l_2 \preceq \db{S'_0}^\sharp l_1$, it also follows that
$$
\db{S'}^\sharp l_2=
\db{S'_0\psi}^\sharp l_2 = (\db{S'_0}^\sharp l_2)\fdiff\delta(\sigma_0) 
 \preceq (\db{S'_0}^\sharp l_1)\fdiff\delta(\sigma_0) =
\db{S'_0\psi}^\sharp l_1 = \db{S'}^\sharp l_1.\eqno{\qEd}
$$
\end{iteMize}

\subsection{Proof of Lemma~\ref{prop:sem:local-global}} 

Before delving into the proof of Lemma~\ref{prop:sem:local-global}, we prove
three important lemmas about our semantics that justify its key steps. The first
concerns the evaluation of a call or a return action: intuitively, it says that
the evaluation of such an action by the client matches that by the library.
\begin{lem}[Preservation]
\label{lem:preservation}
Let $\Gamma$ be a method specification, $\sigma_0,\sigma_1$ states, and
$\varphi$ an action describing a call to or return from a method specified in
$\Gamma$ such that
$$
\db{\Gamma \vdash \varphi}\sigma_0 \neq \top
\;\wedge\;
\db{\varphi : \Gamma}\sigma_1 \neq \top.
$$
Then for all $\sigma'_0,\sigma'_1,\varphi'$,
$$
((\sigma'_0,\varphi') \in \db{\Gamma \vdash \varphi}\sigma_0 
\;\wedge\;
(\sigma'_1,\varphi') \in \db{\varphi : \Gamma}\sigma_1) 
\implies
((\sigma'_0 * \sigma'_1)\fdef \iff (\sigma_0 * \sigma_1)\fdef).
$$
If furthermore $\sigma_0 * \sigma_1$ is defined, then we have
$$
\{ \sigma_0 * \sigma_1 \} =
\{\sigma'_0 * \sigma'_1 \mid
\exists \varphi'.\;
(\sigma'_0,\varphi') \in \db{\Gamma \vdash \varphi}\sigma_0
\;\wedge\;
(\sigma'_1,\varphi') \in \db{\varphi : \Gamma}\sigma_1
\;\wedge\;
(\sigma'_0 * \sigma'_1)\fdef
\}.
$$
\end{lem}
\proof
  Consider $\Gamma,\sigma_0,\sigma_1,\varphi$ satisfying the conditions in the
  lemma. We show the lemma only for the case when $\varphi$ is a call action:
  for some $t, m, p$, we have $\varphi = (t,\call\ m)$ and $\{p\}\ m\ \{\_\} \in
  \Gamma$. The proof for the other case is symmetric.

  To show the first claim of the lemma, consider $\sigma'_0,\sigma'_1,\varphi'$
  such that
$$
(\sigma'_0,\varphi') \in \db{\Gamma \vdash \varphi}\sigma_0 
\;\wedge\;
(\sigma'_1,\varphi') \in \db{\varphi : \Gamma}\sigma_1.
$$
By the definition of the action evaluation, there exist $\sigma_2,\sigma_3$ such
that
$$
\varphi' = (t,\call\ m(\sigma_3))
\;\wedge\;
(\sigma_2 * \sigma_3)\fdef
\;\wedge\;
(\sigma_3 * \sigma_1)\fdef
\;\wedge\;
\sigma_0 = \sigma_2 * \sigma_3
\;\wedge\;
\sigma'_0 = \sigma_2
\;\wedge\;
\sigma'_1 = \sigma_3 * \sigma_1.
$$
Hence,
$$
(\sigma'_0 * \sigma'_1)\fdef
\iff
(\sigma_2 * (\sigma_3 * \sigma_1))\fdef
\iff
((\sigma_2 * \sigma_3) * \sigma_1)\fdef
\iff
(\sigma_0 * \sigma_1)\fdef.
$$

Let us move on to the second claim of the lemma. Since $\db{\Gamma \vdash
  (t,\call\ m)}\sigma_0 \neq \top$, $p_t$ is precise and the $*$ operator is
cancellative, there exists a unique splitting $\sigma_2 * \sigma_3 = \sigma_0$ of
$\sigma_0$ such that $\sigma_3 \in p_t$. Let $\varphi_0 = (t,\call\
m(\sigma_3))$.  Then
$$
(\{(\sigma_2,\varphi_0)\} = \db{\Gamma \vdash (t,\call\ m)}\sigma_0)
\;\wedge\;
(\forall \sigma'_1.\;
(\sigma'_1,\varphi_0) \in \db{(t,\call\ m) : \Gamma}\sigma_1 
\iff \sigma'_1 = \sigma_3 * \sigma_1).
$$
Hence,
\begin{align*}
&
\{ \sigma'_0 * \sigma'_1 \,\mid\,
\exists \varphi'.\;
(\sigma'_0,\varphi') \in \db{\Gamma \vdash (t,\call\ m)}\sigma_0
\wedge
(\sigma'_1,\varphi') \in \db{(t,\call\ m) : \Gamma}\sigma_1
\wedge
(\sigma'_0 * \sigma'_1)\fdef
\}
\\
&
{} =
\{ \sigma_2 * (\sigma_3 * \sigma_1) \} 
=
\{ \sigma_0 * \sigma_1 \}.\rlap{\hbox to 274 pt{\hfill\qEd}}
\end{align*}\vspace{-2 pt}

\noindent The second lemma describes the decomposition and composition properties of trace evaluation.
\begin{lem}[Trace Decomposition and Composition]
\label{lem:comp:eval}
Consider traces $\tau,\kappa,\lambda$ without interface actions such that
$\cover(\tau,\kappa,\lambda)$.  For all states $\sigma_0,\sigma_1$, if
\begin{equation}
\label{eqn:comp:eval1}
{(\sigma_0 * \sigma_1)\fdef}
\;\wedge\;
{\db{\Gamma \vdash \kappa}\sigma_0 \neq \top}
\;\wedge\;
{\db{\lambda : \Gamma}\sigma_1 \neq \top},
\end{equation}
then 
\begin{equation}
\label{eqn:comp:eval2}
\db{\tau}(\sigma_0 * \sigma_1) = 
\{(\sigma',\tau) \in \db{\Gamma \vdash \kappa}\sigma_0 \otimes \db{\lambda : \Gamma}\sigma_1\}
\end{equation}
and
\begin{multline}
\label{eqn:comp:eval2-1}
\forall \sigma_0',\sigma_1',\kappa',\lambda'.\,
(\cover(\tau,\kappa',\lambda')
\,\wedge\,
(\sigma_0',\kappa') \in \db{\Gamma \vdash \kappa}\sigma_0
\,\wedge\, 
(\sigma_1',\lambda') \in \db{\lambda : \Gamma}\sigma_1)
\\
\implies
(\sigma'_0 * \sigma_1')\fdef.
\end{multline}
\end{lem}
\begin{proof}
Consider $\tau,\kappa, \lambda, \sigma_0,\sigma_1, \Gamma$
satisfying the assumptions.
We prove the lemma by induction on the length of $\tau$. 
The base case of $\tau$ being the empty sequence is trivial.


Now suppose that $\tau = \tau'\varphi$ for some 
$\tau',\varphi$. Then there exist $\kappa'$
and $\lambda'$ such that
$$
\cover(\tau',\kappa',\lambda') 
{} \;\wedge\;
((\kappa = \kappa'\varphi \;\wedge\; \lambda = \lambda')
\;\vee\;
(\kappa = \kappa' \;\wedge\; \lambda = \lambda'\varphi)
\;\vee\;
(\kappa = \kappa'\varphi \;\wedge\; \lambda = \lambda'\varphi)).
$$
By the assumption of the lemma, we have that
\begin{equation}
\label{eqn:comp:eval3}
{\db{\Gamma \vdash \kappa'}\sigma_0 \neq \top}
\;\wedge\;
{\db{\lambda' : \Gamma}\sigma_1 \neq \top}.
\end{equation}
Hence, by the induction hypothesis, we have that
\begin{equation}
\label{eqn:comp:eval4}
\db{\tau'}(\sigma_0 * \sigma_1)
=
\{
(\sigma',\tau') \in
\db{\Gamma \vdash \kappa'}\sigma_0
\otimes
\db{\lambda' : \Gamma}\sigma_1\}
\end{equation}
and
\begin{multline}
\label{eqn:comp:eval4-1}
\forall \sigma_0',\sigma_1',\kappa'',\lambda''.\;
(\cover(\tau',\kappa'',\lambda'')
\,\wedge\,
(\sigma_0',\kappa'') \in \db{\Gamma \vdash \kappa'}\sigma_0
\,\wedge\, 
(\sigma_1',\lambda'') \in \db{\lambda' : \Gamma}\sigma_1)
\\
\implies
(\sigma'_0 * \sigma_1')\fdef.
\end{multline}

\noindent From \eqref{eqn:comp:eval4} it follows that:
\begin{equation}
\label{eqn:comp:eval5}
\db{\tau'}(\sigma_0 * \sigma_1) \neq \top.
\end{equation}
Next, we prove that
\begin{equation}
\label{eqn:comp:eval6}
\db{\tau'\varphi}(\sigma_0 * \sigma_1) \neq \top.
\end{equation}
For the sake of contradiction, suppose this disequality does not hold. 
Because of \eqref{eqn:comp:eval5}, 
there exists $\sigma''$ such that
\begin{equation}
\label{eqn:comp:eval7}
(\sigma'',\_) \in \db{\tau'}(\sigma_0 * \sigma_1) 
\;\wedge\;
\db{\varphi}\sigma'' = \top.
\end{equation}
By \eqref{eqn:comp:eval4}, this implies the existence of 
$\sigma''_0,\sigma''_1$ such that
$$
(\sigma_0'',\_) \in \db{\Gamma \vdash \kappa'}\sigma_0
\;\wedge\;
(\sigma_1'',\_) \in \db{\lambda' : \Gamma}\sigma_1
\;\wedge\;
\sigma'' = \sigma_0'' * \sigma_1''.
$$
We split cases based on the relationships among $\kappa$, $\kappa'$, $\lambda$ and $\lambda'$.
\begin{enumerate}[(1)]
\item If $\kappa = \kappa'\varphi$ and $\lambda = \lambda'$, 
then $\varphi=(t,c)$ for some $t,c$. By \eqref{eqn:comp:eval1},
$\db{\Gamma \vdash \varphi}\sigma''_0 \neq \top$, so that $f_c^t(\sigma''_0)
\neq \top$. Hence, by the Strong Locality of $f^t_c$,
$f_c^t(\sigma_0'' * \sigma_1'') \neq \top$, so that $\db{\varphi}(\sigma'')
\neq \top$. But this contradicts \eqref{eqn:comp:eval7}.

\item If $\kappa = \kappa'$ and $\lambda = \lambda'\varphi$, then $\varphi=(t,c)$ for some $t,c$.
This case is symmetric to the previous one.

\item If $\kappa = \kappa'\varphi$ and $\lambda = \lambda'\varphi$, then
  $\varphi$ is a call or a return action. Then, by the
  definition of evaluation, $\db{\varphi}\sigma'' = \{(\sigma'',\varphi)\} \neq
  \top$.  This gives the desired contradiction.
\end{enumerate}

\noindent The remainder of the proof is again done by a case analysis 
on the relationships among $\kappa$, $\kappa'$, $\lambda$ and $\lambda'$. We
consider three cases.

\medskip

1. $\kappa = \kappa'\varphi$ and $\lambda = \lambda'$. 
In this case, $\varphi = (t,c)$ for some $t,c$.
As shown in $\eqref{eqn:comp:eval6}$, $\db{\tau'\varphi}(\sigma_0 * \sigma_1) \neq \top$. 
Pick $\sigma''$ such that
\be\label{alexey1}
(\sigma'',\tau'\varphi) \in \db{\tau'\varphi}(\sigma_0  * \sigma_1).
\ee
By the definition of the trace evaluation and the induction
hypothesis in \eqref{eqn:comp:eval4}, there exist
$\sigma''_0$, $\sigma''_1$, $\kappa''$ and $\lambda''$ such that
\be \label{alexey2}
(\sigma'',\varphi) \in \db{\varphi}(\sigma''_0 * \sigma''_1)
\;\wedge\;
(\sigma''_0,\kappa'') \in \db{\Gamma \vdash \kappa'}\sigma_0
\;\wedge\;
(\sigma''_1,\lambda'') \in \db{\lambda' : \Gamma}\sigma_1
\;\wedge\;
\cover(\tau',\kappa'',\lambda'').
\ee
Then
$$
\cover(\tau'\,(t,c),\kappa''\,(t,c),\lambda'').
$$
We have $\db{\Gamma \vdash (t,c)}\sigma''_0 \neq \top$, because
$\kappa = \kappa'\,(t,c)$ and $\db{\Gamma \vdash \kappa}\sigma_0 \neq \top$.
Hence $f_c^t(\sigma''_0) \neq \top$ and, furthermore,
$\sigma'' \in f_c^t(\sigma''_0 * \sigma''_1)$.
Hence, by the Strong Locality of $f^t_c$,
there exists $\sigma'''_0$ such that
$$
\sigma'''_0 \in f_c^t(\sigma''_0)
\;\wedge\;
\sigma'' = \sigma'''_0 * \sigma''_1.
$$
This implies 
$$
(\sigma'''_0,\kappa''\,(t,c)) \in \db{\Gamma \vdash \kappa'\,(t, c)}\sigma_0 
= \db{\Gamma \vdash \kappa}\sigma_0.
$$
From what we have shown so far, it follows that
$$
(\sigma'',\tau'\,(t,c)) 
 =
(\sigma'''_0 * \sigma''_1, \tau'\,(t,c))
 \in
\db{\Gamma \vdash \kappa}\sigma_0
\otimes
\db{\lambda : \Gamma}\sigma_1.
$$
Thus,
\be\label{alexey3}
\db{\tau}(\sigma_0 * \sigma_1)
\subseteq
\{
(\sigma'',\tau) \in
\db{\Gamma \vdash \kappa}\sigma_0
\otimes
\db{\lambda : \Gamma}\sigma_1\}.
\ee

To show the other inclusion and the \eqref{eqn:comp:eval2-1} part of
the lemma, consider $\sigma''_0,\sigma''_1,\kappa'',\lambda''$ such that
$$
(\sigma''_0,\kappa'') \in \db{\Gamma \vdash \kappa'\,(t,c)}\sigma_0
\;\wedge\;
(\sigma''_1,\lambda'') \in \db{\lambda' : \Gamma}\sigma_1
\;\wedge\;
\cover(\tau'\,(t,c),\kappa'',\lambda'').
$$
Then for some $\kappa'''$ we have
$$
\kappa'' = \kappa'''\,(t,c)
\;\wedge\;
\cover(\tau',\kappa''',\lambda'').
$$
By the definition of the evaluation function, there exists $\sigma'''_0$ such that
$$
(\sigma'''_0,\kappa''') \in \db{\Gamma \vdash \kappa'}\sigma_0
\;\wedge\;
(\sigma''_0,(t,c)) \in \db{\Gamma \vdash (t,c)}\sigma'''_0.
$$
By the induction hypothesis in \eqref{eqn:comp:eval4} and \eqref{eqn:comp:eval4-1},
$$
(\sigma'''_0 * \sigma''_1)\fdef
\;\wedge\;
(\sigma'''_0 * \sigma''_1, \tau') \in \db{\tau'}(\sigma_0 * \sigma_1).
$$
Now by the Footprint Preservation property of $f^t_c$, the first conjunct above implies that 
$(\sigma''_0 * \sigma''_1)\fdef$, which proves 
$\eqref{eqn:comp:eval2-1}$. By the Strong Locality of $f^t_c$,
$$
(\sigma''_0 * \sigma''_1,(t,c)) 
\in 
\db{(t,c)}(\sigma'''_0 * \sigma''_1).
$$
From what we have shown above it follows that 
$$
(\sigma''_0 * \sigma''_1, \tau'\,(t,c)) 
\in \db{\tau}(\sigma_0 * \sigma_1).
$$
Hence, 
\be\label{alexey4}
\db{\tau}(\sigma_0 * \sigma_1)
\supseteq
\{
(\sigma'',\tau) \in
\db{\Gamma \vdash \kappa}\sigma_0
\otimes
\db{\lambda : \Gamma}\sigma_1\}.
\ee

\medskip

2. $\kappa = \kappa'$ and $\lambda = \lambda'\varphi$.  This case is symmetric
to the previous one.

\medskip

3. $\kappa = \kappa' \varphi$ and $\lambda = \lambda'\varphi$.
In this case $\varphi$ is a call to or a return from a method
in $\Gamma$. As shown in \eqref{eqn:comp:eval6}, 
$\db{\tau'\varphi}(\sigma_0 * \sigma_1) \neq \top$.
Pick $\sigma''$ such that~(\ref{alexey1}) holds.
By the definition of evaluation and the induction
hypothesis in \eqref{eqn:comp:eval4}, there exist
$\sigma''_0$, $\sigma''_1$, $\kappa''$ and $\lambda''$ such that~(\ref{alexey2})
holds. But 
$$
\db{\Gamma \vdash \varphi}\sigma''_0 \neq \top
\;\wedge\;
\db{\varphi : \Gamma}\sigma''_1 \neq \top.
$$
Furthermore, $\sigma'' = \sigma''_0 * \sigma''_1$. 
Hence, by Lemma~\ref{lem:preservation} and the definition of the evaluation, 
there exist $\sigma'''_0,\sigma'''_1,\varphi'$ such that
$$
\sigma'' = \sigma'''_0 * \sigma'''_1 \;\wedge\;
(\sigma'''_0,\varphi') \in \db{\Gamma \vdash \varphi}\sigma''_0 \;\wedge\;
(\sigma'''_1,\varphi') \in \db{\varphi : \Gamma}\sigma''_1 
$$
This in turn implies that
$$
(\sigma'''_0,\kappa''\varphi') \in \db{\Gamma \vdash \kappa'\varphi}\sigma_0
\;\wedge\;
(\sigma'''_1,\lambda''\varphi') \in \db{\lambda'\varphi : \Gamma}\sigma_1.
$$
From what we have shown so far, it follows that
$$
(\sigma'',\tau'\varphi) 
=
(\sigma'''_0 * \sigma'''_1, \tau'\varphi)
\in
\db{\Gamma \vdash \kappa}\sigma_0
\otimes
\db{\lambda : \Gamma}\sigma_1.
$$
Thus,~(\ref{alexey3}) holds.

To show the other inclusion and the \eqref{eqn:comp:eval2-1} part of the lemma,
consider
$$
\sigma''_0,\sigma''_1,\sigma'''_0,\sigma'''_1,\kappa'',\lambda'',\varphi'
$$
such that
\begin{align*}
&
(\sigma''_0,\kappa'') \in \db{\Gamma \vdash \kappa'}\sigma_0
\;\wedge\;
(\sigma''_1,\lambda'') \in \db{\lambda': \Gamma}\sigma_1
\;\wedge\;
(\sigma'''_0,\varphi') \in \db{\Gamma \vdash \varphi }\sigma''_0
\\
&
{} \;\wedge\;
(\sigma'''_1,\varphi') \in \db{\varphi : \Gamma}\sigma''_1
\;\wedge\;
\cover(\tau'\varphi, \kappa''\varphi', \lambda''\varphi').
\end{align*}
We need to show that 
$$
(\sigma'''_0 * \sigma'''_1)\fdef \;\wedge\;(\sigma'''_0*\sigma'''_1,\tau'\varphi) \in \db{\tau'\varphi}(\sigma_0*\sigma_1).
$$
In particular, this establishes~(\ref{alexey4}).

Since $\cover(\tau'\varphi, \kappa''\varphi', \lambda''\varphi')$ and
$\varphi'$ is a call to or a return from a method in $\Gamma$, 
$$
\cover(\tau', \kappa'', \lambda'') \;\wedge\;\varphi = \erase(\varphi').
$$
We now use the induction hypothesis in \eqref{eqn:comp:eval4} and \eqref{eqn:comp:eval4-1}
and derive that
$$
(\sigma''_0 * \sigma''_1)\fdef \;\wedge\;(\sigma''_0*\sigma''_1,\tau') \in \db{\tau'}(\sigma_0*\sigma_1).
$$
But $\db{\Gamma \vdash \varphi }\sigma''_0 \neq \top$ and $\db{\varphi : \Gamma}\sigma''_1 \neq \top$.
Hence, by Lemma~\ref{lem:preservation},
$$
(\sigma'''_0 * \sigma'''_1)\fdef \;\wedge\;(\sigma''_0 * \sigma''_1 = \sigma'''_0 * \sigma'''_1).
$$
By the definition of evaluation,
$$
\db{\varphi}(\sigma''_0*\sigma''_1) = \{(\sigma''_0*\sigma''_1,\varphi)\} = \{(\sigma'''_0*\sigma'''_1,\varphi)\}. 
$$
From what we have shown, it follows that 
$$
(\sigma'''_0*\sigma'''_1,\tau'\varphi) \in \db{\tau'\varphi}(\sigma_0*\sigma_1),
$$
as required.
\end{proof}

The following lemma shows that the trace-set generation of our semantics also
satisfies the decomposition and composition properties.
\begin{lem}
\label{lem:comp:tracesemantics}
$
\forall \tau.\, \tau \in \dbp{\Cc(\cL)}
\iff
(\exists \kappa,\lambda.\;
\kappa \in \dbp{\Gamma \vdash \Cc}
\;\wedge\;
\lambda \in \dbp{\cL : \Gamma}
\;\wedge\;
\cover(\tau,\kappa,\lambda)).
$
\end{lem}
\begin{proof}
Let
$$
\Cc  = {\sf let}\ [-]\ {\sf in}\ C_1\parallel{\ldots}\parallel C_n;
\ \,
\cL  = \{m = C_m \mid m \in \{m_1,{\ldots},m_j\}\};
\ \, 
C_\cmgc  = (m_1 + {\ldots} + m_j)^*.
$$

First, consider $\tau \in \dbp{\Cc(\cL)}$. By the definition of the semantics,
for some trace $\tau'$, $\tau$ is a prefix of $\tau'$,
$$
\forall t \in \{1,\ldots,n\}.\,
\tau'|_t \in \dbp{C_t}_t(\mylambda (m,t).\,\dbp{C_m}_t(\_))
$$
and all actions in $\tau'$ are done by some thread $t \in \{1,\ldots,n\}$.
Then
$$
\forall t \in \{1,\ldots,n\}.\,
\client(\tau'|_t) \in \dbp{C_t}_t(\mylambda (m,t).\,\{\varepsilon\}) \;\wedge\;
\lib(\tau'|_t) \in \dbp{C_\cmgc}_t(\mylambda (m,t).\,\dbp{C_m}_t(\_)).
$$
Since all actions in $\tau'$ are done by some thread $t \in \{1,\ldots,n\}$, we have
$$
\client(\tau') \in 
(\client(\tau'|_1) \parallel \ldots \parallel \client(\tau'|_n))
\;\wedge\;
\lib(\tau') \in 
(\lib(\tau'|_1) \parallel \ldots \parallel \lib(\tau'|_n)).
$$
Hence,
$$
\client(\tau') \in \dbp{\Gamma \vdash \Cc} \;\wedge\; \lib(\tau') \in \dbp{\cL : \Gamma}.
$$
Since $\client(\tau)$ is a prefix of $\client(\tau')$ and $\lib(\tau)$ is a
prefix of $\lib(\tau')$, this implies
$$
\client(\tau) \in \dbp{\Gamma \vdash \Cc} \;\wedge\; \lib(\tau) \in \dbp{\cL : \Gamma}.
$$
Furthermore, $\cover(\tau,\client(\tau),\lib(\tau))$, as desired.

Assume now that 
$$
\kappa \in \dbp{\Gamma \vdash \Cc} 
\;\wedge\;
\lambda \in \dbp{\cL : \Gamma}
\;\wedge\;
\cover(\tau,\kappa,\lambda).
$$
Then for some traces $\kappa'$ and $\lambda'$, $\kappa$ is a prefix of $\kappa'$, 
$\lambda$ is a prefix of $\lambda'$, and 
$$
\forall t \in \{1,\ldots,n\}.\;
\kappa'|_t \in \dbp{C_t}_t(\mylambda (m,t).\,\{\varepsilon\})
\;\wedge\;
\lambda'|_t \in \dbp{C_\cmgc}_t(\mylambda (m,t).\,\dbp{C_m}_t(\_)).
$$
The definition of our semantics in Figure~\ref{fig:trace} allows us to choose
$\lambda'$ in such a way that for some trace $\tau'$, $\tau$ is a prefix of
$\tau'$ and $\cover(\tau', \kappa', \lambda')$. Then 
$$
\forall t \in \{1,\ldots,n\}.\, \tau'|_t \in \dbp{C_t}_t(\mylambda (m,t).\,\dbp{C_m}_t(\_))
$$
and $\tau' \in (\tau'|_1 \parallel \ldots \parallel \tau'|_n)$.  Thus,
$\tau' \in \dbp{\Cc(\cL)}$, which implies $\tau \in \dbp{\Cc(\cL)}$, as desired.
%
\end{proof}

\paragraph{\em Proof of Lemma~\ref{prop:sem:local-global}.}
We first show that $\Cc(\cL)$ is safe for $I_0*I_1$. Pick states $\sigma_0,\sigma_1,\tau$
such that
$$
\sigma_0 \in I_0 
\;\wedge\;
\sigma_1 \in I_1
\;\wedge\;
(\sigma_0 * \sigma_1)\fdef
\;\wedge\;
\tau \in \dbp{\Cc(\cL)}.
$$
By Lemma~\ref{lem:comp:tracesemantics}, there exist traces $\kappa,\lambda$ such that
\be\label{alexey10}
\kappa \in \dbp{\Gamma \vdash \Cc}
\;\wedge\;
\lambda \in \dbp{\cL : \Gamma}
\;\wedge\;
\cover(\tau,\kappa,\lambda).
\ee
By our assumptions,
$\Gamma \vdash \Cc$ and $\cL : \Gamma$ are safe for $I_0$ and $I_1$, respectively.
Hence, 
$\db{\Gamma \vdash \kappa}\sigma_0 \neq \top$ and 
$\db{\lambda : \Gamma}\sigma_1 \neq \top$.
By Lemma~\ref{lem:comp:eval}, these disequalities imply that
$\db{\tau}(\sigma_0*\sigma_1) \neq \top$.
We have just shown the safety of $\Cc(\cL)$ for $I_0 * I_1$.

Next, we show that 
$$
\db{\Cc(\cL),I_0*I_1}
\subseteq
\db{\Gamma \vdash \Cc, I_0}
\otimes
\db{\cL : \Gamma, I_1}.
$$
Pick $(\sigma,\tau) \in \db{\Cc(\cL),I_0*I_1}$. Then
for some $\sigma_0,\sigma_1$ we have
$$
\sigma_0 \in I_0
\;\wedge\;
\sigma_1 \in I_1
\;\wedge\;
\sigma = \sigma_0*\sigma_1
\;\wedge\;
\tau \in \dbp{\Cc(\cL)}
\;\wedge\;
(\_,\tau') \in \db{\tau}(\sigma_0*\sigma_1).
$$
By Lemma~\ref{lem:comp:tracesemantics}, there are $\kappa,\lambda$ such
that~(\ref{alexey10}) holds.
We use Lemma~\ref{lem:comp:eval} and deduce that for
some $\kappa',\lambda'$ we have
$$
(\_,\kappa') \in \db{\Gamma \vdash \kappa}\sigma_0
\;\wedge\;
(\_,\lambda') \in \db{\lambda : \Gamma}\sigma_1
\;\wedge\;
\cover(\tau,\kappa', \lambda').
$$
Furthermore, $\sigma_0 \in I_0$, $\sigma_1 \in I_1$ and $(\sigma_0 * \sigma_1)\fdef$.
Hence,
$$
(\sigma,\tau) = (\sigma_0*\sigma_1,\tau) \in \db{\Gamma \vdash \Cc, I_0} \otimes \db{\cL : \Gamma, I_1},
$$
as desired.

Finally, we prove that
$$
\db{\Cc(\cL), I_0*I_1}
\supseteq
\db{\Gamma \vdash \Cc, I_0}
\otimes
\db{\cL : \Gamma, I_1}.
$$
Pick $(\sigma,\tau) \in \db{\Gamma \vdash \Cc,I_0} \otimes \db{\cL : \Gamma,I_1}$.
By the definition of the $\otimes$ operator and our semantics, there exist 
$\sigma_0,\sigma_1,\sigma_0',\sigma_1',\kappa,\lambda,\kappa',\lambda'$
such that 
\begin{align*}
&
\sigma = \sigma_0 * \sigma_1
\;\wedge\;
\sigma_0 \in I_0
\;\wedge\;
\sigma_1 \in I_1
\;\wedge\;
\kappa \in \dbp{\Gamma \vdash \Cc}
\;\wedge\;
\lambda \in \dbp{\cL : \Gamma}
\\
&
{} \;\wedge\;
(\sigma_0',\kappa') \in \db{\Gamma \vdash \kappa}\sigma_0
\;\wedge\;
(\sigma_1',\lambda') \in \db{\lambda : \Gamma}\sigma_1
\;\wedge\;
\cover(\tau,\kappa',\lambda').
\end{align*}
By the definition of our semantics, $\kappa = \erase(\kappa')$ and $\lambda =
\erase(\lambda')$. Because of this and $\cover(\tau,\kappa',\lambda')$
we have $\cover(\tau,\kappa,\lambda)$. By
Lemma~\ref{lem:comp:tracesemantics}, this implies $\tau \in
\dbp{\Cc(\cL)}$. Also, by Lemma~\ref{lem:comp:eval}, we have that
$$
(\sigma'_0*\sigma_1')\fdef \;\wedge\; (\sigma'_0*\sigma'_1,\tau) \in  \db{\tau}(\sigma_0*\sigma_1).
$$
Hence,
$$
(\sigma,\tau) = (\sigma_0*\sigma_1, \tau) \in \db{\Cc(\cL),I_0*I_1},
$$
as desired.
\qed

\subsection{Proof of Lemma~\ref{frame-out}\label{proof:frame-out}}

Consider $\lambda \in \db{\cL : \Gamma'}(\sigma_0*\sigma'_0)$. Then there exist
$\sigma_1$ and $\zeta \in \dbp{\cL}$ such that $(\sigma_1, \lambda) \in \db{\zeta :
  \Gamma'}(\sigma_0*\sigma'_0)$. We show that for some $\sigma_2$ we have
$$
((\sigma_2, \llfloor \lambda \rrfloor_{\Gamma}) \in \db{\zeta :
  \Gamma}\sigma_0)\ \wedge\ 
(\sigma_1 = \sigma_2 *(\dba{\llceil \history(\lambda) \rrceil_{\Gamma}}\sigma'_0)).
$$

We proceed by induction on the length of $\zeta$. The base case of $\zeta =
\varepsilon$ is trivial. Assume that the above holds for some $\lambda, \zeta,
\sigma_1, \sigma_2$ and consider $\varphi, \varphi', \sigma'_1$ such that
$$
\zeta\varphi \in \dbp{\cL}\ \wedge\ 
(\sigma'_1, \varphi') \in \db{\varphi : \Gamma'}\sigma_1\ \wedge\ 
\dba{\llceil \history(\lambda\varphi') \rrceil_{\Gamma}}\sigma'_0 \not= \top.
$$
We show that for some $\sigma'_2$ we have
$$
((\sigma'_2, \llfloor \lambda \varphi' \rrfloor_{\Gamma}) \in \db{\zeta\varphi :
  \Gamma}\sigma_0)\ \wedge\ 
(\sigma'_1 = \sigma'_2 * (\dba{\llceil \history(\lambda\varphi') \rrceil_{\Gamma}}\sigma'_0)).
$$
We consider three cases, depending on the type of the actions $\varphi$ and $\varphi'$.
\begin{iteMize}{$\bullet$}
\item
$\varphi = \varphi' = (t,c)$.
Then
$\history(\lambda\varphi')= \history(\lambda)$.
Since $\cL:\Gamma$ is safe at $\sigma_0$, $f_c^t(\sigma_2)\not=\top$.
Hence, by the Strong Locality property, we have
\begin{multline*}
\sigma'_1 \in f_c^t(\sigma_1) = f_c^t(\sigma_2*(\dba{\llceil \history(\lambda)
  \rrceil_{\Gamma}}\sigma'_0)) = {}\\ 
f_c^t(\sigma_2)*\{\dba{\llceil \history(\lambda) \rrceil_{\Gamma}}\sigma'_0\} = 
f_c^t(\sigma_2)*\{\dba{\llceil \history(\lambda\varphi')\rrceil_{\Gamma}}\sigma'_0\}.
\end{multline*}
Then $\sigma'_1 = \sigma'_2*(\dba{\llceil \history(\lambda\varphi')
  \rrceil_{\Gamma}}\sigma'_0)$ for some $\sigma'_2\in f_c^t(\sigma_2)$.
\item
$\varphi = (t, \call\ m)$ and 
$\varphi' = (t, \call\ m(\sigma_p*\sigma'_p))$, where 
$\sigma_p\in p^m_t$ and $\sigma_p*\sigma'_p\in r^m_t$.
Then 
$$
\sigma'_1 = \sigma_1*\sigma_p*\sigma'_p =
(\sigma_2*\sigma_p)*((\dba{\llceil \history(\lambda) \rrceil_{\Gamma}}\sigma'_0) *\sigma'_p)
= 
(\sigma_2*\sigma_p)*
(\dba{\llceil \history(\lambda\varphi') \rrceil_{\Gamma}}\sigma'_0).
$$
Hence, the required holds for $\sigma'_2 = \sigma_2*\sigma_p$.
\item 
$\varphi = (t, \ret\ m)$ and
$\varphi' = (t, \ret\ m(\sigma_q*\sigma'_q))$, where $\sigma_q\in q^m_t$
  and $\sigma_q*\sigma'_q\in s^m_t$.  Then 
$$
\sigma_2*(\dba{\llceil \history(\lambda) \rrceil_{\Gamma}}\sigma'_0)=\sigma_1 =
  \sigma'_1*\sigma_q*\sigma'_q.
$$
Since $\cL:\Gamma$ is safe at $\sigma_0$, $\sigma_2 = \sigma'_2*\sigma_q$ for
some $\sigma'_2$, so that
$$
\sigma'_2*\sigma_q*(\dba{\llceil \history(\lambda) \rrceil_{\Gamma}}\sigma'_0)
=\sigma'_1*\sigma_q*\sigma'_q.
$$
By the cancellativity of $*$, this entails 
$$
\sigma'_2*(\dba{\llceil \history(\lambda)
  \rrceil_{\Gamma}}\sigma'_0)=\sigma'_1*\sigma'_q.
$$
We also know that 
$\dba{\llceil \history(\lambda\varphi')\rrceil_{\Gamma}}\sigma'_0 \not=\top$, 
so that
$$
\dba{\llceil \history(\lambda\varphi')\rrceil_{\Gamma}}\sigma'_0= 
(\dba{\llceil \history(\lambda)\rrceil_{\Gamma}}\sigma'_0) \diff \sigma'_q
$$
is defined. Hence,
$\sigma'_1=\sigma'_2*
(\dba{\llceil \history(\lambda\varphi')\rrceil_{\Gamma}}\sigma'_0)$.\qed
\end{iteMize}

\subsection{Proof of Lemma~\ref{frame-in}\label{proof:frame-in}}  

In the following, we extend the $\llfloor \cdot \rrfloor_\Gamma$ operation to
non-interface actions by assuming that it does not change them.

Consider $\lambda, \sigma_0, \sigma'_0,\sigma''_0$ satisfying the conditions of
the lemma. Then $\history(\lambda)$ is balanced from
$\delta(\sigma''_0*\sigma'_0)$ for $\delta(\sigma_0) \preceq
\delta(\sigma''_0)$, and there exist $\sigma$ and $\zeta \in \dbp{\cL}$ such
that $(\sigma, \llfloor \lambda \rrfloor_\Gamma) \in \db{\zeta :
  \Gamma}\sigma_0$.  We show that
$$
(\sigma*(\dba{\llceil \history(\lambda) \rrceil_{\Gamma}}\sigma'_0), 
\lambda) \in \db{\zeta : \Gamma'}(\sigma_0*\sigma'_0)
$$
by induction on the length of $\zeta$. The base case of $\zeta = \varepsilon$ is
trivial. Assume that the above holds for some $\lambda, \sigma_0, \sigma'_0,
\sigma''_0, \sigma, \zeta$ and consider $\varphi, \varphi', \sigma'$ such that
\begin{gather*}
(\zeta\varphi \in \dbp{\cL})\ \wedge\ 
((\sigma', \llfloor \varphi' \rrfloor_\Gamma) \in \db{\varphi : \Gamma}\sigma) \ \wedge\ 
\\
(\history(\lambda\varphi') \mbox{ is balanced from } \delta(\sigma''_0*\sigma'_0)
\mbox{ for } \delta(\sigma_0) \preceq \delta(\sigma''_0)) \ \wedge\ 
(\dba{\llceil \history(\lambda\varphi') \rrceil_{\Gamma}}\sigma'_0\not=\top).
\end{gather*}
We show that
$$
(\sigma'*(\dba{\llceil \history(\lambda\varphi')\rrceil_{\Gamma}}\sigma'_0),
\lambda\varphi') \in \db{\zeta\varphi : \Gamma'}(\sigma_0 * \sigma'_0).
$$
We consider three cases, depending on the type of the actions $\varphi$ and $\varphi'$.
\begin{iteMize}{$\bullet$}
\item $\varphi = \varphi' = (t,c)$. Then 
$\history(\lambda) = \history(\lambda\varphi')$
and $\sigma' \in f_c^t(\sigma)$.  Since 
$$
(\sigma*(\dba{\llceil \history(\lambda) \rrceil_{\Gamma}}\sigma'_0))\fdef,
$$
by the Footprint
Preservation property, $(\sigma' *(\dba{\llceil \history(\lambda)
  \rrceil_{\Gamma}}\sigma'_0))\fdef$.  Then by the Strong Locality property,
\begin{multline*}
\sigma' * (\dba{\llceil \history(\lambda\varphi') \rrceil_{\Gamma}}\sigma'_0) =
\sigma' * (\dba{\llceil \history(\lambda) \rrceil_{\Gamma}}\sigma'_0) 
\in {}\\
f_c^t(\sigma)*\{\dba{\llceil \history(\lambda) \rrceil_{\Gamma}}\sigma'_0\} =
f_c^t(\sigma*(\dba{\llceil \history(\lambda)
  \rrceil_{\Gamma}}\sigma'_0)).
\end{multline*}
\item $\varphi = (t, \call\ m)$ and $\varphi' = (t, \call\
  m(\sigma_p*\sigma'_p))$, where $\sigma_p\in p^m_t$ and $\sigma_p*\sigma'_p\in
  r^m_t$. In this case we have $\sigma' = \sigma*\sigma_p$.  Since
  $\history(\lambda\varphi')$ is balanced from $\delta(\sigma''_0*\sigma'_0)$, we
  have $(\sigma*(\dba{\llceil \history(\lambda) \rrceil_{\Gamma}}\sigma'_0) *
  \sigma_p*\sigma'_p)\fdef$. 
Then
$$
\sigma*(\dba{\llceil \history(\lambda)
  \rrceil_{\Gamma}}\sigma'_0)*\sigma_p*\sigma'_p
 = \sigma'*((\dba{\llceil \history(\lambda) \rrceil_{\Gamma}}\sigma'_0)*\sigma'_p)
= \sigma'*\dba{\llceil \history(\lambda\varphi') \rrceil_{\Gamma}}\sigma'_0.
$$
\item $\varphi = (t, \ret\ m)$ and $\varphi' = (t, \ret\
  m(\sigma_q*\sigma'_q))$, where $\sigma_q\in q^m_t$ and $\sigma_q*\sigma'_q\in
  s^m_t$.  In this case we have $\sigma' = \sigma \diff \sigma_q$. We know 
  $\dba{\llceil \history(\lambda\varphi') \rrceil_{\Gamma}}\sigma'_0 \not= \top$. Thus,
  $((\dba{\llceil \history(\lambda) \rrceil_{\Gamma}}\sigma'_0) \diff
  \sigma'_q)\fdef$. Since $\sigma*(\dba{\llceil \history(\lambda)
    \rrceil_{\Gamma}}\sigma'_0)$ is defined, so is  
\begin{multline*}
\sigma' * (\dba{\llceil \history(\lambda\varphi') \rrceil_{\Gamma}}\sigma'_0)=
\sigma' * ((\dba{\llceil \history(\lambda) \rrceil_{\Gamma}}\sigma'_0)\diff\sigma'_q)={}\\
(\sigma\diff \sigma_q)*((\dba{\llceil \history(\lambda) \rrceil_{\Gamma}}\sigma'_0)\diff\sigma'_q)=
(\sigma*(\dba{\llceil \history(\lambda) \rrceil_{\Gamma}}\sigma'_0)) \diff (\sigma_q*\sigma'_q).\ \rlap{\hbox to 10 pt{\hss\qEd}}
\end{multline*}
\end{iteMize}
\section*{Glossary}

{\small

\begin{longtable}{@{}l@{\qquad}l@{}}
\bf Symbol & \bf Meaning and section
\\[4pt]
$\Sigma$ & separation algebra, \ref{sec:sa}
\\[1pt]
$*$ & operation accompanying a separation algebra, \ref{sec:sa}
\\[1pt]
$e$ & unit of a separation algebra, \ref{sec:sa}
\\[1pt]
$\sigma, \theta$ & state, \ref{sec:sa}
\\[1pt]
$p, q, r, s$ & predicate on states, \ref{sec:sa}; parameterised predicate, \ref{sec:mspecs}
\\[1pt]
$g(x)\fdef$ & function $g$ is defined on $x$, \ref{sec:sa}
\\[1pt]
$g(x)\fundef$ & function $g$ is undefined on $x$, \ref{sec:sa}
\\[1pt]
$g[x:y]$ & the function that has the same value as $g$ everywhere,
\\
& except for $x$, where it has the value $y$, \ref{sec:sa}
\\[1pt]
$\pi$ & permission, \ref{sec:sa}
\\[1pt]
$\diff$ & subtraction: on states, \ref{sec:sa}; on a state and a predicate,
\ref{sec:mspecs} 
\\[1pt]
$l$ & footprint, \ref{sec:foot}
\\[1pt]
$\delta(\sigma)$ & footprint of a state $\sigma$, \ref{sec:foot}
\\[1pt]
$\Foot(\Sigma)$ & the set of all footprints in a separation algebra $\Sigma$,
\ref{sec:foot} 
\\[1pt]
$\circ$ & addition operation on footprints, \ref{sec:foot}
\\[1pt]
$\fdiff$ & subtraction operation on footprints, \ref{sec:foot}
\\[1pt]
$\preceq$ & ``smaller-than'' relation on footprints, \ref{sec:foot}
\\[1pt]
$\psi$ & interface action, \ref{sec:observ}
\\[1pt]
$t$ & thread identifier, \ref{sec:observ}
\\[1pt]
$m$ & library method, \ref{sec:observ}
\\[1pt]
$H, S$ & history, \ref{sec:observ}
\\[1pt]
$\varepsilon$ & empty history or trace, \ref{sec:observ}
\\[1pt]
$\tau(i)$ & the $i$-th element of $\tau$, \ref{sec:observ}
\\[1pt]
$\tau\pref_k$ & the prefix of $\tau$ of length $k$, \ref{sec:observ}
\\[1pt]
$|\tau|$ & is the length of $\tau$, \ref{sec:observ}
\\[1pt]
$\db{H}^\sharp$ & footprint tracking function, \ref{sec:observ}
\\[1pt]
$\cH$ & interface set, \ref{sec:observ}
\\[1pt]
$\sqsubseteq$ & linearizability: on histories, \ref{sec:observ}; interface sets,
\ref{sec:observ}; libraries, \ref{sec:refinement}
\\[1pt]
$\rho$ & bijection on history indices, \ref{sec:observ}
\\[1pt]
$c$ & primitive command, \ref{sec:prog}
\\[1pt]
$C$ & command, \ref{sec:prog}
\\[1pt]
$\cL$ & library, \ref{sec:prog}
\\[1pt]
$\cS$ & complete program, \ref{sec:prog}
\\[1pt]
$\Cc$ & open program with a client, \ref{sec:prog}
\\[1pt]
$\cO$ & open or complete program, \ref{sec:prog}
\\[1pt]
$\Gamma$ & method specification, \ref{sec:mspecs}
\\[1pt]
$\top$ & error state, \ref{sec:prim} 
\\[1pt]
$f_c^t$ & transformer for a primitive command $c$ and thread identifier $t$,
\ref{sec:prim} 
\\[1pt]
$\varphi$ & action, \ref{sec:traces}
\\[1pt]
$\tau$ & trace, \ref{sec:traces}
\\[1pt]
$\kappa$ & client trace, \ref{sec:traces}
\\[1pt]
$\lambda, \zeta, \alpha, \beta$ & library trace, \ref{sec:traces}
\\[1pt]
$\lib(\tau)$ & projection of $\tau$ to library actions, calls and returns, \ref{sec:traces}
\\[1pt]
$\client(\tau)$ & projection of $\tau$ to client actions, calls and returns, \ref{sec:traces}
\\[1pt]
$\history(\tau)$ & projection of $\tau$ to calls and returns, \ref{sec:traces}
\\[1pt]
$\eta$ & mapping from methods to trace sets, \ref{sec:tsets}
\\[1pt]
$\dbp{\Gamma \vdash \cO : \Gamma'}$ & trace set, \ref{sec:tsets}
\\[1pt]
$\db{\Gamma \vdash \cO : \Gamma'}$ & denotation of a program $\cO$, \ref{sec:eval}
\\[1pt]
$\db{\Gamma \vdash \tau : \Gamma'}$ & evaluation of a trace $\tau$, \ref{sec:eval}
\\[1pt]
$\db{\Gamma \vdash \varphi : \Gamma'}$ & evalutation of an action $\varphi$,
\ref{sec:eval}  
\\[1pt]
$I$ & set of initial states, \ref{sec:decomp}
\\[1pt]
$\db{\cO, I}$ & set of traces of $\cO$ run from states in $I$, \ref{sec:decomp}
\\[1pt]
$\erase(\tau)$ & erasure of state annotations from actions in $\tau$, \ref{sec:decomp}
\\[1pt]
$\interf(\cL,\cI)$ & interface set of $\cL$ run from initial states in $\cI$,
\ref{sec:refinement}
\\[1pt]
$\sim$ & equivalence of client traces, \ref{sec:refinement}
\\[1pt]
$\llfloor \lambda \rrfloor$ & selects state annotations corresponding to unextended
specifications, \ref{sec:frame} 
\\[1pt]
$\llceil \lambda \rrceil$ & selects state annotations recording extra state, \ref{sec:frame}
\\[1pt]
$\dba{H}$ & evaluation of a history $H$ recording extra state, \ref{sec:frame}
\end{longtable}
}

\end{document}